\documentclass[aps,prx,showpacs,twocolumn,reprint]{revtex4-2}

\usepackage{qcircuit}
\usepackage[dvips]{graphicx}
\usepackage{amsmath,amssymb,amsthm,mathrsfs,amsfonts,dsfont}
\usepackage{epsfig}
\usepackage{braket}
\usepackage{hyperref}
\usepackage{bm}
\usepackage{enumerate}
\usepackage{color}
\usepackage{graphicx}
\usepackage{pgf}
\usepackage{tikz}
\usepackage{enumitem}
\usepackage[capitalize]{cleveref}
\usepackage[normalem]{ulem}

\usepackage{amsthm}

\newtheorem{definition}{Definition}
\newtheorem{statement}{Statement}
\newtheorem{theorem}{Theorem}
\newtheorem{lemma}{Lemma}
\newtheorem{corollary}{Corollary}

\Crefname{theorem}{Theorem}{Theorems}

\newcommand{\tr}{\mathrm{Tr}}

\newcommand{\re}{\mathrm{Re}}
\newcommand{\var}{\mathrm{Var}}

\newcommand{\cmat}{\mathbf{C}}
\newcommand{\hvec}{\underline{h}}
\newcommand{\covar}{f}
\newcommand{\jac}{\mathbf{J}}
\newcommand{\npool}{r_p}
\newcommand{\pool}{\mathcal{P}}
\newcommand{\bilinear}[3]{\langle#1,#2\rangle_{#3}}
\newcommand{\nc}{N_c}

\newcommand{\cv}{\textsc{CoVaR} }
\newcommand{\cvns}{\textsc{CoVaR}}

\begin{document}


\title{Training Variational Quantum Circuits with CoVaR: Covariance Root Finding with Classical Shadows}

\author{Gregory Boyd}
\affiliation{Department of Materials, University of Oxford, Parks Road, Oxford OX1 3PH, United Kingdom}

\author{B\'alint Koczor}
\email{balint.koczor@materials.ox.ac.uk}
\affiliation{Department of Materials, University of Oxford, Parks Road, Oxford OX1 3PH, United Kingdom}


\begin{abstract}
Exploiting near-term quantum computers and achieving practical value is a considerable and exciting challenge. Most prominent candidates as variational algorithms typically aim to find the ground state of a Hamiltonian by minimising a single classical (energy) surface which is sampled from by a quantum computer. Here we introduce a method we call \cvns, an alternative means to exploit the power of variational circuits: We find eigenstates by finding joint roots of a polynomially growing number of properties of the quantum state as covariance functions between the Hamiltonian and an operator pool of our choice. The most remarkable feature of our \cv approach is that it allows us to fully exploit the extremely powerful classical shadow techniques, i.e., we \emph{simultaneously} estimate a very large number $>10^4-10^7$ of covariances. We randomly select covariances and estimate analytical derivatives at each iteration applying a stochastic Levenberg-Marquardt step via a large but tractable linear system of equations that we solve with a classical computer. We prove that the cost in quantum resources per iteration is comparable to a standard gradient estimation, however, we observe in numerical simulations a very significant improvement by many orders of magnitude in convergence speed. \cv is directly analogous to stochastic gradient-based optimisations of paramount importance to classical machine learning while we also offload significant but tractable work onto the classical processor. As we demonstrate numerically, the approach shares features with phase-estimation protocols that prepare eigenstates with a dominant initial fidelity contribution.
\end{abstract}

\maketitle

\section{Introduction}

Quantum computers are becoming a reality and with an accelerating pace
experiments set more and more impressive records~\cite{aruteQuantumSupremacyUsing2019,zhongPhaseProgrammableGaussianBoson2021,wuStrongQuantumComputational2021,ebadiQuantumPhasesMatter2021,gongQuantumWalksProgrammable2021a}.
Current  generations of machines are already well beyond the $50$-qubit frontier and have been
demonstrated to being
capable of significant computational advantage over the best classical supercomputers.
Despite rapid progress in improving hardware it is generally believed the fault-tolerant, error corrected
systems that are expected to emerge ultimately require significantly better and larger hardware and
may thus not be within reach in the near term. The reason is that quantum states are highly vulnerable to experimental
imperfections and correcting those errors requires highly non-trivial measures, such as encoding a single logical
qubit into potentially thousands of physical qubits.

\begin{figure*}[t]
	\begin{centering}
	        \includegraphics[width=\linewidth]{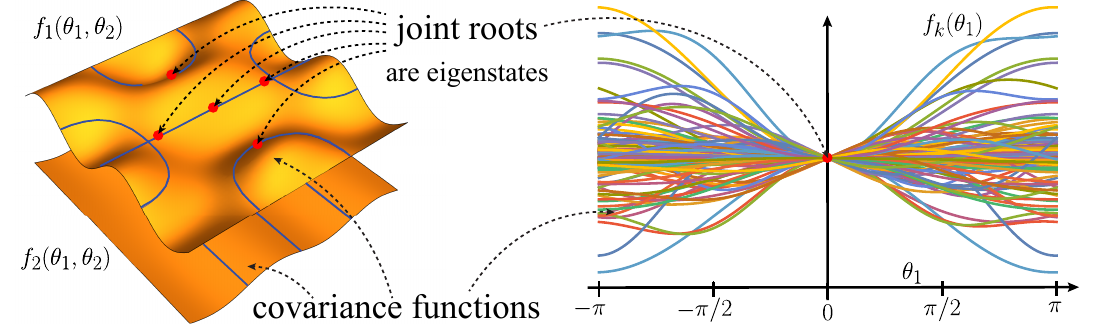}
	\end{centering}
			\caption{(left) A toy-example of a $2$-qubit problem
				whose eigenstates we aim to find by finding parameters of a variational quantum state
				$|\psi(\theta_1, \theta_2)\rangle$ prepared by two parametrised gates. Covariances 
				$\covar_k:=\bilinear{O_k}{\mathcal{H}}{\psi}$ between our problem Hamiltonian and between observables $O_k$			
				span classical surfaces (orange surfaces) and express uncertainty relations between those operators.
				Blue lines show roots as regions in parameter space where these uncertainties (covariances) vanish, i.e.,
				at roots $\underline{\theta}^\star$ the equation is satisfied
				$f_1(\underline{\theta}^\star)  = 0$. Intersections of the lines in the above surface with those
				of the below surface (red dots) guarantee eigenstates of the problem Hamiltonian as joint roots $f_1(\underline{\theta}) = f_2(\underline{\theta}) = 0$.
				(right) We use the extremely powerful classical shadow techniques to determine a very large
				number of these covariances $f_k(\theta_1)$ whose slices along the parameter $\theta_1$ are shown in a
				practically relevant variational circuit (solid lines). We initialise at $\theta_1 \neq 0 $ and iteratively
				find the joint root at  $\theta_1^\star = 0$ (red dot): We use a Levenberg-Marquardt step whereby we linearise the covariances 
				through computing a Jacobian and solve the resulting large, overdetermined linear system of equations.				
					\label{fig:plot_trdist}
				}
\end{figure*}

It is thus a very exciting challenge in the near term to achieve practical value
with these noisy intermediate-scale quantum (NISQ)~\cite{preskillQuantumComputingNISQ2018a} devices despite the damaging noise in the
hardware. The most promising candidates, generally known as variational quantum algorithms~\cite{farhi2014quantum,peruzzo2014variational,endoHybridQuantumClassicalAlgorithms2021, cerezoVariationalQuantumAlgorithms2021a, bharti2021noisy},
are robust against noise given the quantum circuit is restricted to a shallow depth.
The most prominent example is the variational quantum eigensolver (VQE) whereby
a circuit of shallow depth is constructed of parametrised quantum gates such that
the emerging quantum state is powerful enough to express the ground state of a problem of interest, e.g.,
the Hamiltonian of a chemical system. Nearly all such techniques proceed by
efficiently estimating the energy (expected value of the Hamiltonian) or an equivalent cost function
via sampling with a quantum computer and then the circuit parameters are variationally optimised
to find the solution to the desired problem. While these techniques seem promising there are many challenges, especially
in reducing high sampling costs
and performing non-linear parameter optimisations which suffer from the presence of local traps and possibly flat regions as barren plateaus~\cite{PhysRevLett.127.120502,mcclean2018barren,laroccaTheoryOverparametrizationQuantum2021,van2020measurement}.

Here we make significant progress towards addressing these challenges:
First, our approach converges faster
	than VQE in order(s) of magnitude fewer iterations and has a \emph{logarithmic} measurement cost via classical shadows when increasing our constraint size -- finding a solution as an eigenstate thus has a significantly reduced sampling cost. 
	Second, VQE optimisations have been shown to be NP hard~\cite{PhysRevLett.127.120502} due to
	local traps -- our approach is particularly robust against local traps due to a stochastic generation
	of a large number of constraints.
	Third, since in the present work we resort to local Hamiltonians as we use the NISQ-friendly variant
	of classical shadows~\cite{classical_shadows}, barren-plateaus do not necessarily exist and thus pose a less significant issue than local traps~\cite{anschuetz2022beyond,cerezoCostFunctionDependent2021a}.

In contrast to usual variational minimisation of a single cost function,
we define an entirely new class class of algorithms by leveraging the following observation:
in order to find an eigenstate a large number of properties of the variational quantum state
must satisfy certain uncertainty relations with respect to observable measurements. We define these properties as covariances~\cite{ferraro2005gaussian,carmi2018significance,tripathi2020covariance} between the problem
Hamiltonian and elements of an operator pool of our choice. This definition leaves us great flexibility in choosing our operator
pools 
and the ability to pose the problem of finding eigenstates as joint roots of covariances.
As we illustrate in \cref{fig:plot_trdist}(left) these covariances form surfaces as a function of circuit parameters
and roots of the individual covariances form submanifolds (blue lines in \cref{fig:plot_trdist}(left)).
Intersections of these as joint roots (red dots) then correspond to eigenstates of the problem Hamiltonian. In our CoVariance Root finding (\cvns) approach
we randomly select a large number of such covariances as illustrated in \cref{fig:plot_trdist}(right)
and apply powerful classical numerical techniques:
We linearise the surfaces by computing their analytical Jacobian with a quantum computer and solve
a large but tractable linear system of equations to estimate the root ($\theta_1^\star=0$ in \cref{fig:plot_trdist}(right)). We iteratively repeat this procedure until a sufficiently
good approximation of an eigenstate is found -- which we can verify classically efficiently
from our reconstructed covariances.

The most significant advantage of \cv is that we can use classical shadows to
reconstruct these covariances with an extreme efficiency: we prove that the cost of estimating
a very large Jacobian is comparable to a standard gradient estimation and grows only
logarithmically with the number of covariances. 
\cv is directly analogous to stochastic gradient-based optimisers that
have been the \emph{de-facto} standard choice for most typical variants of machine learning, e.g., 
Levenberg-Marquardt is considered to be the fastest method for training classical neural networks~\cite{neuralnet,demuth2014neural,beale2010neural,yu2018levenberg}.
As such, \cv is a quantum-classical hybrid that ideally combines the fast convergence speed of Levenberg-Marquardt
with the logarithmically efficient (quantum) computation of our large Jacobian.

We demonstrate in a comprehensive set of numerical
experiments that the efficacy of root finding is significantly increased by
employing such large datasets and our optimisation procedure is robust against local traps, circuit noise,
shot noise and noise due to random sampling of constraints. 
We cover a number of important practical applications, such as recompilation, finding ground and excited states
of local Hamiltonians, where \cv is particularly powerful as we demonstrate in numerical simulations.
Given the rapidly growing literature on variational quantum algorithms we discuss in detail connections and differences to similar approaches.

The structure of this work is the following. In the rest of this introduction section
we briefly introduce covariances and related basic notations in \cref{sec:prelim} which are both fundamental
to quantum mechanics but also form the basis of our approach. In  \cref{sec:vari_circ} we then briefly recapitulate notations related to shallow variational quantum circuits. Our main, general results are presented in \cref{sec:general} where we state conditions for finding eigenstates based on covariances and pose our
problem as root finding. In \cref{sec:covar} we introduce our CoVaR approach that uses classical shadows to find eigenstates of local Hamiltonians and relies on finding joint roots of very large systems. In \cref{sec:applications} we numerically demonstrate the power and utility of \cv in solving practical problems while we compare our technique to various other in  \cref{sec:comparison}.

\subsection{Preliminaries: operator covariances and their properties\label{sec:prelim}}
In this section we introduce all necessary tools for deriving our main results.
First, recall that a \emph{pure} quantum state is an element of the complex
Hilbert space $|\psi\rangle \in \mathbb{C}^d$ with the dimension, e.g., in a system of $N$ qubits $d=2^N$.
We will consider observables as Hermitian operators that
act on this Hilbert space as complex Hermitian matrices $A \in \mathbb{C}^{d \times d}$.
For any pair of such Hermitian operators we can define the following bilinear form
that we will refer to as a covariance.
\begin{definition}[Covariances]\label{def:covars}
Given two arbitrary Hermitian operators $A, B \in \mathbb{C}^{d\times d}$ we can define a covariance functional
between them that depends on a pure quantum state $|\psi \rangle$ via the bilinear form
\begin{equation}
    \bilinear{A}{B}{\psi} := 
    \langle \psi  | A B | \psi \rangle - \langle \psi  | A  | \psi \rangle \langle \psi  | B | \psi \rangle \in \mathbb{C}.
\end{equation}
\end{definition}
These covariances are fundamentally important in quantum mechanics and they are closely related to the
statistics when an observable property of a quantum system is measured -- covariances then express
the compatibility of these observable properties of a quantum system.

It simplifies following derivations to introduce an orthonormal set of Hermitian operators: for example, Pauli strings $O_k \in \{\mathrm{Id}_2, X, Y, Z\}^{\otimes N}$ form a complete orthonormal set with respect to the Hilbert-Schmidt scalar product $\tr[P^\dagger_k P_l]/2^N = \delta_{kl}$. Here $\delta_{kl}$ denotes the Kronecker delta and $X$, $Y$ and $Z$ are Pauli matrices. Let us now define our operator pool as a suitable set of such operators.

\begin{definition}[Operator pool]\label{def:oppool}
We define an operator pool $\pool$
as a collection of $\npool$ orthonormal Hermitian operators as
\begin{equation}
    \pool := \{O_k \}_{k=1}^{\npool},
\end{equation}
where $\tr[O^\dagger_k O_l]/2^N = \delta_{kl}$ and $O_k^\dagger = O_k$.
For example, our operator pool can be constructed of Pauli strings
as $O_k \in \{\mathrm{Id}_2, X, Y, Z \}^{\otimes N}$
where $k \in \{0,1,2,3\}^{N}$ and the overall number of terms is denoted as $\npool \leq 4^N$.
\end{definition}

Let us define the covariance matrix associated to our operator pool from \cref{def:oppool} which depends on a quantum state $| \psi \rangle$.
\begin{definition}[Complex covariance matrix]\label{def:covar_matr}
Given a collection of operators $O_k, O_l \in \pool$ from \cref{def:oppool} we define an associated Hermitian covariance matrix $\cmat(\psi) \in \mathbb{C}^{\npool \times \npool}$
that depends on a pure quantum state $| \psi \rangle$
and has the matrix entries
\begin{align} \label{eq:covar_matr}
	[\cmat^\dagger (\psi)]_{kl} =& [\cmat (\psi)]_{kl} :=
	\bilinear{O_k}{O_l}{\psi} 
	\\
	=& \langle \psi  | O_k O_l | \psi \rangle - \langle \psi  | O_k  | \psi \rangle \langle \psi  | O_l  | \psi \rangle.\nonumber
\end{align}
\end{definition}

Note that the above covariance matrix $\cmat (\psi)$ expresses fundamental
uncertainty relations between the observables $O_k$ in the operator pool via the matrix inequation $\cmat \geq 0$~\cite{tripathi2020covariance}.
We remark that our definition above of a covariance matrix has complex entries: While this definition will simplify our following arguments, it is worth noting that in the literature other conventions are also commonly used \cite{ferraro2005gaussian,carmi2018significance,tripathi2020covariance}.
For example, ref.~\cite{tripathi2020covariance} defines a covariance matrix in terms of the anticommutator as the real part
\begin{equation}\label{eq:covar_real}
   \tfrac{1}{2} \langle \psi  | \{ O_k, O_l\} | \psi \rangle - \langle \psi  | O_k  | \psi \rangle \langle \psi  | O_l  | \psi \rangle =
    \mathrm{Re}[\cmat (\psi)],
\end{equation}
while the imaginary part is often referred to as
the commutator matrix
$\tfrac{1}{2} \langle \psi  | [ O_k, O_l] | \psi \rangle  =   i \mathrm{Im}[\cmat (\psi)]$.
We will find it convenient to compactly describe the covariance matrix with complex entries thereby simultaneously referring to both the anticommutator and commutator matrices.

Given the above definitions we can straightforwardly derive a number of useful identities which we need not prove here.

\begin{corollary}\label{corollary:decompose}
Given the decompositions $A = \sum_k a_k O_k$ and $B = \sum_k b_k O_k$ of Hermitian operators $A$ and $B$ in terms of the orthonormal operator basis from \cref{def:oppool} we can obtain the covariance functional from the covariance matrix as
\begin{equation*}
    \bilinear{A}{B}{\psi} =  
    \sum_{k,l} a_k b_l [\cmat (\psi)]_{kl} 
    =
    \underline{a}^\intercal \,  \cmat(\psi) \, \underline{b},
\end{equation*}
where $\underline{a},\underline{b} \in \mathbb{R}^{\npool}$ are coefficient vectors of the operators and $\cmat(\psi)$ is the covariance matrix in this operator basis. We will later find it useful to express the special cases as the vector of covariances, or \textbf{covariance functions}  $\bilinear{O_k}{A}{\psi} = \cmat \, \underline{a}$, i.e., our primary quantities of concern will be covariances with the system Hamiltonian $\bilinear{O_k}{\mathcal{H}}{\psi}$
\end{corollary}

The variance of any Hermitian operator $A$ can also be calculated conveniently from the covariance matrix using the decomposition of $A$ into orthogonal operators.
\begin{lemma}\label{lemma:variance}
Given the decomposition $A = \sum_k a_k O_k$ of any Hermitian operator $A$ in terms of the orthonormal operator basis from \cref{def:oppool} we obtain the variance of the operator as
\begin{equation*}
\var[A] = 
    \bilinear{A}{A}{\psi} =  
    \sum_{k,l} a_k a_l [\cmat (\psi)]_{kl} 
    =
    \underline{a}^\intercal \,  \cmat \, \underline{a},
\end{equation*}
where $\underline{a}^\intercal \in \mathbb{R}^{\npool}$ is a coefficient vector and $\underline{a}\,  \cmat \, \underline{a} \geq 0$ guarantees that $\cmat$ is positive semidefinite given its Hermiticity. We will later find it useful to express the variance in terms of covariance functions
$\var[A] =   \sum_k a_k  \bilinear{O_k}{A}{\psi}$.
\end{lemma}

\subsection{Parametrisation through variational quantum circuits\label{sec:vari_circ}}
A variational quantum circuit  $U(\underline{\theta})$ usually refers to a series of 
$\nu$
parametrised quantum gates $U(\underline{\theta}) := U_\nu(\theta_\nu) \dots U_2(\theta_2) U_1(\theta_1)$ \cite{cerezoVariationalQuantumAlgorithms2021a} in some specific form, or \emph{ansatz} which often has layers of the form

\begin{equation}\label{eq:vari_gates}
	U_l(\theta_l) =\prod_m e^{-i \theta_l H_m} W_m.
\end{equation}
Where $H_m$ is the Hermitian generator of the gate parametrised by $\theta_m$ and $W_m$ can be a non-parametrised unitary associated with this layer. We will find it useful later to consider specific parametrised
gates where $H_m \in \{\mathrm{Id}_2, X, Y, Z\}^{\otimes N}$ is a Pauli string as typical in practice -- we will 
refer to these as Pauli gates. 
The number of gates in a shallow ansatz circuit is usually chosen such that the circuit depth grows slowly, such as $\mathrm{polylog}(N)$~\cite{farhi2014quantum,peruzzo2014variational,endoHybridQuantumClassicalAlgorithms2021, cerezoVariationalQuantumAlgorithms2021a, bharti2021noisy}.
Let us list here a number of well-studied ans{\"a}tze.\\[-3mm]

\textit{Hardware Efficient Ans{\"a}tze} are designed to be optimised for low circuit depth and maximal expressivity.
$U_l(\theta_l)$ are typically chosen as the native gates of the given hardware platform and the rotation angles $\theta_l$ of each
gate are treated as parameters to be optimised~\cite{kandala2017hardware}.

\textit{Unitary Coupled Cluster ansatz} (UCC) is a problem inspired ansatz for quantum chemistry. It proposes the candidate ground state using excitations of orbitals from some reference state $\ket{\psi_0}$, typically the Hartree-Fock state as $e^{T(\theta)-T(\theta)^\dagger} \ket{\psi_0}$. Here the cluster operator $T(\theta)$ \cite{taubeNewPerspectivesUnitary2006} is often restricted to single and double excitations, leading to the `UCCSD' ansatz (SD for single and double).

  \textit{Quantum Alternating Operator Ansatz, Hamiltonian Variational Ansatz and further variants}
are motivated by a time-discretised and trotterised adiabatic evolution that is guaranteed to drag the eigenstate
of a trivial Hamiltonian $\mathcal{H}_0$ to the desired problem Hamiltonian $\mathcal{H}$ for a sufficiently deep
ansatz. The evolution time $\theta_l$ of each piecewise constant, trotterised evolution $U_l(\theta_l) =  \prod_m e^{-i \theta_l H_m}$
is variationally optimised to find the ground state.

Applying this quantum circuit to an easy-to-prepare reference state of $N$ qubits defines our parametrised ansatz states as
\begin{equation*}
	| \psi(\underline{\theta}) \rangle := U(\underline{\theta}) |0\rangle^{\otimes N }.
\end{equation*}

Parameters of the ansatz circuit $\underline{\theta}$ are then varied through classical optimisation techniques
such that the quantum state $|\psi(\underline{\theta}_{opt})\rangle$ at the optimal set of parameters is a solution to our problem \cite{cerezoVariationalQuantumAlgorithms2021a}. Usually this optimisation is done by minimising a cost function, most typically the energy of a problem Hamiltonian $E(\underline{\theta})  = \langle \psi(\underline{\theta}) | \mathcal{H} | \psi(\underline{\theta}) \rangle$ \cite{peruzzo2014variational, tillyVariationalQuantumEigensolver2021a}, but note that variants of the VQE paradigm allow for the optimisation of other cost functions, such as the
variance of the Hamiltonian \cite{variance_minimisation} or non-linear functions of expected values \cite{koczor2019quantum}.

In the usual case when the gate generators $H_m$ in \cref{eq:vari_gates} are Pauli gates,
the cost function $E(\underline{\theta})$ has been shown to be a trigonometric polynomial~\cite{koczor2020quantumAnalytic}.
Finding the global minimum of $E(\underline{\theta})$ as trignomoteric functions has
been shown to be NP hard~\cite{PhysRevLett.127.120502} given the
rapidly increasing number of local minima.

Here we introduce a different paradigm; Instead of searching for the minimum of a single classical function $E(\underline{\theta})$, 
we efficiently estimate a large number of covariances that each depend on the set of parameters $\underline{\theta}$ and  thus each
corresponds to a unique surface as a function of $\underline{\theta}$ as illustrated in \cref{fig:plot_trdist}.
We prove that these parametrised covariances are indeed smooth functions of the circuit parameters $\underline{\theta}$.
\begin{lemma}[Smooth covariance functions]\label{lemma:smooth}
	Given a variational quantum state $| \psi(\underline{\theta}) \rangle := U(\underline{\theta}) |0\rangle^{\otimes N }$
	as defined via a variational quantum circuit $U(\underline{\theta}) \in \mathrm{SU}(2^N)$ we define the
	parametrised covariances as
	\begin{equation}\label{eq:param_covar_def}
		\covar_k(\underline{\theta}) := \bilinear{O_k}{\mathcal{H}}{\psi(\underline{\theta})},
		\quad \text{with}
		\quad \quad O_k \in \pool.
	\end{equation}
	The covariance functions $ \covar_k: \mathbb{R}^\nu \mapsto \mathbb{C}$
	are smooth, infinitely differentiable functions of the circuit parameters $\underline{\theta} \in \mathbb{R}^\nu$
	for any Hermitian operator $O_k$ and problem Hamiltonian $\mathcal{H}$.
	\end{lemma}
Refer to \cref{sec:proof_smooth} for a proof.
Above we have introduced the more compact notation for these covariance functions as $\covar_k(\underline{\theta})$
to highlight that we pose our problem of finding eigenstates by finding simultaneous roots of a cohort
of smooth functions $\{\covar_k(\underline{\theta}) \}_{k=1}^{\npool}$. 
Furthermore, in the practically important special case when
the ansatz circuit is composed of Pauli gates we show that the covariances $\covar_k(\underline{\theta})$ are
actually trigonometric polynomials in $\underline{\theta}$ via~ref.~\cite{koczor2020quantumAnalytic}.
\begin{corollary}[Trigonometric polynomials]\label{cor:trig_poly}
	In the specific but pivotal scenario when parametrised gates in the
	ansatz circuit in \cref{eq:vari_gates} are Pauli gates, the covariances are trigonometric polynomials as
	$\covar_k(\underline{\theta}) 
	=	\sum_{j=1}^{3^{2\nu}} C_j \mathcal{T}_j(\underline{\theta})$
	where the prefactors $C_j \in \mathbb{C}$ depend on the index $k$
	while $\mathcal{T}_j(\underline{\theta})$ are trignometric monomials, i.e.,
	products of single-variate sine and cosine functions.
\end{corollary}
As such, finding parameters such that $\covar_k(\underline{\theta}) =0$ for all $k$ is equivalent 
to finding roots of the corresponding (trigonometric) polynomial system.

\section{General results: Finding Eigenstates by Finding Roots\label{sec:general}}

This section introduces the main theoretical underpinnings of our approach in a general
setting, i.e., without making any assumptions about the problem Hamiltonian or
the type of operator pool. In contrast, in \cref{sec:covar} we will introduce \cv which is a specific, practically
motivated approach where we restrict operators to local Pauli strings which in return  allows us to
utilise the powerful classical shadow technique. We note that we will also investigate another
theoretically interesting special case of operator pools in \cref{app:orthogonal}.

\subsection{Finding eigenstates of a problem Hamiltonian}
Finding approximate representations of eigenstates of a problem Hamiltonian is a key application of near-term quantum computers. The primary hope for quantum advantage in the near term is usually placed on variational quantum algorithms whereby the solution to a problem is encoded into the ground state of a Hamiltonian. Most notable is the Variational Quantum Eigensolver \cite{peruzzo2014variational} which aims to find the ground state of a Hamiltonian via a variational minimisation of the system's energy. Finding excited states is also of particular importance for, e.g., analysing chemical reactions in drug discovery or in catalysis~\cite{reiherElucidatingReactionMechanisms2017}. 

Furthermore, applications for finding eigenstates also exist outside of quantum simulation, for example, in solving classical optimisation problems using quantum hardware via the Quantum Approximate Optimisation Algorithm (QAOA). These were introduced to solve problems such as constraint satisfaction and Max-Cut~\cite{farhi2014quantum} but have been extended beyond. Furthermore, finding eigenstates is also highly relevant to the continued design and improvement of applications and quantum algorithms. For example, recompilation problems are highly relevant as we will show. Similarly, the preparation of logical states in quantum error correction can be cast as eigenstate finding procedures~\cite{khatriQuantumassistedQuantumCompiling2019, QVECTOR}.

The aforementioned techniques typically proceed by exploiting the fact that the problem Hamiltonians of interest $\sum_{a=1}^r h_a \mathcal{H}_a $ decomposes into a polynomially growing number $r \in \mathrm{poly}(N)$ of Pauli operators whose expected values can be estimated
efficiently with a quantum computer. In the following we will denote the collection of these Pauli strings
as $\mathcal{Q} = \{\mathcal{H}_a\}_{a=1}^r$.
Let us now introduce our main result which uses covariances described in the previous section to finding eigenstates of a problem Hamiltonian.

\begin{theorem}\label{theo:main}
	Given the decomposition of a fixed problem Hamiltonian $\mathcal{H} = \sum_{a=1}^r h_a \mathcal{H}_a $
	into a set of basis operators $\mathcal{H}_a \in \mathcal{Q} \subseteq \pool$ which form a subset of our
	operator pool.
    This subset usually has a polynomial size as $r \in \mathrm{poly}(N)$.
	Given a fixed quantum state $|\psi\rangle$, simultaneous roots of all covariances
	\begin{equation}\label{eq:sufficient_cond}
	    \textbf{sufficient conds.} \quad \bilinear{\mathcal{H}_a}{\mathcal{H}}{\psi}=0, \quad \forall \mathcal{H}_a \in \mathcal{Q}
	\end{equation}
	provide a sufficient condition for the eigenvalue equation
	$\mathcal{H} | \psi \rangle = \langle \mathcal{H} \rangle  | \psi \rangle$ to hold. Further necessary conditions
	can be introduced via roots of the covariances
	\begin{equation}\label{eq:necessary_cond}
	   \textbf{necessary conds.} \quad \bilinear{O_k}{\mathcal{H}}{\psi}=0, \quad O_k \in \pool
	\end{equation}
	with respect to any basis operator in our pool $O_k$.
\end{theorem}
\begin{proof}
    \textbf{sufficient conds.}
	A direct calculation shows that the variance of the operator $\mathcal{H}$ can be expressed as
	$\var[\mathcal{H}] = \langle ( \mathcal{H} - \langle \mathcal{H} \rangle) \psi  | ( \mathcal{H} - \langle \mathcal{H} \rangle) \psi \rangle $
	and therefore the condition $\var[\mathcal{H}] =0$  immediately
	implies the eigenvalue equation $\mathcal{H} | \psi \rangle = \langle \mathcal{H} \rangle  | \psi \rangle$.
	Given our expression $\var[\mathcal{H}] =   \sum_a h_a  \bilinear{\mathcal{H}_a}{\mathcal{H}}{\psi}$ from \cref{lemma:variance} 
	simultaneous roots as $\bilinear{\mathcal{H}_a}{\mathcal{H}}{\psi} =0$ for all $a$ guarantee that 
	$\var[\mathcal{H}]=0$.\\
	\textbf{necessary conds.}
    The explicit form of the covariance as $\bilinear{A}{\mathcal{H}}{\psi} = \langle \psi  | A \mathcal{H} | \psi \rangle - \langle \psi  | A  | \psi \rangle \langle \psi  | \mathcal{H}  | \psi \rangle$ 
    simplifies when the eigenvalue equation $\mathcal{H} | \psi \rangle = \langle \mathcal{H} \rangle  | \psi \rangle$ is satisfied
    as $\bilinear{A}{\mathcal{H}}{\psi} = 0$ for Hermitian operators $A \in \mathbb{C}^{d\times d}$.
\end{proof}

	We note that the individual covariance functions may vanish without implying the presence of eigenstates of the problem Hamiltonian, for example $\bilinear{\mathcal{H}_a}{\mathcal{H}}{\psi}{=}0$ can be satisfied for a single index $a$ in the special case when
$| \psi \rangle$ is an eigenstate of $\mathcal{H}_a$.
These form a submanifold of the smooth covariances when viewed as a function of circuit parameters
as illustrated with blue lines in \cref{fig:plot_trdist}.
We therefore predicate that all covariance functions in our operator pool simultaneously vanish (red dots in \cref{fig:plot_trdist}) for all indexes $k$
which necessarily implies an
eigenstate of the problem Hamiltonian. The problem of searching for eigenstates of $\mathcal{H}$ then becomes
that of finding simultaneous roots of a system of covariances.

The above theorem ensures us that in an eigenstate all the exponentially many covariances vanish (necessary conditions), however, it is sufficient to verify only that the polynomially growing number of covariances are zero (sufficient conditions). Of course $\var[\mathcal{H}]=0$ certainly guarantees an eigenstate, however, the experimental estimation of $\var[\mathcal{H}]$ proceeds by computing expected values of individual Pauli terms and is thus informationally equivalent to estimating the above covariances~\cite{endoHybridQuantumClassicalAlgorithms2021, cerezoVariationalQuantumAlgorithms2021a, bharti2021noisy}.

While \cref{eq:sufficient_cond} lists all sufficient conditions with respect the minimal operator pool that
only contains the Pauli-decomposition terms of our problem Hamiltonian as  $\mathcal{Q}$, 
in the following we consider unions such that our operator pool is enlarged as
$\mathcal{Q}' := \{\mathcal{H}_a\}_{a=1}^r \cup \{ O_k \}_{k=1}^{\nc-r}$
 with a polynomially growing number of operators $O_k$
that are orthogonal to our problem Hamiltonian $\tr [ \mathcal{H} O_k ] = 0$ (via not including common terms $\{\mathcal{H}_a\} \cap \{ O_k \} =  \emptyset$).
Roots of all covariances with respect to our enlarged operator pool then signify an eigenstate
and we will show below that the enlarged operator pool increases the efficacy of our optimisation algorithm, i.e., by
over-constraining the Jacobian of our root finding approach. We will
refer to the size of our enlarged pool $|\mathcal{Q}'| = \nc $
as the number $\nc$ of constraints.

\subsection{Finding joint eigenstates of commuting observables}

Many problems of practical interest are concerned with finding joint eigenstates of a group of observables that all commute with each other.
For example, to prepare logical states for quantum error correction we wish to produce an eigenstate of the generators of the corresponding stabiliser group -- these generators mutually commute~\cite{nielsenChuang}. Another example is the case of recompilation of quantum circuits. Here we wish to transform a given gate sequence $V$ into a native gate sequence $U$ with an optimal circuit depth, e.g., to make it resilient to noise.

Both in the case of Full Unitary Matrix Compilation (FUMC)~\cite{khatriQuantumassistedQuantumCompiling2019} and Fixed Input State Compilation~(FISC) \cite{jonesRobustQuantumCompilation2022} the problem can be stated as applying $U^\dagger$ after $V$ onto our reference state $|\underline{0} \rangle$ (see section \ref{sec:applications_recomp} for details) which at the solution $V=U$ would correspond to the identity operation
$U^\dagger V = \mathrm{Id}$ and the resulting state $|\underline{0}\rangle$ is then the ground state of the Hamiltonian  $-\sum_{j=1}^N Z_j$. While one ultimately aims to find the ground state of this Hamiltonian, note that we can also accept any
computational basis state  $|n\rangle$ which are simultaneous eigenstates of the mutually commuting terms $\{Z_j\}$.
This motivates our approach of finding joint eigenstates of the individual Hamiltonian terms.

\begin{corollary}\label{cor:commuting_observables}
    Let us consider a set of mutually commuting Hermitian operators $\{ \mathcal{H}_a\}_{a=1}^r =: \mathcal{Q} \subseteq \pool$ 
    as our operator pool with $[\mathcal{H}_a, \mathcal{H}_b] =0$ for all $a,b$.
    Simultaneous roots of all variances
    \begin{equation*}
      \text{sufficient conds}  \quad \var[\mathcal{H}_a] =   \bilinear{\mathcal{H}_a}{\mathcal{H}_a}{\psi}
        = 0, \quad \forall a
    \end{equation*}
	provide a set of sufficient conditions such that the fixed quantum state $|\psi \rangle$ is a
	simultaneous eigenstate of all $\mathcal{H}_a$.
	We can consider further necessary constraints as the simultaneous roots of all coavariances
    \begin{equation*}
	   \text{necessary conds} \quad \bilinear{O_k}{\mathcal{H}_a}{\psi}=0, \quad O_k \in \pool, \mathcal{H}_a \in \mathcal{Q},
	   \end{equation*}
	   that need to be satisfied by $|\psi \rangle$ for any pair of operators $O_k$ and $\mathcal{H}_a $.
\end{corollary}
\begin{proof}
    For each individual index $a$ we can apply \cref{theo:main} to the corresponding operator $\mathcal{H}_a$ and
    obtain necessary and sufficient conditions such that $|\psi\rangle$ is an eigenstate of the particular operator $\mathcal{H}_a$.
    It follows that if all sufficient conditions from \cref{theo:main} are satisfied for all indexes $a$, as listed above, then
    $|\psi\rangle$ is a simultaneous eigenstate of all $\mathcal{H}_a$.
\end{proof}

\newcommand{\covarvec}{\mathbf{f}}
\subsection{Conventional techniques for finding roots}
As introduced above, our approach is based on estimating
operator covariances with a quantum computer which we
use to inform our decision of updating parameters of our variational quantum 
circuit. Our aim ultimately is to find a simultaneous root of these covariances at which 
parameters the variational state is guaranteed to be an eigenstate of our problem Hamiltonian.

There are a large number of well-established techniques for finding simultaneous roots of vector-valued functions
and almost all such techniques are in some way related to Newton's original method \cite{dennis1996numerical,press2007numerical}.
Newton's method proceeds by linearising the non-linear (but smooth) vector of covariance functions $\covarvec(\underline{\theta})$ via
the first-order Taylor expansion as
\begin{equation}\label{eq:linear_model}
\covarvec(\underline{\theta} + \Delta \underline{\theta})
=
\covarvec(\underline{\theta})
+
\jac  \Delta \underline{\theta}
+
\mathcal{O}(\lVert \Delta \underline{\theta} \rVert^2).
\end{equation}
Given each covariance function is an infinitely differentiable, smooth function of
the parameters $\underline{\theta}$ one can indeed apply Newton's method
and can approximate roots by solving the equation
$\covarvec(\underline{\theta}+\Delta \underline{\theta}) = \underline{0}$
using the above expansion $\covarvec(\underline{\theta}) = - \jac  \Delta \underline{\theta}$ and neglecting
second-order terms.
This results in a linear system of equations which can be solved using techniques from linear algebra.

\begin{figure}[tb]
	\centering
	\includegraphics[width=\linewidth]{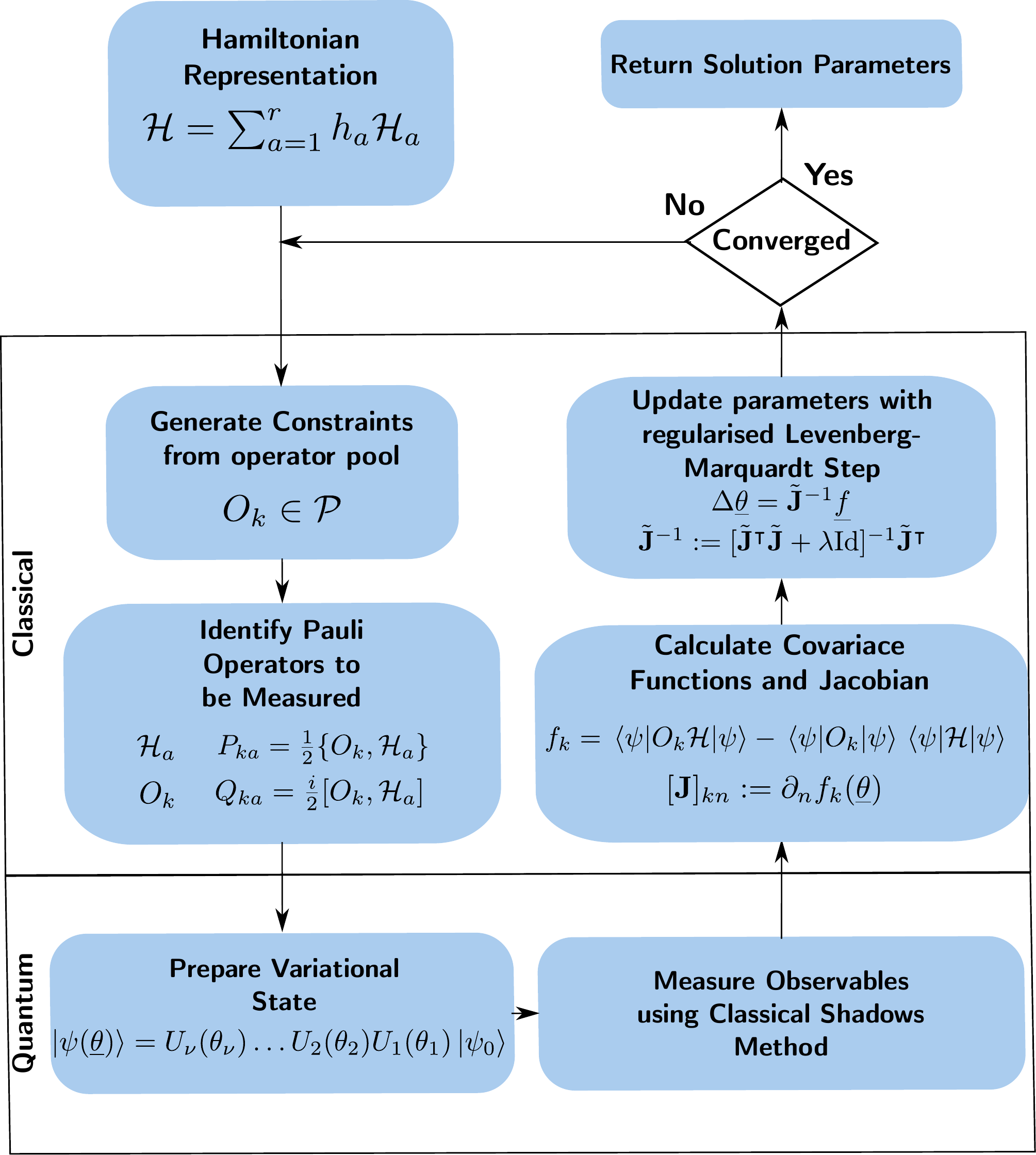}
	\caption{
		Flowchart depicting the \cv algorithm
			showing the separation between quantum and classical computations.
	}
	\label{fig:flowchart}
\end{figure}

The approach results in an iterative procedure whereby at every iteration we compute the Jacobian $\jac$ with a quantum computer
and apply its (regularised pseudo)inverse to the vector of covariances
$\covarvec$ to compute the parameter-update rule as
\begin{equation}\label{eq_jacobian}
	\underline{\theta}_{t+1} = 	\underline{\theta}_{t} - \jac^{-1} \covarvec
	\quad \quad \text{with} \quad \quad
	[\jac]_{kn} := \partial_n f_k(\underline{\theta}).
\end{equation}
We derive expressions for computing the Jacobian with a quantum computer
in \cref{app:jacobian}
using well established techniques from the literature~\cite{endoHybridQuantumClassicalAlgorithms2021, cerezoVariationalQuantumAlgorithms2021a, bharti2021noisy}, e.g.,
 parameter-shift rules.
Furthermore, we also discuss in Appendix~\ref{app_newton_method} that by stacking real parts of $\jac$ and $\covarvec$
on top of the imaginary parts results in real $\tilde{\jac}$ and $\tilde{\covarvec}$ -- enforcing that the solution of the
linear system of equations is a real vector $\Delta \underline{\theta}$.

While we aim to compute the Jacobian and the covariances with a quantum computer,
we note that the resulting linear systems of equations are solved with a classical computer.
It is important to note that we would obtain an under-determined system of equations if our operator pool were smaller
than the number of ansatz parameters as $\npool < \nu$. This is the reason why we require that 
our operator pool, and thus the dimension of the vector $\covarvec$ is at least as large as the number of circuit parameters
-- indeed, later we will aim to set up highly over-determined systems of equations.

While powerful, the vanilla Newton method has its limitations and is only guaranteed to converge
when starting near a root -- given the linear model in Eq.~\eqref{eq:linear_model} is only accurate for
small $\lVert \Delta \underline{\theta }\rVert$.
Nevertheless, a number of advanced techniques have been developed to increase the radius of convergence
and some variants of the Newton method have been proved to be globally convergent
under mild continuity conditions of the functions \cite{dennis1996numerical,okawa2018w4,pasquini1985globally}. 
In particular, the simplest globally convergent approach first attempts a conventional Newton step and if the
norm of the vector-valued function $\lVert \tilde{ \covarvec } \rVert$ does not decrease then a
line search is attempted in the step direction $\lambda \tilde{\jac}^{-1} \tilde{\covarvec}$~\cite{dennis1996numerical}
whereby one searches for $\lambda$ that minimises $\lVert \tilde{ \covarvec } \rVert$ along the 1-dimensional
search direction, see Appendix~\ref{app_newton_method} for more details. This approach is guaranteed
to converge to a root as long as the Jacobian is non-singular and well-conditioned \cite{dennis1996numerical}. 

Another family of closely related approaches are the Levenberg-Marquadt (LM) methods which are
additionally robust against ill-conditioned Jacobian matrices.
The approach can be shown to be equivalent
to the Gauss-Newton algorithm for least-squares minimisation with a trust-region method \cite{dennis1996numerical}.
It attempts steps along
what is formally a ``regularised Newton direction'' via the regularised inverse $\Delta \underline{\theta} =  [ A + \lambda \mathrm{Id}]^{-1} \tilde{\covarvec}$ with $A = \tilde{\jac}^\intercal \tilde{\jac}$ and accepts the regularisation parameter $\lambda$ based on some condition,
e.g., such that $\lVert \tilde{ \covarvec } \rVert$  decreases. The regularisation matrix is either $\mathrm{Id}$, but in
practice it is often chosen to be the diagonal matrix $\mathrm{diag}(\tilde{\jac}^\intercal \tilde{\jac})$.
In many practical applications of non-linear least-squares fitting LM can be interpreted
as an approximate Hessian optimisation, however, we detail in \cref{sec:hessian_comparison}
that this is not the case for our root finding approach.

\begin{figure}[tb]
	\centering
	\includegraphics[width=\linewidth]{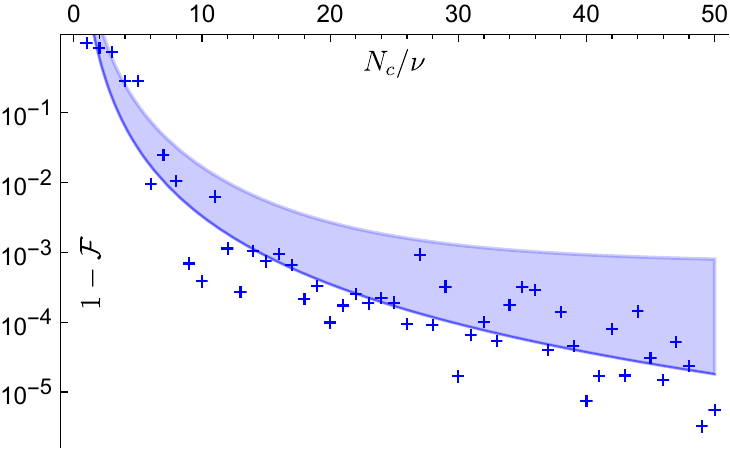}
	\caption{
		Performance improvement when we increase the number of constraints $\nc$
		illustrated on a $14$-qubit recompilation problem of `rediscovering' unknown parameters ($\nu = 124$) of an ansatz
		(refer to \cref{sec:applications_recomp} for more details).
			\cv was run for a fixed number $20$ of iterations (initial average fidelity $\mathcal{F}=46\pm7\%$)
			and the final achieved infidelity $1{-}\mathcal{F}$
			with respect to the ground state is reported:
			Blue crosses show the best of three runs of \cv and we plot their fit
			with the function $a(\nc/\nu)^{-b} +c$ (blue curve)
			as we detail in \cref{sec:ncons_comp_details}.
			The worst of three runs were also fitted with the same function (blue shade)
			confirming a polynomially (in $\nc$) increasing performance with $b \approx 3.2$
			-- although we expect this degree to depend on the problem.
			The large spread of datapoints is due to our randomly generated constraints (with respect to 3-local Pauli strings).
	}
	\label{fig:constraints_noise_comp}
\end{figure}

\section{Covariance Root Finding via classical shadows\label{sec:covar}}

In the previous sections we have described the general theory as the basis for our quantum optimisation algorithm.
We now detail concrete settings where our approach may achieve significant practical value in
exploiting near-term quantum devices.
Our aim is that the number of constraints $\nc$ in the linear system of equations is tractable but is significantly
larger than the number of circuit parameters as $\nu \ll \nc$.
For this reason we choose a $p$-local operator pool of size $\npool \in \mathcal{O}(N^p)$ 
that consists of all Pauli strings that act non-trivially on
only $p$ qubits as $\pool = \{ P_k \in \{\mathrm{Id}_2, X, Y, Z \}^{\otimes N}: \text{$P_k$ is $p$-local} \}$.

This fits very well with the use of classical shadows for determining a very large number of local Pauli strings.
In particular, the recent development of the classical shadows method \cite{classical_shadows} allows us in a NISQ-friendly way to measure $N_c$ covariances and their derivatives using a measurement count that is \emph{only logarithmic} in $N_c$. Therefore, it is possible to offload processing to the classical computer with only a small increase in the number of measurements (quantum resources) required. This combination is ideal for NISQ-era algorithms where quantum resources are limited and it is generally to our benefit if we can offload large, but tractable calculations to a classical computer. In the remainder of this section we describe the application of this method to local Hamiltonians, using a measurement channel of single-qubit Pauli gates \cite{classical_shadows}. In \cref{fig:flowchart} we provide a diagrammatic representation of the \cv algorithm.

\subsection{Stochastic optimisation with very large operator pools\label{sec:large_pool}}

Despite very promising experimental progress~\cite{aruteQuantumSupremacyUsing2019,zhongPhaseProgrammableGaussianBoson2021,wuStrongQuantumComputational2021,ebadiQuantumPhasesMatter2021,gongQuantumWalksProgrammable2021a}, near-term quantum devices  are noisy and 
in order to avoid practically prohibitive accumulation of errors
the circuit depth $D(N)$ is required to be shallow and is usually assumed to
grow poly-logarithmically as $D(N)\in \mathrm{polylog}(N)$.
The Jacobian is generally a non-square matrix with dimension $\jac \in \mathbb{C}^{ \nc \times \nu}$,
where $\nu$ is number of ansatz parameters typically scaling as $\nu =N \mathrm{polylog}(N) $ due to shallow circuit depth.  
As such, for a sufficiently large system we can always over-constrain
the Jacobian just by including covariances with respect to only two-local Pauli strings given then the
number of constraints $\nc = \mathcal{O}(N^2)$ grows faster than the number of circuit parameters.
We can thus conveniently define a very large operator pool for \cref{theo:main}
relative to the number of parameters in the ansatz circuit.

For this reason we set our operator pool $\pool$ to contain all $p$-local Pauli strings and we randomly select constraints $\underline{\covar} \in \mathbb{C}^{\nc}$ of a large size $\nc$ but much smaller than the full operator pool as $\nu \ll \nc \ll \npool$ -- but still much
larger than the number of circuit parameters.
This construction has the following advantages.
First, the large (but tractable) size of the Jacobian yields an over-constrained linear system of equations in \cref{eq_jacobian}
which significantly improves convergence speed as we demonstrate below.
Second, randomly choosing constraints has the advantage of navigating out of local traps as we numerically simulate in~\cref{sec:localmintest}.
Third, we employ stochastic Levenberg-Marquadt (LM) methods that adaptively regularise the Jacobian and
are thus by construction robust against the noise produced by random choice of constraints, as well as the hardware/shot noise on expectation values -- and rigorous proofs of convergence are available in the literature~\cite{bergouStochasticLevenbergMarquardtMethod2021, liewOptimizedSecondOrder2016}.
These are indeed properties why stochastic LM and stochastic gradient descent have been extremely popular in the
classical machine learning context, i.e., due their robustness against noise as well as their robustness against getting stuck in local traps~\cite{ruder2016overview, sweke2019stochastic}~\footnote{
	Stochastic gradient descent for VQE has been termed for instances when shot noise on estimated gradients is significant \cite{sweke2019stochastic}. In contrast, the present approach is stochastic due to the random selection of constraints
}

In \cref{fig:constraints_noise_comp} we confirm numerically on a $14$-qubit recompilation problem
that indeed the performance of root finding increases as the number of constraints $\nc/\nu$ in the linear system of equations is increased.
As we detail below in \cref{sec:applications_recomp}, this recompilation problem is a hard benchmarking task with the advantage that
our ansatz is capable of expressing the exact solution.
In \cref{fig:constraints_noise_comp} we ran \cv for a fixed number $20$ of iterations and plot how close the evolution came to the solution, i.e., 
the infidelity with respect to the ground state. We assume an idealised simulation with no shot noise or circuit noise; 
thus the only source of `error' is the linearisation of the non-linear covariances via \cref{eq:linear_model}
while the performance is significantly improved as we increase the number of constraints.

Let us attempt to intuitively explain on an analytical example why such an increasingly over-constrained system of equations
improves our ability to find the solution
Take for example the simple case when the ansatz circuit has a \emph{single parameter} $\theta$ and (as illustrated
in \cref{fig:plot_trdist} (right)) thus
the covariance function vector $\covarvec(\theta)$ can be linearised via \cref{eq:linear_model} as
\begin{equation}\label{eq:single_variable}
	\covarvec(\theta + \Delta \theta )= 
	\covarvec(\theta)
	+ \jac \Delta \theta
	+\mathcal{O}(|\Delta \theta|^2).
\end{equation}
Here the Jacobian is $J_k = \covar_k'(\theta)$ (assuming $J_k$ and $\covar_k$ are real
as we have stacked real and imaginary parts on top of each other).
For each individual
function $f_k(\theta)$, Newton's single-variate parameter update approximates the root as $\Delta_k \theta =- \covar_k(\theta) /J_k$, however,
we incur an error $\lVert \covarvec(\theta {+} \Delta \theta ) \rVert^2 = \sum_{j \neq k} [ \covar_j(\theta) - \covar_j'\Delta_k \theta ]^2+ ...$ due to the nonlinearity of $f_k(\theta)$ as we illustrate in \cref{fig:single_variable} (blue lines).
On the other hand, the least squares solution 
simultaneously takes into account all linearised constraints as $\covarvec(\theta {+} \Delta \theta ) = \underline{0}$
and is given analytically as $\Delta \theta  = (\jac^\intercal \jac)^{-1} \jac^\intercal \covarvec =   -\sum_k J_k \covar_k(\theta)/[ \sum_k J_k^2 ] $
which inherently minimises the aforementioned error via $\lVert \covarvec(\theta {+} \Delta \theta ) \rVert^2 \rightarrow min$.
Indeed, the least-squares solution (\cref{fig:single_variable} orange line)
approximates the solution much better than either the individual, single-variate Newton solutions (blue lines) or their average (black line).

\begin{figure}[tb]
	\begin{centering}
		\includegraphics[width=0.8\linewidth]{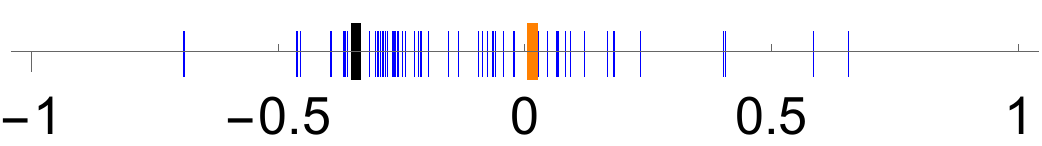}
	\end{centering}
	\caption{
		Illustrating the advantage of overdetermined systems of equations.
		The same variational circuit as in \cref{fig:constraints_noise_comp}
		with all parameters optimal except for a single variable which was disturbed as $\theta = 0.5$ and thus the exact root is at $\theta^\star = 0$. Blue lines show single-variate Newton steps computed for 100 individual covariance functions $\theta + \Delta_k \theta$ and their average is represented by the black line. With the orange line we obtain a better approximation of the root by solving a $\nc = 176$ over-determined system of linear equations (see main text).}
	\label{fig:single_variable}
\end{figure}

Let us now analyse the time complexity of classically computing the (pseudo)inverse of the Jacobian $\tilde{\jac}$.
We prove in \cref{app:classical_comp_time} that computing the least-squares solution to the linear systems of equations is dominated by
the step of computing $A := \tilde{\jac}^\intercal \tilde{\jac}$ which can be performed in time $\mathcal{O}(\nu^2 \nc)$ and, as such,
scales linearly with the number of constraints. The rest of the procedure, including the computation of the inverse  of the small square matrix
$A\in \mathbb{C}^{\nu \times \nu}$ can then be computed in negligible additive time.
Given the number of constraints grows at most as $\nc \leq \mathcal{O}(N^p)$ for our specific choice of $p$-local Pauli strings
the computation time $t$ grows at most as $t \leq \mathcal{O}[N^{p+2} \mathrm{polylog}(N)]$ with the number of qubits $N$.

We confirm these expectations in \cref{fig:classical_comp_fig} and estimate that a very large matrix with
$\nc = 10^6$ constraints for a variational circuit of $\nu = 10^3$ can straightforwardly be computed in
a matter of minutes and fits into the RAM of a typical single node -- while distributed computation for larger datasets $>10^7$ is possible with negligible communication between nodes.
We expect determining necessary expected values from classical shadows has a comparable computation time
which we detail below.

\subsection{Noise robustness}

Let us now demonstrate the aforementioned noise-robustness of our approach: as we experimentally
estimate the Jacobian and the covariances
we always incur a certain amount of shot noise (due to finite sampling) but also
possible noise due to experimental imperfections.
While we demonstrated in a noise-free environment that performance
	is improved when increasing the number of constraints, one might think that it could also lead
to an accumulation of noise. For this reason we prove in
\cref{app:shot_noise_floor} that the error in our estimate of the update rule $\Delta \theta$ in \cref{sec:large_pool} does
not accumulate as we increase the number of constraints, i.e., the error is constant bounded by
the worst-case error in a \emph{single Jacobian/covariance entry}. 

We obtain a similar conclusion for the error (shot noise) propagation in the general multi-variate case by applying the
error propagation formula of ref~\cite{van2020measurement} for matrix inversion.
In particular, the error in the update rule scales with the fourth power of the 
smallest inverse singular value (or regularisation parameter $\lambda^{-4}$) of $\jac$.
Given singular values of our $\nu \times \nc$-dimensional
Jacobian matrix grow with $\nc$, we expect \cv is particularly robust against shot noise. In our $\nu=1$-dimensional example in the
previous subsection we had a singular value
$[\sum_{k=1}^{\nc} J_k^2]^{1/2}$ of $\jac^\intercal \jac$ which indeed grows with the square root of $\nc$
for non-zero derivatives $|J_k|>0$.

In \cref{fig:constraints_noise} we repeated our simulations from \cref{fig:constraints_noise_comp}
 with added shot noise and circuit noise. In particular, \cref{fig:constraints_noise}(orange) shows our simulations with only
 shot noise added
and confirms our above analytical arguments: As we increase $\nc/\nu$ the optimisation is
able to come closer and closer to the root in a fixed number $20$ of iterations
up until a point when we reach a shot-noise floor $\mathcal{E}$ where the performance is no longer increased. This shot-noise floor is indeed below the precision of determining individual entries $N_s^{-1/2} \propto 10^{-2.5}$.

Furthermore, \cref{fig:constraints_noise}(black) shows the performance of root finding under simulated circuit noise but without shot noise.
As we detail in \cref{sec:recomp_details}, we have assumed two- and single-qubit gate error rates $\epsilon_2 =10^{-3}$ and $\epsilon_1 = \epsilon_2/4$, respectively, which is comparable to the performance of state-of-the art hardware. While the optimisation is performed with noise, the plotted infidelities are calculated without noise to reflect that the correct parameters are found as, e.g., error mitigation techniques are typically applied for extracting noise-suppressed
expected values from a final state~\cite{qemreview,koczor2020exponential,koczor2021dominant,huggins2020virtual}.
These results show a very similar performance to the case with shot noise only in \cref{fig:constraints_noise}(orange):
the performance is increased up until a point where a noise-floor is reached -- and the magnitude of this noise floor
in our example appears to be very close to the case of shot noise only.
Interestingly, our approach finds circuit parameters very close to the ideal
ones (small final infidelities) despite circuit noise -- this indeed resembles to the phenomenon of Optimal Parameter Resilience~\cite{sharma2020noise}, meaning this recompilation task is not merely learning in the applied circuit noise. 
These simulations confirm the robustness of our approach against experimental noise.

\subsection{Estimating a large number of covariances via classical shadows}

Recall that a $p$-local problem Hamiltonian can be specified in terms of its Pauli decomposition as $\mathcal{H} = \sum_{a=1}^{r} h_k P^{(p)}_a$ where the Pauli strings $P^{(q)}_a \in \mathcal{Q}^{(p)}$ are $p$-local, i.e., they only act on $p$ qubits non-trivially. Such local Hamiltonians are highly relevant in many important problems which include,
for example, recompilation, spin models in materials science, boolean satisfiability problems (3SAT) and fermionic models using mappings that retain operator locality \cite{derbyCompactFermionQubit2021,cerezo2020variationalreview,endoHybridQuantumClassicalAlgorithms2021,bharti2021noisy, pagano2020quantum,arute2020quantum}.

\begin{figure}[tb]
	\centering
	\includegraphics[width=\linewidth]{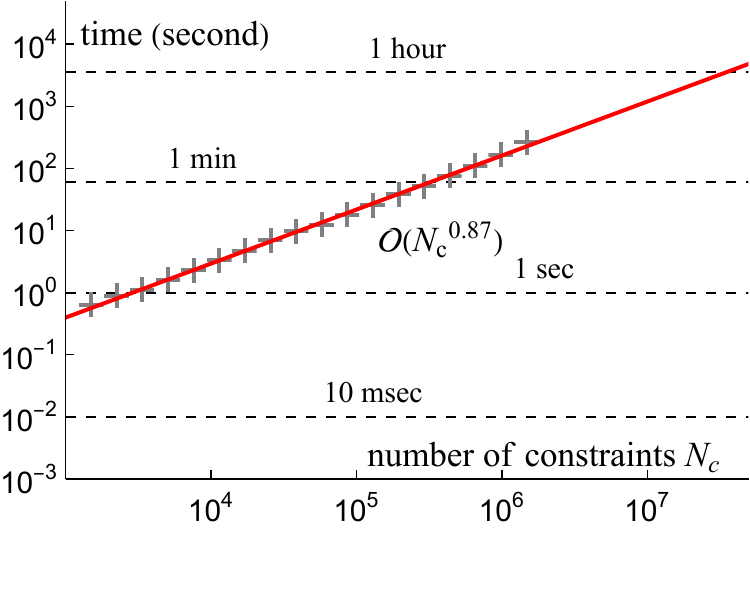}
	\caption{
		\label{fig:classical_comp_fig}
		Computation time with a Mathematica code of the 
		update rule 
		$\Delta \underline{\theta}=\tilde{\jac}^{-1}\covarvec$
		via the regularised inverse of the
		non-square matrix Jacobian $\tilde{\jac} \in \mathbb{R}^{ 2\nc \times \nu}$
		with a fixed number of circuit parameters $\nu = 1000$ and increasing number
		of constraints (covariances) $\nc$.
		The time complexity is $\mathcal{O}(\nc)$ linear in the leading dimension as the number of
		constraints and the absolute time is very reasonable, i.e., less than an hour
		even for very large matrices with $\nc = 10^7$	(Computed with a desktop PC).
	}
\end{figure}

Let us consider an operator pool $\pool^{q}$ that contains all $q$-local Pauli strings and thus has size $\npool \in \mathcal{O}(N^q)$.
We randomly choose covariances $\covar_k$ from this operator pool such that $\nc \ll \npool$,
 and as we discussed we aim to estimate a large number of covariances (constraints)
via a large but tractable $\nc$.
Let us now establish that we need only reconstruct expected values of at most $p{+}q$-local Pauli strings
in order to determine all covariances.
\begin{statement}\label{stat:locality}
Given a $p$-local problem Hamiltonian $\mathcal{H}$ we can estimate covariances $\bilinear{P^{(q)}_k}{\mathcal{H}}{\psi}$ with respect to at most $q$-local Pauli strings by reconstructing expected values of at most $p{+}q$-local Pauli strings of the form $\langle \psi | P^{(p+q)}_k  | \psi \rangle$.
\end{statement}
\begin{proof}
Let us explicitly write the covariances as
\begin{align}\nonumber
\bilinear{P^{(q)}_k}{\mathcal{H}}{\psi}
=&
\sum_{a=1}^r h_k \bilinear{P^{(q)}_k}{P^{(p)}_a}{\psi}
=
\sum_{a=1}^r h_k \langle \psi  | P^{(q)}_k P^{(p)}_a | \psi \rangle \\
-&\sum_{a=1}^r h_k \langle \psi  |  P^{(q)}_k  | \psi \rangle \langle \psi  | P^{(p)}_a | \psi \rangle.
\end{align}
Above the product of the Pauli strings $P^{(q)}_k P^{(p)}_a$
is proportional to a $p{+}q$-local Pauli string as $P^{(p+q)}_k$
up to possibly a prefactor $\pm i$ depending on whether $P^{(q)}_k$ and $ P^{(p)}_a$ commute or anticommute etc.
as discussed in \cref{app:jacobian}. As such, above we obtain a weighted sum of only expected values of Pauli strings, and thus we conclude that any covariance of the form $\bilinear{P^{(q)}_k}{\mathcal{H}}{\psi}$ can be reconstructed by estimating expected values of at most $p{+}q$-local Pauli strings.
\end{proof}

Note that determining all covariances that satisfy the sufficient conditions in \cref{eq:sufficient_cond} require that the locality
of the operator pool is at least as large as the locality of the problem Hamiltonian via $q\geq p$.

The classical shadow procedure~\cite{classical_shadows} fits very well with our
\cv approach as it allows us to estimate a very large number of
these covariances such that the number of samples (quantum resources)
increase only logarithmically with the number of constraints -- while the required measurements are very
NISQ friendly. Let us briefly recapitulate the main steps to reconstructing Pauli strings
using classical shadows.
\begin{itemize}[leftmargin=*]
	\item  We apply a random unitary $U$ to rotate the state. In our case of local Pauli strings the unitaries are chosen randomly from single qubit Clifford gates on each qubit and the procedure is thus equivalent to randomly selecting to measure in the $X, Y$ or $Z$ bases -- we measure each qubit to obtain $N$-bit measurement outcomes $\ket{\hat{b}_i} \in \{0,1\}^N$.
	\item We then generate the classical shadows by applying the inverse of the measurement channel $\mathcal{M}$, which can be done efficiently as the channel chosen is a distribution over Clifford circuits. The classical snapshots are generated as $\hat{\rho}_i = \mathcal{M}^{-1} ( U^\dagger \ket{\hat{b}_i}\bra{\hat{b}_i} U ) $, the classical shadows are collections of these snapshots $S(\rho; N) = [\hat{\rho}_1,...,\hat{\rho}_N]$.
	\item From these classical shadows we can construct $K$ estimators of $\rho$ from our $N_{batch}$ snapshots as \\ $\hat{\rho}_{(k)}=\frac{1}{r} \sum_{i=(k-1)r+1}^{k r} \hat{\rho}_i$ with $r=\lfloor N_{batch}/K \rfloor$ and classically calculate estimators of the Pauli expectation values $\hat{o}_i(N,K) = \text{median} \{ \tr(O_i \hat{\rho}_{(1)}),...,\tr(O_i \hat{\rho}_{(K)}) \}$.
	The classical computational resources are quite modest.
	\item The sample complexity of obtaining these estimators of $M$ Pauli operators of locality $l$ to error $\epsilon$ is $\mathcal{O} [3^l \log(M )/\epsilon^2]$. 
\end{itemize}

\begin{figure}[tb]
	\centering
	\includegraphics[width=\linewidth]{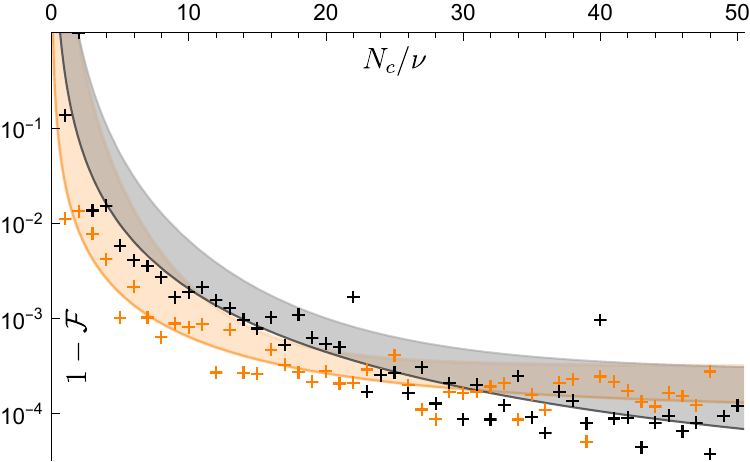}
	\caption{
		Simulations and fits are identical to those in~\cref{fig:constraints_noise_comp} but now with added noise. Orange shows results including a shot noise for $N_s=10^5$ shots, modelled as independent Gaussian noise on all expectation values. At small $N_c$ shot noise regularises the ill-conditioned linear system of equations and does seemingly improve performance.
		Black results show the performance of \cv under simulated circuit noise as described in~\cref{sec:recomp_details}, indicating the resilience of the method to reasonable levels of circuit noise.
		In the large-$\nc$ limit 
			a noise floor is approached whose magnitude in our particular example is
			comparable for both cases of only shot noise (orange) and only circuit noise (black)
			-- large spread of the datapoints is due to our randomly generated constraints and
			fits of the best (solid curves) and worst (shaded area) of 3 runs of \cv are included.
	}
	\label{fig:constraints_noise}
\end{figure}

Using classical shadows allows us to reconstruct all $(p{+}q)$-local Pauli strings with a sample complexity that is merely logarithmic in the system size. This fits particularly well with the preset approach: when the locality $p{+}q$ of Pauli strings is modest then 
we can obtain a large, polynomially growing number of constraints $\nc \in \mathcal{O}(N^{q})$.
Furthermore, given the covariances are fully determined by expected values (of local Pauli strings), we show that their analytical derivatives can be estimated using expected values at shifted circuit parameters via the so-called parameter-shift rules~\cite{paramshift}. In particular, each partial derivative in the Jacobian $[\jac]_{kn} := \partial_n f_k(\underline{\theta})$ is determined by estimating expected values at
two different shifted parameters as discussed in \cref{sec:computing_jacobian}. As such,
we can fully determine our Levenberg-Marquardt step just using expected values of local Pauli strings.

\begin{figure*}[tb]
	\centering
	\includegraphics[width=\textwidth]{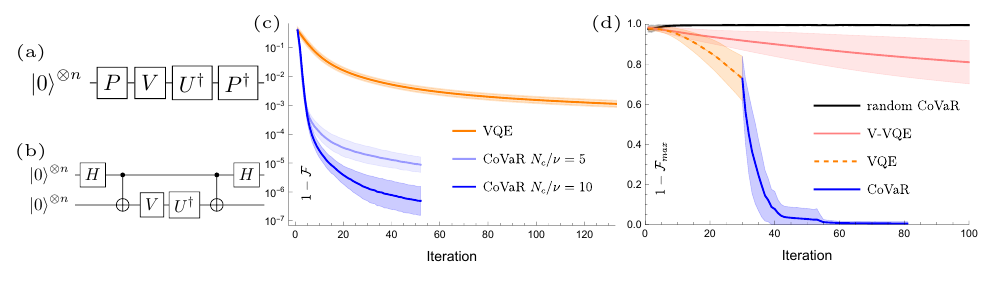}
	\caption{
		Demonstration of CoVaR in fixed input state recompilation for a $10$-qubit and 2-layer ($\nu = 88$) ansatz.
			The circuit in (a) is used with $V$ and $U$ being the same ansatz but with $V$ having solution parameters $\underline{\theta}^\star$ that are randomly chosen as  $| \theta^\star_k| \leq 2\pi$ but remain fixed over 20 runs. We rediscover these hidden parameters using the ansatz
			$U(\underline{\theta})$ where $\underline{\theta} = \underline{\theta}^\star {+} \Delta \underline{\theta}$. All curves show mean and standard deviation of infidelities over 20 runs and simulations
			include shot noise ($N_s=10^5$).
		(b) shows  the circuit for the recompilation of a full unitary. 
		(c) 
		performance of \cv (blue lines) using $N_c/\nu =5,10$ compared with gradient descent (orange). Infidelity $1{-}\mathcal{F}$ from the ground state is plotted as we initialise close to the solution via small random parameter perturbations $|\Delta \theta_k| \leq 0.3$ resulting in an
			initial fidelity $1-\mathcal{F}=41 \pm 8 \%$.
		(d) performance when initialised completely randomly in parameter space via $|\Delta \theta_k| \leq 2\pi$, showing the infidelity $1-\mathcal{F}_{max}$ to the nearest computational basis state. Black shows progress for CoVaR and red for variance-VQE -- when both are initialised randomly, both fail to make significant progress, regardless of choice of $\nc$. Additionally shown is the use of a short period of gradient descent to `initialise' (orange, dashed) and then \cv (blue, $\nc/\nu=20$), which reaches a final infidelity of $0.5\%$ on average (blue curve only includes the 16/20 runs which were able to converge).
	}
	\label{fig:gradient_comparison}
\end{figure*}

Let us now state the sample complexity of \cv whose (quantum) cost is
dominated by estimating the Jacobian and let us compare it to the
cost of determining a gradient vector used in energy minimisation.
\begin{statement}
Given a $p$-local problem Hamiltonian $\mathcal{H}$ we use classical shadows to determine a large number $\nc$ of covariances with respect to $q$-local Pauli strings.
The sample complexity of determining the Jacobian of size $\jac \in \mathbb{C}^{\nc \times \nu}$
to an error $\epsilon$ is 
\begin{equation*}
	N_s = \mathcal{O}[3^{p+q} \, \nu \log( r \nc  )/\epsilon^2 ].
\end{equation*}
In contrast, determining the gradient of the energy expected value $\langle \mathcal{H} \rangle$ using
classical shadows has a complexity $\mathcal{O}(3^p \nu \log( r )/\epsilon^2 )$. As such, determining a very large Jacobian is only logarithmically more expensive then determining an energy gradient (up to a multiplicative constant $3^{q}$ that depends on the modest locality of our choice, e.g, $q=2,3$.  
\end{statement}
\begin{proof}
Theorem~1 of ref.~\cite{classical_shadows} established that $M$ Pauli strings $O_i$ of locality $l$,
can be estimated to precision parameters $\epsilon, \delta$ via the number
of batches $K= 2\log(2M/\delta)$ and the number of samples in the individual batches as
$N_{batch}=\tfrac{34}{\epsilon^2} \max_{i} \lVert O_i \rVert_{shadow}^2$.
This results in an overall number of samples $N_{batch} K$ and the norm is given in Lemma~3 in ref.~\cite{classical_shadows}
as $\lVert O_k \rVert_{shadow}^2 = 3^l$.

We have established in \cref{lemma:jacobian_complexity} that we can determine the
Jacobian matrix $\jac \in \mathbb{C}^{\nc \times \nu }$
by applying the classical shadow procedure $2\nu {+}1 $ times at different circuit-parameter configurations
with $M \leq 3 r \nc$. 
As such, determining these Pauli strings of locality $l = p+q$
requires the number of samples $N_{batch} K \leq \tfrac{68}{\epsilon^2}3^{p+q}  \log(6 r \nc /\delta)$.
Given we apply the classical shadow procedure $2 \nu +1$ times we obtain the
following upper bound on the sample complexity
\begin{equation*}
	N_s  \leq (2 \nu {+}1) 3^{p+q} \tfrac{68}{\epsilon^2}  \log(6 r \nc /\delta)
\end{equation*}

In contrast determining a gradient vector
for gradient descent requires $2\nu$ applications of the classical shadow
procedure each with $M=r$ and thus we obtain the sample complexity
\begin{equation*}
	N_s^{(grad)} \leq 2 \nu 3^p \tfrac{68}{\epsilon^2}\log(2 r /\delta).
\end{equation*}
In both cases we have determined the necessary Pauli strings to precision $\epsilon$ and both the energy gradient and the covariance Jacobian are then obtained from these as a linear combination with respect to Hamiltonian coefficients which leads to a worst-case error propagation of $\mathcal{O}(r)$.
\end{proof}

Actually, this bound on the number of required measurements in terms of the locality is noted to be conservative and it is expected that the actual constants are \emph{much smaller} in practice \cite{classical_shadows}. Furthermore, the inclusion of the development of derandomized classical shadows \cite{shadows_derandomization} has the potential to significantly reduce the number of required measurements, i.e., an order of magnitude reduction has been demonstrated in numerical experiments \cite{classical_shadows}. These techniques could thus greatly improve the speed at which the covariances can be extracted by optimising the Pauli measurement basis to the specific Pauli strings in our operator pool -- but we do not consider these in our above performance bounds.
Furthermore, the overhead of \cv relative to determining a single gradient vector in gradient descent is the constant $3^q$ (up to the logarithmic dependence on $\nc$) and is only
due to the increase in the locality of Pauli strings with $q= 2, 3$ etc.
We can thus expect that determining a \emph{very large} Jacobian has a comparable
complexity to determining just a single gradient vector in gradient descent.
We will demonstrate in the following that this increased size of the Jacobian 
has significant advantages in practical applications.

\section{Applications\label{sec:applications}}

\subsection{Recompilation} \label{sec:applications_recomp}

The ability to recompile a given quantum circuit into an equivalent but practically
feasible or more favourable representation is crucial for the successful exploitation of quantum 
computers. The ideas exist in many variants, from the application of classically tractable analytical 
gate-replacement rules to automatic discovery techniques \cite{mooreIntroductionIntervalAnalysis2009, schuchProgrammableNetworksQuantum2003, moroQuantumCompilingDeep2021}.
In variational recompilation, we want to find a parametrised unitary circuit $U(\underline{\theta})$ that approximates a target unitary $V$.
This target unitary is required to approximate the action of $U$ on the entire Hilbert space in case of recompiling a Full Unitary~\cite{khatriQuantumassistedQuantumCompiling2019} in \cref{fig:gradient_comparison}(b) or just to approximate the action on a specific input state in \cref{fig:gradient_comparison}(a).
After applying the circuits $V$ and $U(\underline{\theta})^\dagger$ consecutively, the goal is to find circuit parameters $\underline{\theta}$ such that the state of the registers is in the ground state $| \underline{0} \rangle$ of the Hamiltonian $\mathcal{H}=- \sum_{j=1}^N Z_j$ \cite{jonesRobustQuantumCompilation2022}. However, the problem would be equally solved by finding ansatz parameters to produce any computational basis state, (i.e. any eigenstate of $\mathcal{H}$) given we can then just append single qubit $X$ rotations to the ansatz to produce the desired operation. This feature, along with the local Hamiltonian allowing for the efficient measurement of many covariances makes it particularly amenable to root finding which is not limited to only searching for the ground state.

For this reason we apply \cref{cor:commuting_observables} to the present problem and
define our problem Hamiltonians as $\mathcal{Q}:= \{ Z_a \}_{a=1}^N$.
We can indeed enlarge this pool by further considering products of single-qubit Pauli $Z$ operators.
Our aim is then to find joint roots of the covariances from \cref{cor:commuting_observables}
which then guarantee that the solution corresponds to a joint eigenstate of all operators in $\mathcal{Q}$
as one of the computational basis states. After having found one of these computational states
we just apply single qubit $X$ rotations to our ansatz to map to the $| \underline{0} \rangle$ state.

Here we consider an example of \emph{parameter rediscovery} as a benchmark, whereby we recompile a unitary $V=U(\theta^\star)$ that has the same form as the parametrised quantum circuit $U(\theta)$ but with the parameters fixed at some random solution values $\theta^\star$. This has the advantage of being a very hard problem to solve variationally, while also giving us a guarantee that our circuit is capable of expressing the solution -- while note that below we will benchmark our \cv approach on practical problems as well.

\cref{fig:gradient_comparison}(c) shows the performance of root finding 
when we initialise relatively close to the solution (by disturbing parameters $|\Delta \theta_k| \leq 0.3$) on a $10$-qubit, $2$-layer parameter rediscovery problem and compares it to the performance of gradient descent.
In this recompilation problem our operator pool $\pool$ contained all 3-local Pauli strings
and we chose randomly $\nc$ operators at every iteration, see details in \cref{sec:recomp_details}.
Indeed, \cref{fig:gradient_comparison}(c) confirms that root finding is able to converge significantly faster to the shot noise floor, i.e.,
a limitation due to finite sampling of expected values.
Furthermore, \cref{fig:gradient_comparison}(c) confirms that root finding has a significantly improved convergence rate (steeper slope), which is improved with a greater number of constraints (blue vs. light blue), while note that the quantum resources required for a single iteration is comparable to that of gradient descent.

\cref{fig:gradient_comparison}(d/blue) shows applying root finding to an initial state that we obtained by 
a short period of gradient descent from a random state -- applying gradient descent to a random state has the effect
of producing a state with an appreciable overlap with the lowest energy eigenstates, allowing root finding to efficiently converge.
In contrast, in \cref{fig:gradient_comparison}(d/black) we demonstrate the performance of root finding when we start from a randomly chosen initial point in parameter space on the same problem. It fails to make any progress; This is due to the fact that root finding works well when there are only a small number of eigenstates that significantly contribute to the state produced by the PQC. In contrast, random states as nearly equal superpositions of a large number of basis states do not contain a dominant eigenstate towards which root finding could converge thus \cv fails to make significant progress.
 This is very much analogous with fault-tolerant phase-estimation protocols which do indeed similarly fail under
 random initialisation, but enable us to efficiently prepare any eigenstate given a good initialisation is possible. 
 These signify the importance of initialisation when searching for eigenstates and clearly demonstrate that even a short period of gradient descent may be sufficient for these purposes.

\begin{figure}[tb]
	\begin{centering}
		\includegraphics[width=\linewidth]{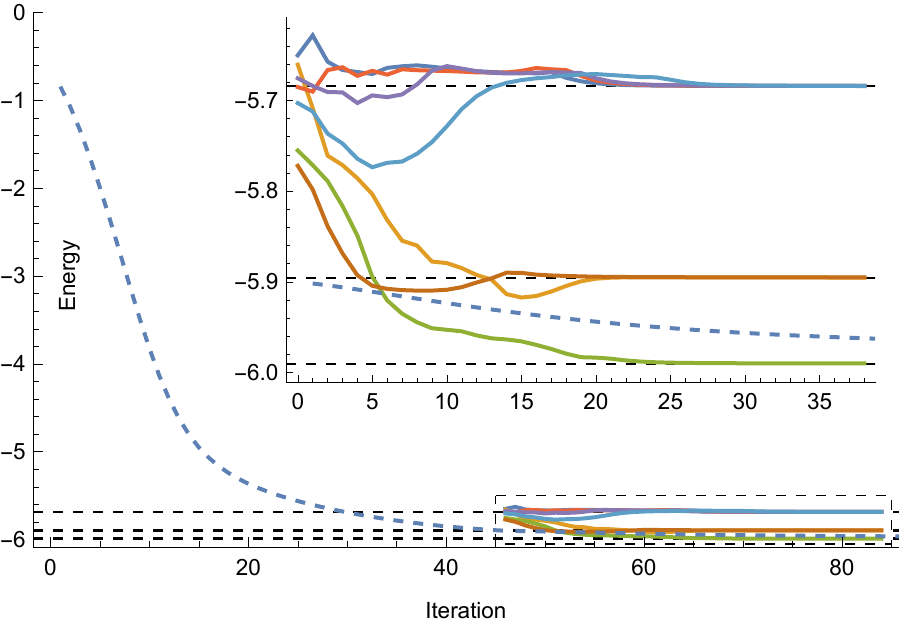}
		\caption{Using root finding to explore a low energy subspace of a $10$-qubit spin chain via an ansatz circuit of 20 layers ($\nu = 610$). The initialisation to the low energy subspace is performed by Imaginary Time Evolution (Blue, Dashed) and is followed by runs of \cv to find the three lowest eigenstates (marked by dashed lines), to an accuracy of $\Delta E \le 4 \times 10^{-4}$ in all cases. Imaginary Time Evolution was run for 250 iterations after the point where root finding started and converged to a state with $\Delta E = 0.012$. These simulations do not include shot noise.}
		\label{fig:spinchainLowestEigs}
	\end{centering}
\end{figure}

The performance of variance-VQE (a gradient based method that minimises the variance of the Hamiltonian, see \cref{sec:variance} for more details) is also shown for comparison -- it is another method which, like root finding, is not only searching for the ground state of the Hamiltonian. Variance-VQE is not stuck in the same way as root finding, but makes slow progress due to its relatively (compared to root finding) slow convergence speed.

\subsection{Spin Models \label{sec:application_spin}}
Spin models are highly relevant in the study of condensed matter physics, quantum statistical mechanics
and many problems in practice can be mapped to spin models, such as NP problems~\cite{lucasIsingFormulationsMany2014}, cf. approximate optimisation algorithms (QAOA) or spin systems in materials science~\cite{cerezo2020variationalreview,endoHybridQuantumClassicalAlgorithms2021,bharti2021noisy, pagano2020quantum,arute2020quantum}.
Furthermore, lattice models of quantum field theories \cite{latticeschwinger} typically have local Hamiltonians.

Here we simulate \cv when used to search for low energy excited states of a spin chain 
described by the Hamiltonian
\begin{equation}\label{eq:spin-ring}
    \mathcal{H}=J\sum_{i=1}^N \Vec{\sigma}_i \cdot \Vec{\sigma}_{i+1} + c_i Z_i.
\end{equation}
We use periodic boundary conditions ($N{+}1{=}1$) and $c_i$ are randomly selected onsite interactions comparable in strength to the couplings $J$, and $\Vec{\sigma}_i$ is the Pauli vector for the $i$th qubit. This Hamiltonian cannot be simulated classically for large $N$ with reasonable computational resource despite its simple structure~\cite{luitzManybodyLocalizationEdge2015,childs2018toward} -- and it is relevant for studying the phenomenon of many-body localization in condensed matter systems \cite{nandkishoreManyBodyLocalizationThermalization2015}.

We use a hardware-efficient ansatz of 20 layers for $10$ qubits.
Before root finding was performed, the ansatz parameters were initialised to $\underline{\theta}_{imag}$ using Imaginary Time Evolution \cite{samimagtime} to a state with low expected energy. This ensured that only a limited number of eigenstates contributed significantly to the state produced by the ansatz. \cv was then performed from points in parameter space with small random variations from $\underline{\theta}_{imag}$ to map out the low energy subspace. This process is shown in \cref{fig:spinchainLowestEigs} in mapping out the lowest energy subspace of a $10$-qubit spin chain; \cv rapidly finds a state an order of magnitude closer to the ground state in energy error than Imaginary Time Evolution is able to converge to. This indeed confirms that root finding has a significantly improved convergence rate when compared to
Imaginary Time Evolution, where the latter is equivalent to natural gradient \cite{quantumnatgrad, yamamoto2019natural} and is thus a second-order method
that requires increased absolute quantum resources (samples) \cite{quantumnatgrad}.
While here we focused on practical applications of \cvns,
we demonstrate additional numerical simulations in \cref{fig:spinchain_close} in the Appendix
whereby we explicitly compare \cv to VQE;
we conclude that in comparison to other variational methods, \cv
again exhibits superior performance on the present spin-ring problem.

\begin{figure}[tb]
	\begin{centering}
		\includegraphics[width=0.9\linewidth]{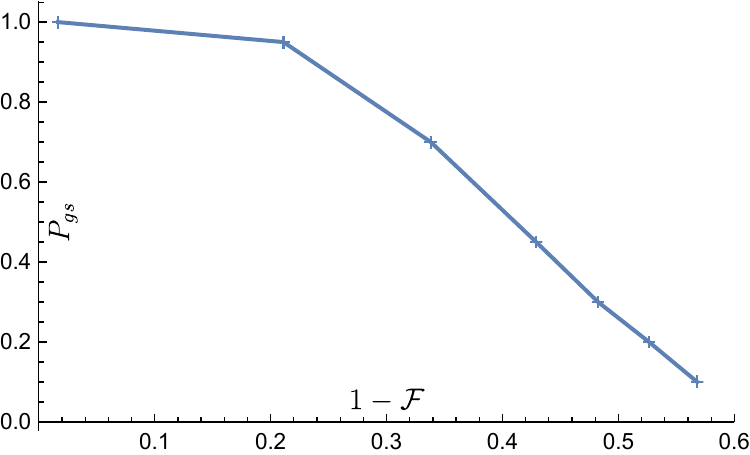}
		\caption{
			Probability $P_{gs}$ as the proportion of 20 runs where \cv converged to the ground state of the spin chain as
				a function of starting infidelity from the ground state. While here we only report the probability of converging to the ground state specifically, we note that $97.1 \%$ of all runs converged to one of the low-energy eigenstates -- the distribution of eigenstates found is detailed in \cref{sec:groundstate_barchart}.
			Initialisation was done the same way as in \cref{fig:spinchainLowestEigs} by running Imaginary Time Evolution from a random initial state until reaching selected energies between $-5.8$ and the ground state. A linear decay of success probability---reminiscent of phase-estimation protocols---with ground state overlap is found, although we expect the rate of this decay to strongly depend on the population-distribution of the lowest-energy eigenstates.
		}
		\label{fig:spinchain_groundstate}
	\end{centering}
\end{figure}

\subsection{Finding excited states}
Searching for excited states of a Hamiltonian is important in many practical applications and indeed variational
quantum algorithms have been proposed for solving this class of problems.
On the other hand, finding excited states using orthogonality constrained VQE techniques \cite{higgott2018variational,jonesRobustQuantumCompilation2022} can be difficult because we need to discover the parameters for and project out every state from the ground state up to the energy of the state we wish to find. \cv is agnostic to the energy of the state and acts to find states close to the one it is initialised into.
Although \cv could converge to states that have been found in previous runs of the algorithm, similar techniques of projecting out previously found eigenstates could also be applied. We can also potentially use classical techniques, such as Interval Analysis \cite{mooreIntroductionIntervalAnalysis2009}, to find all the roots within an area of parameter space, reducing the problem converging to an already known eigenstate (root). 

In \cref{fig:spinchain_groundstate} we show the probability of converging to the ground state of our spin chain as a function of
the overlap between the ground state and the initial state. These initial states were obtained by performing Imaginary Time Evolution from random points in parameter space down to energies between $-5.8$ and the ground state (at $E=-5.99$). The initial fidelity where root finding was started was recorded (\cref{fig:spinchain_groundstate} horizontal axis) and the fraction of runs that converged to the ground state is listed.
Indeed we find a nearly linear relationship and the vast majority of runs that do not converge to the ground state
converge to one of the other low-lying energy levels. We provide the distribution of how these runs converged in \cref{sec:groundstate_barchart}.
Furthermore, the observed relation between fidelity and probability is directly analogous
to phase estimation whereby a measurement is used to collapse the system into the desired eigenstate
with a probability that is given by the fidelity with respect to that state.
In contrast, variational quantum algorithms converge to local minima and 
only an exponentially small fraction of such local energy minima may be close to the ground state
as proved in ref~\cite{anschuetz2022beyond}.

\section{Comparison to existing techniques\label{sec:comparison}}
\subsection{Relation to variational quantum algorithms}

A large subset of variational quantum algorithms is concerned with minimising
a cost function $E(\underline{\theta})$ which is usually the expected value of a Hamiltonian
$\langle \mathcal{H} \rangle$, such as in case of VQE. Among other practical limitations, this surface may have a large number of
local minima that can trap local optimisers \cite{bittelTrainingVariationalQuantum2021}.
The main difference is that \cv uses a large number of such	 surfaces that are
randomly selected and computing this large data requires similar quantum resources (shots)
as a standard gradient estimation in case of variational quantum algorithms. In the
following sections we compare \cv in more detail to specific variational quantum 
algorithms and related techniques.

\subsection{Comparison to variance minimisation	 \label{sec:variance}}

Minimising the variance $\sigma^2 := \langle \mathcal{H}^2  \rangle - \langle \mathcal{H}\rangle^2$ of a Hamiltonian 
allows us to find eigenstates and has been explored in the context of quantum chemistry~\cite{variance_minimisation}.
Furthermore, the so-called variance-VQE~\cite{zhangVariationalQuantumEigensolvers2020} approach uses variational quantum
circuits and estimates this variance as well as its derivative using a quantum computer. Note that
the tools we have introduced can naturally be applied in this context: 
Given a minimal pool $\pool = \{\mathcal{H}_a\}_{a=1}^r$ that only contains Pauli terms of the
Hamiltonian and thus a covariance matrix $\bilinear{\mathcal{H}_a}{\mathcal{H}_b}{\psi}$ of the Hamiltonian terms only,
our Jacobian allows us to compute the gradient of the variance  $\partial_n \sigma^2 = \jac^\intercal \underline{h}$,
as we detail in \cref{app:variance_grad}.

While our \cv approach contains full information about the variance and its gradient, it is important to recognise, however, that
a gradient descent optimisation has an inferior convergence rate when compared to the quadratic Newton method.
In \cref{fig:gradient_comparison} we compare the two techniques and indeed find that \cv has a
superior performance. Furthermore, note that \cv uses strictly more information than variance minimisation:
while the operator pool $\pool = \{\mathcal{H}_a\}_{a=1}^r$ containing only the Hamiltonian terms
is sufficient in \cref{theo:main}, we significantly enlarge this pool  such that $\nc \gg r$ and aim
to find joint roots of this large number of covariances.

\subsection{Comparison to subspace expansion}

Subspace expansion \cite{mccleanSubspaceExpansion} is a method for discovery of low-energy excited states starting from an estimated ground state $| \tilde{\psi}_G \rangle$. One then explores directions in Hilbert space by applying low-weight operators as excitations to the ground state.
Typically operator pools of Pauli operators are used to produce a new set of states $\ket{\psi_k}=\mathcal{P}_k \ket{\tilde{\psi}_G}$ for calculating the overlaps $H_{kj}=\bra{\psi_k}\mathcal{H}\ket{\psi_j}$ and $S_{kj}=\braket{\psi_k|\psi_j}$.
Diagonalising $H_{kj}$ then reveals better ground-state energies than that of $| \tilde{\psi}_G \rangle$. As opposed to
\cvns, subspace expansion cannot prepare the ``good quality'' representation of the eigenstates with a quantum computer, but is rather limited to estimating their energies with a classical computer.

The connection to \cv is elucidated further in \cref{app:state_overlaps} where we express covariance functions
in terms of a set of quantum state overlaps $\bilinear{O_k}{\mathcal{H}}{\psi} = 	\langle \phi_{O_k}  | \phi_\mathcal{H} \rangle$
where we can define the (unnormalised) vectors  $| \phi_A \rangle := (A - \langle A \rangle)  |\psi \rangle $.
The covariance matrix in Eq.~\eqref{eq:covar_matr} is thus a positive-semidefinite overlap matrix
$[\cmat (\psi)]_{kl}  = \langle \phi_{O_k} |\phi_{O_l} \rangle$.
We can make a weak analogy to subspace expansion based on the following observation.
Given covariance functions can be expressed as state overlaps, the method explores possible
directions in Hilbert space via our operator pool that is beyond the capabilities of the ansatz and we gain information through a linearisation as the Jacobian at what parameters the nearest eigenstate may be found via the vanishing overlaps. Similarly, subspace expansion also uses operators additional to the ansatz to explore around the estimated ground state to extract low energy excited states.
As such, we may be able to use existing heuristics from subspace expansion for selecting problem-specific operator pools.

\subsection{Comparison to Hessian optimisation} \label{sec:hessian_comparison}

A Hessian-based Newton-Raphson optimisation \cite{gillPracticalOptimization1981} of the function $E(\theta)$ uses the update rule
\begin{equation}
    \theta^{(t+1)} = \theta^{(t)} - \eta \mathbf{H}^{-1} \nabla E(\theta),
\end{equation}
Where $\mathbf{H}$ is the Hessian matrix of second derivatives $\partial_n \partial_n E(\theta)$
which can be estimated on a quantum device using standard techniques such as the parameter-shift rule \cite{mariEstimatingGradientHigherorder2021,koczor2020quantumAnalytic}.

Let us compare our root finding with this Hessian optimisation. A similarity may be that
both methods use additional information to provide improved convergence over methods that use only gradient information.
The Hessian obtains information about the local curvature of the $E(\theta)$ manifold which is a ``classical'' multivariate function
and its local curvature can even be accurately captured by analytical approximations~\cite{koczor2020quantumAnalytic} -- with the
use of a quantum computer.
In contrast, \cv uses information from operator covariances of the variational state $\ket{\psi(\theta)}$ extracted from exploring directions in Hilbert space through a randomly chosen, large operator pool as detailed in \cref{app:state_overlaps}. In this sense, \cv is a quantum-aware method, as it does not only use the fact that $E(\theta)$ (and its derivatives) are efficiently calculable on a quantum computer, but also the relationship between $\ket{\psi(\theta)}$, the Hamiltonian and possible directions in Hilbert space.
The most pronounced difference between the two techniques is from a practical point of view: while we can use the classical shadow procedure $\nu$-times to estimate a very large Jacobian, for the Hessian we need to use it $\mathcal{O}(\nu^2)$-times at different circuit-parameter configurations. In this sense \cv obtains significantly more information using fewer measurements.

It is also interesting to point out that the Gauss-Newton and LM techniques
can be interpreted as approximate Hessian optimisations of the vector norm $\lVert \covarvec\rVert^2$ when the Jacobian
$\mathbf{H} \approx 2 \tilde{\jac}^\intercal \tilde{\jac}$ gives a good approximation of the Hessian matrix by
keeping only the first order derivatives $[\partial_n \covar_k(\underline{\theta})] [\partial_m \covar_k(\underline{\theta})]$
but neglecting second order derivatives of the form $\partial_n \partial_m \covar_k(\underline{\theta})$.
This, however, does not apply to \cv given in our case the functions $\covar_k(\underline{\theta})$
are trigonometric polynomials for which the second order derivatives are dominant, especially near solutions.
Take for example the simple function $g(\underline{\theta})  = -\prod_{n=1}^\nu \cos(\theta_n)$ which has a minimum at $\underline{\theta}=(0,0 \dots,0)$
and while its first derivatives vanish near the minimum, the second derivatives $|\partial_n \partial_n g(\underline{\theta}) |\approx 1$
are dominant.
Note also that the vector norm $\lVert \covarvec\rVert^2$ has no immediate relation with the energy surface
$E(\theta)$, and thus it is not related to a second-order energy minimisation.

\subsection{Comparison to natural gradient}
Quantum Natural Gradient, which is equivalent to Imaginary Time Evolution for ideal, unitary circuits improves over both the convergence rate and ability to avoid traps of vanilla gradient descent by taking into account the geometry of the space of quantum states.
Here we estimate the Quantum Fisher Information~\cite{samimagtime,quantumnatgrad, yamamoto2019natural,koczor2019quantum} as
\begin{equation}
    \mathbf{F}_{ij}(\underline{\theta})=\langle \partial_i \psi |\partial_j \psi \rangle - \langle \partial_i \psi |\psi \rangle\langle \psi |\partial_j \psi \rangle
\end{equation}
which is equivalent to the real part of the Quantum Geometric Tensor $ \mathbf{F}(\underline{\theta})=\re [\mathbf{F}_{ij}]$ and is used
to compute the parameter update
\begin{equation}
    \underline{\theta}_{t+1} = \underline{\theta}_t - \eta \mathbf{F}^+(\underline{\theta}_t) \nabla E(\underline{\theta}_t).
\end{equation}
Similarly to Hessian optimisation, this method has a complexity $\mathcal{O}(\nu^2)$ for calculating the tensor which,
as opposed to the Hessian, is independent of the energy surface, and rather expresses relations between states reached by varying different parameters.

\cv can also be thought of as an optimisation using additional information about the space of quantum states, but in this case extracted from the many covariance functions (which do depend on $\mathcal{H}$) instead of the geometric tensor.

\subsection{Relation to parent Hamiltonians}
Let us now relate \cv to existing techniques that do not aim to find eigenstates, but rather aim to find Hamiltonians  $\mathcal{H}_{parent}$ that encode a fixed state $|\psi\rangle$ as an eigenstate. The  $\mathcal{H}_{parent}$
are then called as parent Hamiltonians to the state $|\psi\rangle$.

In particular, these techniques proceed by assuming that the parent Hamiltonian
can be expressed in terms of the ansatz as a linear combination of basis operators $\mathcal{H}_a $
as~\cite{PhysRevX.8.031029,Qi2019determininglocal}
\begin{equation}
	\mathcal{H}_{parent} = \sum_{a=1}^r h_a \mathcal{H}_a,
\end{equation}
via the real coefficient vector $\hvec \in \mathds{R}^{r}$. Here $r$ is the rank of the decomposition, i.e., the number of independent basis operators.
The covariance matrix $\cmat(\psi) \in \mathds{C}^{r \times r}$
then depends on our trial quantum state $|\psi \rangle$
that we have defined in \cref{eq:covar_matr}.
The parent Hamiltonian is then found by finding the nullspace of this covariance
matrix given every coefficient vector $\underline{h}$ in the nullspace
satisfies $\cmat \underline{h} = 0 $ and given our expression for the variance
in \cref{lemma:variance} it guarantees that $\underline{h}^\intercal \,  \cmat \, \underline{h} = \var[\mathcal{H}_{parent}] =0$ in the
particular state $|\psi \rangle$.

CoVaR clearly works according to a reverse logic whereby the problem Hamiltonian is fixed and we search for quantum 
states that result in $0$ covariances. We then search the space of quantum states via an efficient parametrisation, i.e.,
variational circuits, using a quantum computer.

\section{Discussion}
We have demonstrated that \cv shows significantly improved performance by many orders of magnitude compared to analogous variational algorithms due to its effective use of classical shadows. {However, a main limitation is its vulnerability to random parameter initialisation. Although this seemingly has a resemblance to barren plateaus whereby expected-value landscapes suffer from flat regions due to vanishing gradients \cite{anschuetz2022beyond, cerezoCostFunctionDependent2021a, mcclean2018barren}, the present limitation is quite different in nature. 
In particular, recall that phase-estimation protocols provably efficiently find an eigenstate of
an efficiently simulable Hamiltonian given an initial state is provided with a sufficiently large overlap with the desired eigenstate;
Our approach is quite similar as it gets attracted to any eigenstate with a significant contribution to the initial state. For this reason we expect the main limitation of the present approach is not decoupled from the general challenge of finding good initialisation for fault-tolerant phase estimation protocols
or finding problem specific ans{\"a}tze for VQE problems -- and this challenge may be attributed to general hardness results of finding ground/eigenstates~\cite{bookatz2012qma}.}
Interestingly, we have demonstrated in numerical simulations that even a short period of gradient descent evolution provides sufficient initialisation in practice. 

Furthermore, barren plateaus do not necessarily exist for our focus of local Hamiltonians and shallow ans{\"a}tze~\cite{anschuetz2022beyond,cerezoCostFunctionDependent2021a}. Nevertheless,
random initialisation of gradient-based VQE optimisers still prohibits finding eigenstates of large systems 
due to local traps~\cite{anschuetz2022beyond}:
First, optimising VQAs has been shown to be NP-hard due to persistent local minima~\cite{bittelTrainingVariationalQuantum2021};
Second, ref.~\cite{anschuetz2022beyond} proved
that a broad class of shallow VQA models that exhibit no barren plateaus are untrainable due to local traps.
As we demonstrated in numerical simulations, our stochastic Levenberg-Marquardt approach indeed does mitigate the effect of
these local traps:
while
gradient based optimisers fail to make progress around a local trap as the gradient of the \emph{energy surface} vanishes, \cv is not an energy minimiser and those specific parameters may well yield a non-zero step for \cvns.
Furthermore, \cv is also less vulnerable to getting stuck due to our randomly generated constraints (covariances): even
if a single iteration makes no progress, in the next iteration a new, randomly generated set of constraints may well yield a non-zero step 
as we demonstrate in \cref{fig:localmintest} where sometimes several steps of \cv are required to escape a trap.
This is analogous to the well-known advantage of stochastic gradient descent in the machine-learning context~\cite{ruder2016overview}
and we note that exploring globally convergent root-finding techniques may also be a fruitful direction for future research~\cite{dennis1996numerical,okawa2018w4,pasquini1985globally}.

There are a number of apparent extensions to our approach that we have not considered here. 
First, given the classical shadows are stored in a classical computer we can in principle determine multiple sets of update rules from them and apply the one that most decreases the variance or any other metric as opposed to our fully randomised scheme.
Second, it would be worth exploring some specific use cases of \cv in more detail, such as finding highly excited states.
Third, in this work we have used classical shadows to extract a large number of covariances in the case where both the Hamiltonian and operator pool is constructed of local Pauli strings. It is an interesting direction for further work to attempt to use other randomised measurement channels to measure covariances with similar efficiency for non-local Hamiltonians or operator pools.
Several works have appeared recently that make significant progress by
developing shadow-measurement channels that interpolate between Pauli measurements (as in this work) and the powerful global Clifford measurements~\cite{akhtarScalableFlexibleClassical2022, bertoniShallowShadowsExpectation2022}. These intermediate-depth techniques are amenable to NISQ devices and allow for the measurement
of non-local properties.
As a matter of fact, related techniques leveraging simultaneous measurements of commuting observables are also highly relevant
given they allow the efficient reconstruction of a large number of not necessarily local Pauli strings as crucial in applications of quantum chemistry \cite{Crawford2019,yen2020measuring,jena2019pauli,gokhale2020n}.

{Finally, we expect the present approach to be resilient against reasonable levels of experimental noise. First,} our update rule in \cref{eq_jacobian} is invariant under global depolarising noise when the expected value $\langle\mathcal{H}\rangle$ is known to high precision, e.g., via well-established error mitigation techniques~\cite{qemreview,koczor2020exponential,koczor2021dominant,huggins2020virtual}.
{Second, we numerically simulated an approximate noise model that goes beyond global depolarisation and we observed a very good robustness against experimental noise. While these observations speak for the practicality of the present approach, we leave it to future work to confirm the performance in current and near-future generation hardware.
}

\section{Conclusion}

In this work we considered powerful variational quantum circuits that have been
extensively investigated in a hope to exploit near-term quantum computers. Most of these
near-term quantum algorithms aim to encode the solution to a practical problem of
interest to eigenstates of a Hamiltonian, typically the ground state. As a direct
analogy to successful variational techniques in quantum chemistry, nearly all
quantum variational algorithms so far have proceeded by posing the problem as a variational search.
In this paradigm we minimise the single classical cost function---typically the expected value of a Hamiltonian---with respect to circuit parameters.

Our work opens a new research direction in the efforts of achieving practical value with near-term quantum
computers: we observe that the condition for finding  eigenstates
can be posed as finding joint roots of a large number of properties of the quantum state
as covariance functions -- these express fundamental quantum-mechanical uncertainty relations. We have devised the powerful root finding technique \cv
and demonstrated that increasing the number of these constraining covariances significantly increases the
efficacy of the search procedure.

The most remarkable feature of \cv is that
it allows us to fully exploit the extremely powerful classical shadow techniques in a way
that prior variational techniques could not, i.e., we \emph{simultaneously}
estimate a very large number of randomly chosen properties, e.g., $>10^4-10^7$ of the quantum state and their derivatives
with respect to circuit parameters. These inform our search procedure via a large but tractable
linear system of equations that we solve with a classical computer.

Our approach can  be viewed as directly
analogous with (stochastic gradient descent and) stochastic Levenberg-Marquardt techniques
that have been extremely popular in the context of classical machine learning
 	-- and we generally expect \cv inherits the fast convergence speed of Levenberg-Marquardt as we indeed demonstrated in practical examples.
 	In fact, Levenberg-Marquardt is
	the default and fastest method for training classical neural networks~\cite{neuralnet,demuth2014neural,beale2010neural,yu2018levenberg}
	-- 	but with a limitation that handling a large Jacobian becomes the bottleneck for too deep neural networks.
	In stark contrast, we view this limitation of the classical technique a major advantage
	of our approach given we can populate the large Jacobian using
	only \emph{logarithmic} quantum resources. \cv thus allows us to offload non-trivial but
	tractable calculations onto the classical computer and combines the best of both worlds, i.e., fast convergence and
	fast (quantum) computation of the Jacobian.
	Furthermore, as we demonstrated, using a large number of randomly generated 
	constraints makes \cv particularly robust against getting stuck in local traps
	in analogy with stochastic techniques in the classical machine learning context~\cite{neuralnet, ruder2016overview}.

We proved that the quantum resources, using classical shadows, required for a single iteration of our procedure is
comparable to that of a standard gradient estimation in conventional VQE. 
In addition to its significantly improved convergence speed,
our approach exhibits a robustness against shot noise and against experimental imperfections thanks to our large dataset.
We have explored a number of practically motivated important applications whereby the 
problem Hamiltonian is local given the classical shadow procedure is very NISQ-friendly 
in such scenarios, requiring only single-qubit measurements in a random basis. These include, recompiling quantum circuits and finding ground and excited
states. Our numerical simulations confirm the superiority of our approach and indicated that
it can significantly outperform others by many orders of magnitude. Furthermore, previous
techniques for finding excited states of Hamiltonians assumed a sequential search whereas
ours naturally converges to any of the eigenstates -- and can thus be applied naturally in this context.
Similarly, recompilation is another natural set of problems for \cv given any eigenstate of 
the problem Hamiltonian can be accepted as a solution. 
While the presented applications tackling local Hamiltonian problems are ideal for \cvns, important quantum chemistry problems 
may be non-local and may thus be challenging for classical shadows depending on the encoding.
Fortunately, two fields of active research are making progress to alleviate this issue:
first, compact fermion encodings result in local Hamiltonians at the cost of a modest qubit overhead, and second, recent advancements in classical shadows allow for efficiently measuring non-local properties or specifically,
measuring fermionic operators~\cite{wanMatchgateShadowsFermionic2022, zhaoFermionicPartialTomography2021,akhtarScalableFlexibleClassical2022, bertoniShallowShadowsExpectation2022}.
It is thus expected the present approach will
be highly competitive and will spark further developments in the field.

Finally, our work makes exciting connections to various fields, including fundamental  uncertainty relations in quantum mechanics as covariances, exploitation of classical shadows, stochastic optimisation in machine learning and working with big data. We believe
it will be interesting to explore these connections to further improve the presented techniques
in the hope of achieving practical value with near-term quantum computers.

\section*{Acknowledgments}
We thank Jonathan Foldager, Hsin-Yuan Huang and Suguru Endo for providing us with useful comments.
We thank Simon C Benjamin for his support throughout this work.
B.K. conceived the idea and contributed to writing the manuscript, 
G.B. performed numerical simulations and contributed to writing the manuscript.
G.B. and B.K. acknowledge the EPSRC Hub grant under the agreement number
EP/T001062/1 for hardware provision.
B.K. thanks the University of Oxford for
a Glasstone Research Fellowship and Lady Margaret Hall, Oxford for a Research Fellowship.
The numerical modelling involved in this study made
use of the Quantum Exact Simulation Toolkit (QuEST), and the recent development
QuESTlink\,\cite{QuESTlink} which permits the user to use Mathematica as the
integrated front end. We are grateful to those who have contributed
to both these valuable tools.

\appendix

\section{Orthogonal constraints\label{app:orthogonal}}

We presented the general theory of our approach in \cref{sec:general}
whereby we compute covariances with respect to an arbitrary operator pool $\pool$
in order to search for eigenstates of an arbitrary Hamiltonian $\mathcal{H}$.
While our practically motivated CoVAR approach leverages powerful classical
shadows, it restricts the problem Hamiltonian and the operator pool to local Pauli strings.
While Pauli strings are orthonormal in operator space, their actions  on quantum
states are generally not orthogonal directions in Hilbert space. In the present section we
explore another kind of operator pool $\pool$ whereby the operators represent orthogonal
directions in state space. Let us first start by interpreting covariances
as overlaps in Hilbert spaces.

\subsection{Interpretation as state overlaps\label{app:state_overlaps}}

\begin{lemma}
	The covariances from \cref{def:covars} can be interpreted as overlaps between quantum states as
	\begin{align}
		\bilinear{A}{B}{\psi} = 	\langle \phi_A  | \phi_B \rangle,
	\end{align}
	where we can define the (unnormalised) vectors  $| \phi_A \rangle := (A - \langle A \rangle)  |\psi \rangle $ for any operator $A \in \mathbb{C}^{d \times d}$ 
	with norm $\lVert \phi_A \rVert^2 = \var[A]$.
\end{lemma}
\begin{proof}
	The above property immediately follows from the defining expression of covariances from \cref{def:covars}.
\end{proof}
The above lemma informs us that in \cref{theo:main} and in \cref{cor:commuting_observables} we compute overlaps as $\bilinear{O_k}{\mathcal{H}}{\psi} = \langle \phi_k | \phi_\mathcal{H} \rangle$ and thus we actually decompose the quantum state $\mathcal{H} |\psi\rangle$ into a set of quantum states $| \phi_l \rangle$ that we obtain by acting on  $|\psi \rangle$ with elements of our operator pool. Given that $| \phi_\mathcal{H} \rangle$ must be the null vector when the eigenvalue equation is satisfied, it is necessary that any (non-parallel) overlap with this vector must vanish.

While Pauli strings form an orthonormal basis of operator space, they have the disadvantage that in Hilbert space they result in non-orthogonal actions $\langle \phi_k | \phi_l \rangle \neq \delta_{kl}$, i.e., we decompose the vector $\mathcal{H} |\psi\rangle$ into a non-orthogonal basis.
On the other hand, it is possible to define an operator pool that corresponds to orthogonal directions in Hilbert space.
\begin{lemma}[\textbf{Orthogonal operator pool}]\label{lemma:orth}
	Let us consider strings of $X$ Pauli operators as $X_k \in \{ \mathrm{Id}, X \}^{\otimes N}$ using the binary index $k \in \{0,1\}^N$.
	We define the orthogonal operator pool $\pool := \{O_k := U X_k U^\dagger, k \in \{0,1\}^N\} $ via the operators,
	where the unitary quantum circuit $U$ maps our reference $|\psi\rangle = U|0\rangle$ onto our quantum state.
	The quantum states $|\phi_k \rangle := O_k | \psi\rangle$ form an orthonormal system  $\langle \phi_k | \phi_l \rangle= \delta_{kl}$ and
	thus the operators $O_k$ map to orthogonal directions in Hilbert space.
\end{lemma}
\begin{proof}
		Orthonormality in Hilbert space follows from $\langle O_k \rangle = \delta_{k0}$ and via
	\begin{equation*}
		\langle \psi| O_k O_l |\psi \rangle = \langle 0 | U^\dagger U X_k U^\dagger U X_l U^\dagger U |\underline{0} \rangle = \langle k | l \rangle,
	\end{equation*}
	where $|k\rangle$ are standard basis vectors with $\langle k | l \rangle = \delta_{kl}$.
	
\end{proof}

The resulting covariances can actually be shown to be entries in a column vector of the Hamiltonian matrix. In particular, given  the states $| \phi_k \rangle $ form an orthonormal basis they can be used to represent the Hamiltonian matrix
as the covariances $\covar_k  = \langle \psi  | O_k \mathcal{H} | \psi \rangle =  \langle \phi_k  | \mathcal{H} | \phi_0 \rangle = \mathrm{Col}_0(\mathcal{H})$, which are then actually elements of the first column vector of the problem Hamiltonian.
These covariances, as entries of the Hamiltonian matrix, can be computed using the Hadamard-test techniques presented in ref.~\cite{li2017}.
In particular, the covariances are obtained by applying our variational quantum circuit $U$ onto the standard computational basis states as $| \phi_k \rangle = U X_k U^\dagger |\psi \rangle = U |k \rangle$
and we compute the overlap of this state with our variational quantum state $|\psi\rangle$.  The approach can straightforwardly be
implemented via the Hadamard test, whereby we apply the $X_k$ operations in $|k \rangle = X_k | \underline{0}\rangle$ controlled on an ancilla qubit.
Derivatives of these covariances can be similarly computed by applying the generator of the quantum gate controlled on the same ancilla.
Let us now show that sum of squares of the covariances is equivalent to the variance of the Hamiltonian.
\begin{lemma}\label{lemma:orth_sum_squares}
	Given the orthogonal operators introduced in \cref{lemma:orth} we compute the corresponding covariances.
	While $\covar_0 = \langle \psi |\mathcal{H} | \psi \rangle$ is the energy expected value, we can show that
	\begin{equation} 
		\sum_{k=1}^{2^N-1} |\covar_k|^2 =  \lVert \covar \lVert^2 = \mathrm{Var}[\mathcal{H}].
	\end{equation}
\end{lemma}

\begin{proof}
	\begin{equation*}
		\sum_{k=1}^{2^N-1} |\covar_k|^2 = \sum_{k=0}^{2^N-1} \langle \psi  | \mathcal{H} | \phi_k \rangle        \langle \phi_k  | \mathcal{H} | \psi \rangle - \langle \psi  | \mathcal{H} | \psi \rangle^2.
	\end{equation*}
	As the $| \phi_k \rangle $ are a complete basis set due to being a unitary transformation of the computational basis, $\sum_k | \phi_k \rangle        \langle \phi_k  | = \mathrm{Id}$ and we therefore obtain 
	\begin{equation*}
		\lVert \covar \lVert^2 =\langle \psi  | \mathcal{H}^2 | \psi \rangle - \langle \psi  | \mathcal{H} | \psi  \rangle^2 = \mathrm{Var}[\mathcal{H}].
	\end{equation*}
\end{proof}
Importantly, while this operator pool has the advantage that the covariances represent independent, orthogonal directions
in state space it is clear that we would generally need to compute all $2^N -1$ of these orthogonal constraints, as elements
of the first column of the Hamiltonian matrix, in order to be able to compute the variance and thus to verify that
the quantum state $|\psi\rangle$ is an eigenstate of the Hamiltonian. In stark contrast, in \cref{theo:main} we have shown
that given a decomposition of a Hamiltonian into an operator basis $\mathcal{H} = \sum_{k=1}^r h_a \mathcal{H}_a $ which typically
grows polynomially with the system size, it suffices to only compute the corresponding polynomial number
of covariances. Although these operators, such as Pauli strings, are orthonormal in operator space,
they do not correspond to orthogonal directions in Hilbert space.

\begin{figure}[tb]
	\begin{centering}
		\includegraphics[width=\linewidth]{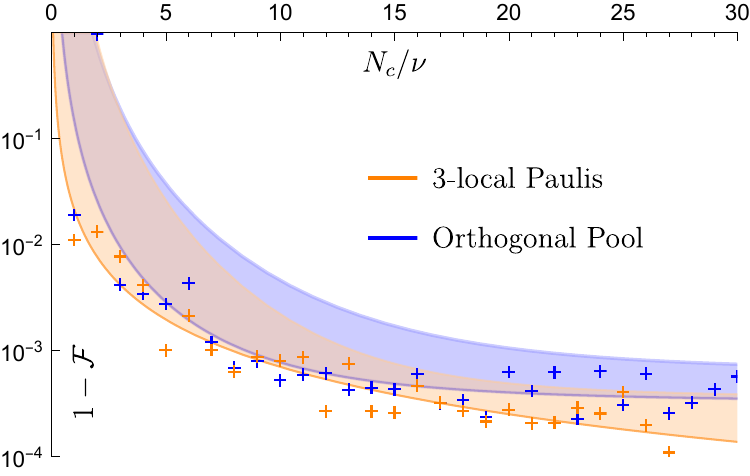}
		\caption{Performance of recompiling the parameters of a $2$-layer unitary on $14$ qubits. Plot shows the minimum infidelity reached over 3 runs of 20 iterations of root finding with the ratio of constraints to ansatz parameters $\nu$, for two different choices of operator pool with shot noise equivalent to taking $10^5$ shots. Blue for orthogonal pool and orange for the set of all 3-local Pauli strings. Fits shown are of the form $1-\mathcal{F}_{min}=a \left(\nc/\nu \right)^{-b} +c$.}
		\label{fig:operator_pool_comp}
	\end{centering}
\end{figure}

\subsection{Finding eigenstates via orthogonal operator pools}
Let us now apply our orthogonal constraints to finding eigenstates by finding roots. It is interesting to note
that it follows from our relation in \cref{lemma:orth_sum_squares} that the Newton step through the inverse Jacobian from \eqref{eq_jacobian} as $\jac^{-1} \underline{\covar}$ is guaranteed to represent a descent direction for the variance $\mathrm{Var}[\mathcal{H}]$ given that the gradient vector $\mathrm{grad} ( \lVert \covar \lVert^2 ) = \jac^\intercal \covarvec = \mathrm{grad} (\mathrm{Var}[\mathcal{H}])$.

Furthermore, we have written our problem as a least-squares minimisation of the constraints $f_k$ and thus the Gauss-Newton and the Levenberg-Marquardt approaches
can be interpreted straightforwardly: our root finding approach is equivalent to a non-linear least squares minimisation.

An advantage of this scheme is that we can randomly sample the constraints according to an importance sampling, i.e., the constraints are picked with a probability proportional to their magnitude $p_k = |\covar_k|^2$. We can efficiently upper bound these probabilities in an experiment as $p_k = \langle \phi_k  | \mathcal{H} | \psi \rangle =  |\langle k  |  U^\dagger \mathcal{H} | \psi \rangle|^2$. In particular, we run the quantum circuit $U^\dagger \mathcal{H}_a U$ and measure samples in the standard basis, whereby we obtain the binary string $k$ with probability $|\langle k  |  U^\dagger \mathcal{H}_a | \psi \rangle|^2$. It then follows that $p_k \leq \sum_a c_a |\langle k  |  U^\dagger \mathcal{H}_a | \psi \rangle|^2$.
There is of course no guarantee that these probabilities have structure, however, when performing energy minimisation first, the probabilities are more likely to be peaked around the lower energy basis vectors $|\phi_k\rangle$.

In case of finding eigenstates of a set of mutually commuting Hamiltonians $\mathcal{H}_a$ we can compute covariances as $\langle \phi_k| \mathcal{H}_a |\psi \rangle$ individually for all operators $\mathcal{H}_a$. If all such variances vanish then we are guaranteed that $\var[\mathcal{H}_a] = 0$ for all $a$. In Figure~\ref{fig:operator_pool_comp}) we compare the performance of root finding for two choices of operator pool on a parameter rediscovery problem (3-local Pauli strings and the orthogonal operator pool), showing that both pools give very similar performance.

\section{Properties of Covariances}

\subsection{Proof of Lemma~\ref{lemma:smooth}: Smooth covariance functions \label{sec:proof_smooth}}
Let us prove that the covariance functions are smooth functions of the parameters $\underline{\theta}$ of the
variational quantum state $|\psi(\underline{\theta})\rangle$. In particular let us
expand our expression from Eq.~\eqref{eq:param_covar_def} as
\begin{align} \nonumber
	\covar_k(\underline{\theta}) =& 	\langle \psi (\underline{\theta}) | O_k \mathcal{H} | \psi (\underline{\theta}) \rangle
	{-}
	\langle \psi (\underline{\theta}) | O_k  | \psi (\underline{\theta}) \rangle \langle \psi (\underline{\theta}) | \mathcal{H}  | \psi (\underline{\theta}) \rangle
	\\
	=& A_{re}(\underline{\theta}) + i A_{im}(\underline{\theta}) - B(\underline{\theta}) C(\underline{\theta}),
	\label{eq:covar_expected_herm}
\end{align}
which expression we have re-written in terms of 4 expectation values of 4 different Hermitian operators
via using the real and imaginary parts as in \cref{eq:covar_real} as
\begin{align}\nonumber
	A_{re} &:= \langle \tfrac{1}{2} \{ O_k, \mathcal{H} \}\rangle,
	\quad
	A_{im} := \langle - \tfrac{i}{2} [ O_k, \mathcal{H} ] \rangle,\\
	B &:= \langle  O_k \rangle,
	\quad
	C:= \langle  \mathcal{H} \rangle
	\label{eq:covar_hermitian_terms}
\end{align}
where we use $\langle \cdot \rangle$ to denote the expected value with respect to
the parametrised quantum state  $\psi (\underline{\theta})$ and we have dropped the
dependence on $\underline{\theta}$ for ease of notation.
Indeed above in \cref{eq:covar_expected_herm} all terms are expected values of Hermitian operators.

It suffices to show that the expectation value of any Hermitian observable $O \in \mathbb{C}^{d \times d}$
is a smooth function of the circuit parameters as
\begin{equation}
	\langle O \rangle(\underline{\theta}) = 	\langle \psi (\underline{\theta}) | O | \psi (\underline{\theta}) \rangle
	=
	\langle \underline{0} | U^\dagger (\underline{\theta}) O  U(\underline{\theta}) | \underline{0}  \rangle.
\end{equation}
Indeed the ansatz circuit is by definition (via \cref{eq:vari_gates}) a smooth mapping as $U(\underline{\theta}) \in \mathrm{SU}(2^N)$ and thus
$\langle O \rangle(\underline{\theta})$ is a smooth function of the parameters for any $O$.

\begin{figure}[tb]
	\centering
	\includegraphics[width=\linewidth]{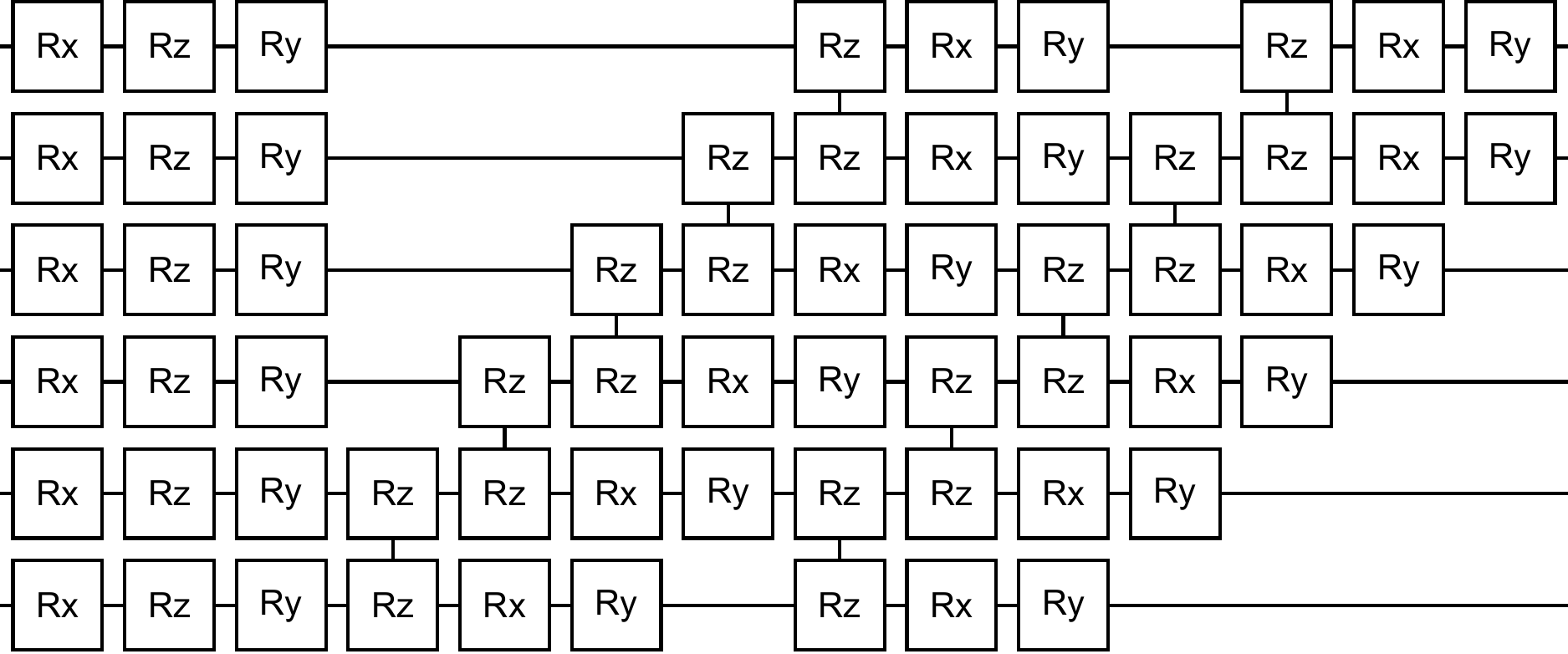}
	\caption{One layer of the Hardware Efficient Ansatz used in numerics for $6$ qubits}
	\label{fig:ansatz}
\end{figure}

\subsection{Proof of Corollary~\ref{cor:trig_poly}: Trigonometric polynomials}
Given the specific but pivotal scenario when every parametrised gate in the ansatz circuit
is a Pauli gate, as in the present work as illustrated in \cref{fig:ansatz}, ref.~\cite{koczor2020quantumAnalytic} established 
the following.
The expected value of any Hermitian observable is a trigonometric polynomial of the form
\begin{equation*}
	\langle O \rangle(\underline{\theta}) = \sum_{j=1}^{3^\nu} c_j T_j(\underline{\theta}),
\end{equation*}
where $c_j \in \mathbb{R}$ are real coefficients that depend on $O$ and $T_j(\underline{\theta})$
are trigonometric monomials as products
$T_j(\underline{\theta}) \in \prod_{n=1}^\nu \{ a(\theta_n), b(\theta_n), c(\theta_n)  \}$
of single-variate trigonometric functions~\footnote{
here the product of sets produces a set that contains all possible products of the elements
}, for example $b(\theta_1) \prod_{n=2}^\nu a(\theta_n)$.
The single variate functions are  	$a(\theta) :=  (1+\cos[\theta])/2$, 
$b(\theta) := 	 \sin[\theta]/2 $ and 
$c(\theta) := 	 (1-\cos[\theta])/2$. Given in \cref{eq:covar_hermitian_terms} all four terms
must be of this form, we obtain the expression for the covariances via \cref{eq:covar_expected_herm}
as
\begin{equation*}
	\covar_k(\underline{\theta}) 
	=
	\sum_{j=1}^{3^\nu} \tilde{c}_j T_j(\underline{\theta})
-	\sum_{j,l=1}^{3^\nu} c'_j c''_l T_j(\underline{\theta}) T_l(\underline{\theta})
=
	\sum_{j=1}^{3^{2\nu}} C_j \mathcal{T}_j(\underline{\theta})
\end{equation*}
where $\tilde{c}_j \in \mathbb{C}$ and $c'_j, c''_j \in \mathbb{R}$ are coefficients that
depend on the index $k$. Indeed, here $T_j(\underline{\theta}) T_l(\underline{\theta})$
are also trigonometric monomials and in the last equation we denoted these as
$\mathcal{T}_j  \in [\prod_{n=1}^\nu \{ a(\theta_n), b(\theta_n), c(\theta_n)  \}]^2$
with prefactors $C_j \in \mathbb{C}$.

\subsection{Experimentally estimating covariances\label{app:jacobian}}
Let us first compute covariances assuming the Hamiltonian $\mathcal{H} = \sum_{a=1}^r h_a \mathcal{H}_a $
is given in terms of Pauli strings $\mathcal{H}_a \in \{\mathrm{Id}_2, X, Y, Z\}^{\otimes N}$. The covariances
$\covar_k = \bilinear{O_k}{\mathcal{H}}{\psi} = A_{re}+ i A_{im} - B C$ are completely determined by the expected values from \cref{eq:covar_hermitian_terms}
of Hermitian operators. We can significantly simplify these when $O_k$ and $\mathcal{H}_a$ are Pauli strings
as
\begin{align}
	A_{re} &= \sum_{a=1}^r h_a \langle \tfrac{1}{2} \{ O_k, \mathcal{H}_a \}\rangle = \sum_{a=1}^r h_a \langle P_{ka} \rangle,\\
	A_{im} &= \sum_{a=1}^r h_a  \langle - \tfrac{i}{2} [ O_k, \mathcal{H}_a ] \rangle = \sum_{a=1}^r h_a \langle Q_{ka} \rangle,
\end{align} 
where indeed $P_{ka}, Q_{ka} \in \pm \{\mathrm{Id}_2, X, Y, Z\}^{\otimes N}$ are Hermitian Pauli strings given
any two Pauli strings $O_k$ and $\mathcal{H}_a$ either commute or anticommute and thus
\begin{equation*}
	\tfrac{1}{2} \{ O_k, \mathcal{H}_a \} {=} 
	\begin{cases}
		0, & \text{if $O_k$ and $\mathcal{H}_a$ }\\[-1mm]
		   & \text{anticommute},\\[2mm]
		 P_{ka} \in \pm \{\mathrm{Id}_2, X, Y, Z\}^{\otimes N} & \text{otherwise,}
	\end{cases}
\end{equation*}
and similarly for the imaginary  part
\begin{equation*}
	\tfrac{i}{2} [ O_k, \mathcal{H}_a ]  {=} 
	\begin{cases}
		0, & \text{if $O_k$ and $\mathcal{H}_a$ }		\\[-1mm]
		& \text{commute},\\[2mm]
		Q_{ka} \in \pm \{\mathrm{Id}_2, X, Y, Z\}^{\otimes N}, & \text{otherwise.}
	\end{cases}
\end{equation*}
Here the particular Pauli strings $P_{ka}$ and $Q_{ka}$  and their signs $\pm$ can be determined straightforwardly and efficiently
from the indexes $k, a \in \{ 0,1,2,3\}^{N}$ using the algebra of Pauli matrices, i.e., the Pauli group.
We therefore conclude that the covariances can be computed 
in terms of only expected values of Hermitian Pauli strings as
\begin{equation}\label{eq:covar_pauli_expected}
	\covar_k
	=
	\sum_{a=1}^r h_a \big( \langle P_{ka} \rangle + i \langle Q_{ka} \rangle 
	+
	 \langle \mathcal{H}_a \rangle \langle O_k \rangle
	 \big).
\end{equation}
We need to estimate overall $3 r$ expected values of Pauli strings to estimate a covariance $f_k$ given $P_{ka}=0$ when $Q_{ka} \neq 0$ and
vice versa.

\section{Properties and applications of the Jacobian \label{sec:prop_of_jacobian}}
\subsection{Computing the Jacobian} \label{sec:computing_jacobian}

We consider the covariances $\covar_k(\underline{\theta})$ with respect to Pauli strings $O_k \in \pool$
in our operator pool as defined in \cref{eq:param_covar_def} for a fixed Hamiltonian $\mathcal{H} = \sum_{a=1}^r h_a \mathcal{H}_a $. Recall that the Jacobian is defined in terms of the
partial derivatives $\jac_{kn}:= \partial_n \covar_k(\underline{\theta})$.  We can explicitly compute these derivatives
by recalling that the covariances can be expressed in terms of expected values of Pauli strings
via \cref{eq:covar_pauli_expected} as
\begin{align} \label{eq:jacobian_expected}
		\jac_{kn}
	&=
	\sum_{a=1}^r h_a\\
	&\times
	 \bigg( 
	    \frac{ \partial \langle P_{ka} \rangle }{\partial \theta_n}
	{+} 	i \frac{ \partial \langle Q_{ka} \rangle }{\partial \theta_n}
	{+}
	\langle O_k \rangle \frac{ \partial \langle \mathcal{H}_a  \rangle }{\partial \theta_n}
	{+}
	\langle \mathcal{H}_a \rangle \frac{ \partial \langle O_k  \rangle }{\partial \theta_n}
	\bigg),
	\nonumber
\end{align}
where $P_{ka}, Q_{ka}, \mathcal{H}_a$ and $O_k$ are Pauli strings.
Above we can use well established techniques for experimentally estimating derivatives
of general expected values for a variety of ansatz constructions and gatesets~\cite{endoHybridQuantumClassicalAlgorithms2021, cerezoVariationalQuantumAlgorithms2021a, bharti2021noisy}.
Furthermore, in \cref{sec:covar} we focus on the typical practical scenario when the ansatz circuit consists
of Pauli gates and thus we can use parameter-shift rules~\cite{paramshift} for computing derivatives---while generalisations in~\cite{PhysRevA.104.052417, Wierichs2022generalparameter} are also applicable---as linear combinations of two expected values as
\begin{equation*}
	\frac{ \partial \langle O  \rangle(\underline{\theta}) }{\partial \theta_n}
	=
\tfrac{1}{2}	\langle O  \rangle(\underline{\theta} + \underline{v}_n \pi/2)
	-
\tfrac{1}{2}	\langle O  \rangle(\underline{\theta}-\underline{v}_n\pi/2),
\end{equation*}
for any Hermitian observable $O$ where $\underline{v}_n$ is the $n^{th}$ standard Euclidean basis vector.
As such, we can compute all derivatives in \cref{eq:jacobian_expected} by applying parameter-shift rules.

\begin{figure}[tb]
	\begin{centering}
		\includegraphics[width=\linewidth]{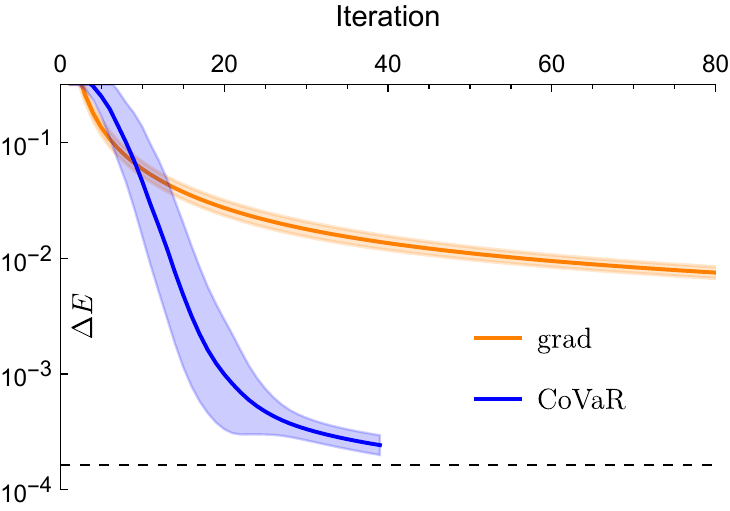}
		\caption{
			Comparison of the performance of \cv to gradient descent when starting from a point close to the ground state of the $10$-qubit spin chain, analogous to \cref{fig:gradient_comparison}(c) for recompilation. This initialisation was done by randomly perturbing the parameters of a state close to the ground state. Average initial fidelity to the ground state was $\mathcal{F}=80 \pm 4 \%$. Mean and standard deviation are shown for 20 runs of gradient descent and the 16/20 runs of \cv that converged to the ground state rather than one of the excited states. The ansatz can only approximate the
				ground state of the spin model, the grid line at $\Delta E \approx 1.6 \times 10^{-4}$ marks the minimum achieved energy as the limit of precision. One can only further improve $\Delta E$ by increasing the ansatz depth.
		}		\label{fig:spinchain_close}
	\end{centering}
\end{figure}

\subsection{Computing the Jacobian using classical shadows}

Let us now describe an explicit measurement protocol in the specific case when our Hamiltonian is local and the operator pool $\pool= \{O_k\}_{k=1}^{\nc}$ is also
local as in \cref{stat:locality}
and thus we can use classical shadows to determine a large number $M$ of local Pauli-string expected values.
We then compute the large Jacobian by estimating Pauli strings at different shifted circuit parameters via
the parameter-shift  rules for computing partial derivatives.

\begin{lemma}\label{lemma:jacobian_complexity}
Given an operator pool $\pool = \{O_k\}_{k=1}^{\nc}$ and a problem Hamiltonian
$\mathcal{H} = \sum_{a=1}^r h_a \mathcal{H}_a $ in terms of local Pauli strings $O_k$ and $\mathcal{H}_a$
and a variational circuit with $\nu$ parametrised Pauli gates,
we can determine a the Jacobian $\jac \in \mathbb{C}^{\nc \times \nu }$ by
applying the classical shadow procedure $2\nu +1 $ times at different circuit-parameter configurations
with each estimating $M \leq 3 r \nc$ Pauli expected values. In contrast determining a gradient vector
for gradient descent requires $2\nu$ applications of the classical shadow
procedure each with $M=r$.
\end{lemma}
\begin{proof}
First, at parameters $\underline{\theta}$ we determine overall $M=r+\nc$ expected
values as $\langle O_k \rangle $ and $\langle \mathcal{H}_a \rangle$ from \cref{eq:jacobian_expected} using a single application
of the classical shadow procedure.
Second, we estimate derivatives of expected values ($\frac{ \partial \langle \mathcal{H}_a  \rangle }{\partial \theta_n}$,
$\frac{ \partial \langle O_k  \rangle }{\partial \theta_n}$ and
either $\frac{ \partial \langle P_{ka} \rangle }{\partial \theta_n}$
or $i \frac{ \partial \langle Q_{ka} \rangle }{\partial \theta_n}$
) via the above parameter-sift rule. We determine all derivatives with respect to a fixed
parameter $\theta_n$  by estimating Pauli strings
at parameters $\underline{\theta} +\underline{v}_n\pi/2$ as
well as 
at parameters $\underline{\theta} -\underline{v}_n\pi/2$
by two applications of the classical shadow procedure each with $M \leq r + \nc +  r \nc $.
Thus determining all derivatives requires $2 \nu$ applications of the classical shadow.

\end{proof}

\subsection{Computing the variance gradient from the Jacobian \label{app:variance_grad}}
The gradient of the variance can be computed via \cref{lemma:variance} 
when the operator pool is $\pool = \{\mathcal{H}_a\}_{a=1}^r$ as
\begin{equation*}
	\partial_n \sigma^2 = \partial_n [ \sum_{a=1}^r h_a \covar_a ] = \sum_{k=1}^r h_a \partial_n \covar_a(\psi) 
	= [\jac]_{an} h_a = \jac^\intercal \underline{h}.
\end{equation*}
This confirms that the gradient of the variance is determined by our Jacobian.

We can also compute the gradient of the vector norm $\lVert \covar \rVert^2 = \sum_{k=1}^{\npool} | \covar_k |^2$.
When the operator pool is larger than the problem Hamiltonian terms then we compute the gradient of the norm of the covariance
vector
\begin{align*}
	\partial_n \lVert \covar \rVert^2
	&=
	\partial_n \sum_{k=1}^{\npool}  \covar_k \covar_k^*
	= \sum_{k=1}^{\npool}  [ (\partial_n \covar_k) \covar_k^*  +  \covar_k (\partial_n \covar_k^* )]\\
	&= \sum_{k=1}^{\npool}  [ \jac_{kn} \covar_k^*  +  \covar_k \jac_{kn}^*]
	 = 2 \mathrm{Re} [ \jac^\intercal \underline{\covar}^* ]
	 = 2\tilde{\jac}^\intercal \tilde{\covarvec},
\end{align*}
where in $ \tilde{\jac}$ and in $\tilde{\covarvec}$ we have stacked real and imaginary parts on top of each other.

\begin{figure}[tb]
	\begin{centering}
		\includegraphics[width=\linewidth]{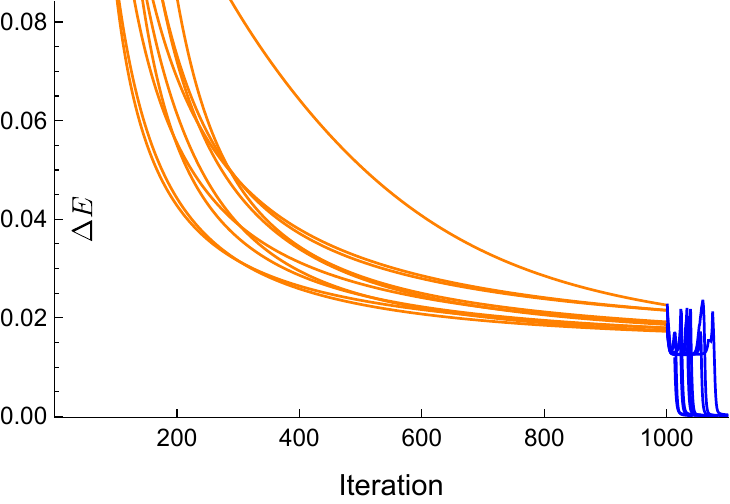}
		\caption{
			Ten examples of \cv escaping from practical local traps demonstrated on the same $10$-qubit spin-chain problem
			as in \cref{fig:spinchainLowestEigs}. Gradient descent (orange) for $1000$ adaptive size steps to reach local minima, followed by $100$ iterations of \cv (blue). Several steps of \cv were sometimes required before finding a set of operators that produced a step that escaped the minimum. Mean energy differences from the ground state are:  $\Delta E_{localmin} = 0.019$, $\Delta E_{\cv} =  1.8 \times 10^{-4}$. The last step of each run of gradient descent had an energy improvement below $2 \times 10^{-5}$ indicating being stuck in a local trap.
			Note that a single iteration of \cv does not necessarily decrease the energy given \cv is not an energy minimiser as is demonstrated here but is also visible in \cref{fig:spinchainLowestEigs}. Indeed, escaping from a local trap of the \emph{energy} surface with \cv
			may result in an intermediate increase of the energy.
		}
		\label{fig:localmintest}
	\end{centering}
\end{figure}

\subsection{Computing the update rule classically \label{app_newton_method}}

We need to first stack real and imaginary parts on top to solve for only real parameter upadtes. In particular, we want to solve the linear system of equations 
$ \jac  \Delta \underline{\theta}
  = \underline{\covar}$.
Here we estimate both $\jac$ and $\underline{\covar}$  with a quantum computer and we want to
compute the parameter update $\Delta \underline{\theta}$.
If we merely compute the pseudoinverse of $\jac$ and apply it as $\jac^{-1} \underline{\covar}$ then the resulting solution vector $\Delta \underline{\theta}$ is generally complex. However, we require that our parameter-update be real as $\Delta \underline{\theta} \in \mathbb{R}$
which we enforce via the following set of linear equations
\begin{equation*}
    \mathrm{Re}[\jac]  \Delta \underline{\theta}
  =
    \mathrm{Re}[\underline{\covar}],
    \quad \quad \text{and}
    \quad \quad
    \mathrm{Im}[\jac]  \Delta \underline{\theta}
  =
    \mathrm{Im}[\underline{\covar}].
\end{equation*}
We can simultaneously solve these systems of equations by stacking real and imaginary parts
$\tilde{\jac} = (\mathrm{Re}[\jac],\mathrm{Im}[\jac])^T$
and 
$\tilde{\underline{\covar}} = (\mathrm{Re}[\underline{\covar}],\mathrm{Im}[\underline{\covar}])^T$
on top of each other and this will guarantee that the solution is real $\tilde{\jac}^{-1} \tilde{\underline{\covar}} \in \mathbb{R}$
since $\tilde{\jac}$ and $\tilde{\underline{\covar}}$ are real matrices and vectors, respectivelty.

We invert the Jacobian via the usual damped, regularised inverse
\begin{equation*}
    \tilde{\jac}^{-1} := [\tilde{\jac}^\intercal \tilde{\jac} + \lambda \mathrm{Id}]^{-1} \tilde{\jac}^\intercal.
\end{equation*}
Here $\lambda$ is a regularisation parameter that we dynamically set by choosing $\lambda=0.0001 \times 2^i$ where $i$ is incremented from $0$ until the condition $\lVert \underline{\covar}(\underline{\theta}_t) \rVert < \lVert \underline{\covar}(\underline{\theta}_{t-1}) \rVert$ is met, i.e., the norm of the covariance vector has decreased from the previous iteration. If the step $\Delta \underline{\theta}$ to be taken is too large for any individual parameter is too large $|\Delta \theta_i|> 1$, the update step is rescaled to normalise this value to $1$. This prevents our algorithm from taking overly large steps. For the selection of constraints, $\nc$ constraints are selected randomly from the chosen operator pool at every iteration.

We implemented a linesearch algorithm
whereby we compute the update rule $\Delta \underline{\theta}$ and compute the value of
the vector norm $\lVert \underline{\covar}(\underline{\theta}) \rVert$ at parameter values
$\underline{\theta} = \underline{\theta}_{t} + \kappa \Delta \underline{\theta}$
in small increments in $\kappa$. However, it was observed empirically that a step close to the canonical choice $\kappa=-1$ was almost always chosen, so linesearch was not used in the numerics in this work.

\subsection{Existence of a shot noise-floor \label{app:shot_noise_floor}}
We extend our example in \cref{sec:large_pool} and consider noisy entries of the Jacobian and the covariance vector as $J_k + \epsilon_k$ and $\covar_k(\theta) + \eta_k$
for random variables $\epsilon_k$ and $\eta_k$ due to shot noise (and possibly other sources
of random errors). The solution to our
linear equation in \cref{eq:single_variable} is thus modified as
\begin{equation*}
	\Delta \theta \approx  -\frac{ \sum_k J_k \covar_k(\theta) } {  \sum_k J_k^2  }
	- \frac{ \sum_k J_k \eta_k + \covar_k(\theta) \epsilon_k } {  \sum_k J_k^2 },
\end{equation*}
where we dropped all higher order terms as products $\epsilon_k \eta_k$ as well as considered the approximation
$ [ \sum_k J_k^2 + \mathcal{O}(\epsilon_k) ]^{-1} = [ \sum_k J_k^2]^{-1} + \mathcal{O}(\epsilon_k / [\sum_k J_k^2]^2)$
where we also dropped the term $\mathcal{O}(\epsilon_k / [\sum_k J_k^2]^2)$.

Let us now consider the error propagation to the solution $\Delta \theta$ by considering the
linear error propagation formula
\begin{align*}
	\var [ \Delta \theta]
	 &\approx  
	 \frac{ \sum_k J_k^2 \var[\eta_k] +  \covar_k(\theta)^2 \var[\epsilon_k] } {  \sum_k J_k^2 }\\
	&\leq
	\mathcal{E}^2 [1 + \frac{  |\covar_k(\theta)|^2  } {  \sum_k J_k^2 }],
\end{align*}
where we have simplified the formula by bounding the variance of each $\var[ \eta_k] \leq \mathcal{E}^2$ and $\var[ \epsilon_k ] \leq \mathcal{E}^2$,
i.e., this is indeed the case in practice where we determine each covariance and Jacobian entry to the same guaranteed precision $\mathcal{E}$ using classical shadows. Indeed we see that as we approach an eigenstate $|\covar_k(\theta)|^2 \rightarrow 0$ and thus
the error that propagates into our solution is bounded by  $\var [ \Delta \theta] \leq \mathcal{E}^2$ by the worst-case error
of a \emph{single Jacobian/vector entry}.  This is a very powerful averaging of random errors: by increasing the number of constraints
we gain increasingly more information about the root, however, the random error that propagates into
our solution does not scale with the number of constraints.

\subsection{Time complexity of classically solving the linear system of equations\label{app:classical_comp_time}}

Let us now derive the time complexity of computing the regularised inverse as $(\tilde{\jac}^\intercal  \tilde{\jac} + \lambda R)^{-1} \tilde{\jac}^\intercal  \tilde{\covarvec}$
of the Jacobian $\tilde{\jac} \in \mathbb{R}^{ 2\nc \times \nu}$ in which we
have stacked real and imaginary parts on top of each other and $R$ is a regularisation matrix.
This computation can be broken up into four steps.

\begin{figure}[tb]
	\begin{centering}
		\includegraphics[width=\linewidth,trim={0 3.5cm 0 0.8cm},clip]{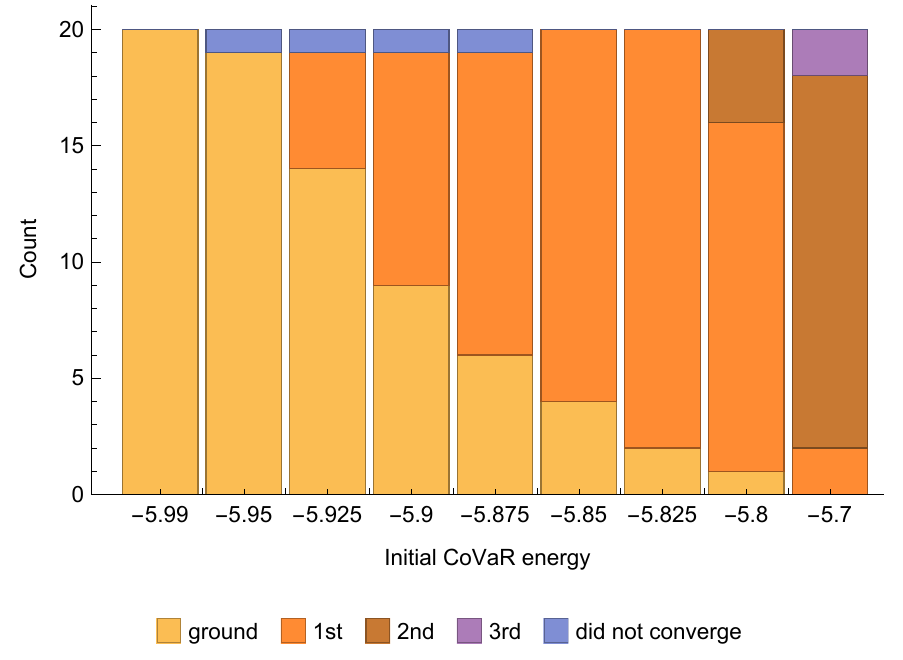}
		\caption{Bar chart showing the distributions of states which \cv converged to for the runs in \cref{fig:spinchain_groundstate} showing the probability of converging to the ground state with initial overlap (convergence being an energy difference of less than $10^{-3}$ in this case) as a function of its initialisation energy. Each bar corresponds to one point in \cref{fig:spinchain_groundstate}, but there are two additional points on the right of the figure.
		}
		\label{fig:groundstate_barchart}
	\end{centering}
\end{figure}

\begin{itemize}
	\item First, compute
	the square matrix $A := \tilde{\jac}^\intercal  \tilde{\jac}+ \lambda R$ as the product of two non-square matrices
	as
	$[A]_{mn}
	=\lambda R_{mn}+ \sum_{k=1}^{2\nc} [\tilde{\jac}]_{mk} [\tilde{\jac}]_{kn}$.
	Computing all $\nu^2$ entries of $A$ requires
	$2 \nu^2 \nc + \nu^2$ operations.
	
	\item Second, we compute the inverse of $A$
	which requires between $\mathcal{O}(\nu^{2.373})$ and $\mathcal{O}(\nu^{3})$ operations
	depending on the algorithm.

	\item Third, compute the matrix-vector product
	$\mathbf{v}:=\tilde{\jac}^\intercal \tilde{\covarvec}$ as
	$	[\tilde{\jac}^\intercal \tilde{\covarvec}]_n = \sum_{k=1}^{2\nc} [\tilde{\jac}_{nk}] \tilde{\covar}_k$,
	which requires overall $2 \nu \nc$ operations. 
	
	\item Finally, we compute the matrix-vector product $A \mathbf{v}$ which requires $\nu^2$ operations.
	
\end{itemize}

Given $\nc \gg \nu$, the overall computation time is dominated by the
first step as computing $A$ and thus the time complexity is
$t \in \mathcal{O}(\nu^2 \nc)$ which is merely
linear in the dominant dimension $\nc$.
We confirm this theoretical scaling in \cref{fig:classical_comp_fig} and conclude that
the absolute times are very reasonable, i.e., we can compute the update rule in a matter of minutes for up to a very large number of covariances $\nc = 10^6$. Of course, the computation can be heavily parallelised and can also be preformed in a distributed-memory model with negligible communication between nodes. As such, in principle one could straightforwardly use a very large number of covariances $\nc \approx 10^8$ that would still require reasonable classical computational resources.

\begin{figure*}[tb]
	\begin{centering}
		\includegraphics[width=\linewidth]{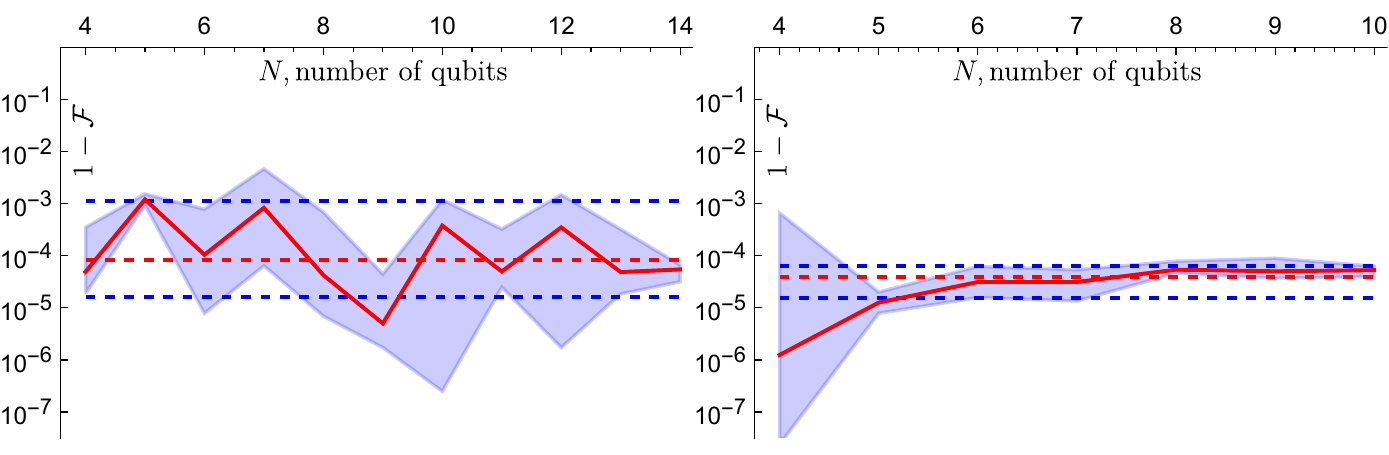}	
		\caption{
			Scaling results showing the infidelity with respect to the ground state of the state achieved after $20$ steps of \cv from an initial state of $\mathcal{F}_{init}$ overlap with the ground state.
				The red curve and blue shaded region shows the median and quartiles of $10$ runs of \cv for each $N$ and the dotted lines show the median and quartiles of the data as a whole. 
				(left) simulations for our recompilation problem from \cref{sec:applications_recomp} using a constant initial fidelity
				$\mathcal{F}_{init} = 50\%$ and a constant ansatz depth of 2 layers.
				(right) simulations for the spin chain problem using an increasing ansatz depth that can estimate the ground state
				to at least  an infidelity of $10^{-5}$. An initial overlap of $\mathcal{F}_{init} = 80\%$ was used to suppress convergence to excited states and only the runs that did converge to the ground state (and not to an excited state) are included in the present statistics.
				Compared to our recompilation simulations (left) using a constant ansatz depth,
				here the performance may appear to be decreasing slightly as we increase the number of qubits
				-- which would be explained by 
				the increasing depth (polylogarithmically in \cref{sec:large_pool}) of the ansatz circuit.
		}
		\label{fig:scaling}
	\end{centering}
\end{figure*}

\section{Further numerical simulations}
\subsection{Comparison with VQE for the spin-chain problem}

While in \cref{sec:application_spin} we focused on practical applications of \cv here we demonstrate a comparison
of convergence speed between gradient descent and \cvns.
In particular, in \cref{fig:spinchain_close} we compare the speed of convergence to the ground state of the spin chain,
using an identical ansatz and hyperparameters as those used for runs of \cv in \cref{fig:spinchainLowestEigs}.

\subsection{Demonstration of escaping local traps} \label{sec:localmintest}
As we noted in the main text, even when gradient descent is stuck
in a local trap due to the gradient of the energy surface vanishing, \cv can be used to navigate out from
such a trap given \cv is not an energy minimiser and may yield a non-zero step -- especially
that constraints that determine a \cv step are generated randomly at every iteration.
We numerically demonstrate this in \cref{fig:localmintest} using our spin-chain problem, with all parameters the same as those used for \cref{fig:spinchainLowestEigs}. The local traps were found using $1000$ iterations of gradient descent with an adaptive step size and achieved a mean energy difference from the ground state of $\Delta E = 0.019
$. \cv was then run for $100$ further iterations and was able to achieve a mean energy difference of $\Delta E = 1.8 \times 10^{-4}$.

\subsection{Convergence of root finding to eigenstates} \label{sec:groundstate_barchart}

In this section we further analyse the property of \cv that it converges to eigenstates
	that have a dominant contribution to the initial state as reported in \cref{fig:spinchain_groundstate}. For this reason, in \cref{fig:groundstate_barchart} we performed simulations
	of \cv optimisations with initial states of increasing expected energy starting near the  ground state.

Recall that the expected value of the energy can be written as $E = \langle \psi | \mathcal{H} | \psi \rangle = \sum_k E_k p_k$, where
$E_k$ are eigenenergies (eigenvalues) of $\mathcal{H}$ and $p_k$ are the probabilities (fidelity) that $|\psi\rangle$ is
 in the $k$-th eigenstate of $\mathcal{H}$.
Given a gapped Hamiltonian with energies $E_0 < E_1 < E_2 \dots$, quantum states with a low energy must necessarily
have a high probability (fidelity) to be in the lowest lying eigenstates. As such, an initial state that has expected energy $E = E_0 p_0 + E_1 p_1 + \dots$ close to the ground state energy as $E\approx E_0$
guarantees the high probability (fidelity) $p_0 \approx 1$. Indeed, \cv nearly always converges to
the ground state in \cref{fig:groundstate_barchart} (bars on the left) when $E\approx E_0$ and thus $p_0 \approx 1$.

On the other hand, starting \cv from initial states that have higher expected energies than $E_0$ causes \cv to
not always converge to the ground state due to the necessarily increased populations of excited states:  \cref{fig:groundstate_barchart}(bars in the middle) feature instances where $\cv$ does not converge to the ground state, however, it then nearly always converges to the first excited state.
For example, when no higher eigenstates are populated the probability of the first excited state is $p_1 = 1 - p_0$.
Interestingly, we find a nearly linear relationship between the probability of the eigenstate,
such as $p_0$, and the percentage when \cv converges to that state in \cref{fig:spinchain_groundstate} -- and note that this relationship would be exactly
linear for the case of fault-tolerant phase-estimation protocols.
As we keep increasing the energy, higher excited states start to contribute as in \cref{fig:groundstate_barchart}(bars on the right)
and thus \cv may also converge to those eigenstates.

\subsection{Scaling of performance} \label{sec:scaling}

We performed simulations to assess the scaling of the performance of \cvns, for both recompilation and spin-chain problems
and plot the results in \cref{fig:scaling}. For the recompilation problem we knew by construction the parameters $\underline{\theta}^\star$
of the ground state and the circuit depth was constant due to a fixed number $2$ of ansatz layers for all qubit counts $N$. 
For the spin-chain problem we first searched for the solution at each qubit number $N$
using natural gradient descent followed by \cvns. 
The number of ansatz layers was then increased until
a desired precision with respect to the ground state
was achieved (here an infidelity of $10^{-5}$)
thus obtaining a series of solution parameters $\underline{\theta}^\star$ at
every qubit count.

These solution parameters were perturbed to obtain initial states for $\cv$ with a desired initial overlap with the ground state
($50\%$ for recompilation and $80\%$ for the spin chain to suppress convergence to excited states).
Of course, such states could equally well have been created through alternative initialisation methods,
including performing an initial period of gradient or natural gradient descent.
Finally, \cv was run for a fixed number $20$ of iterations and statistics of the final achieved fidelity are plotted 
in \cref{fig:scaling}.

\section{Details of numerical simulations}

We performed all statevector simulations using the open-source tools QuEST~\cite{quest} and its high level Mathematica based interface QuESTlink~\cite{QuESTlink}. Shot noise was simulated by adding Gaussian noise of standard deviation $1/\sqrt{N_s}$ to computed matrix and
vector elements.
Numerics for recompilation and the spin chain were both performed using a hardware efficient ansatz of the form in \cref{fig:ansatz}.

\subsection{Effect of Constraint Number on Performance} \label{sec:ncons_comp_details}

\begin{table}[tb]
	\centering 
	\begin{tabular}{l@{\hspace{5mm}}  l@{\hspace{5mm}}  c@{\hspace{5mm}}   c@{\hspace{5mm}} c }
		\\[-3mm]
		\hline\hline
		\\[-4.5mm]
		fitted data & type &	$a$   &  $b$  &   $c$		\\
		\hline \\[-4mm] 	
		\cref{fig:constraints_noise_comp} noise free & max  & $15.6$ & $3.00$ &  $7 {\times} 10^{-4}$\\
		& min  & $5.62$ & $3.23$ &  $ 0.0$\\
		\cref{fig:constraints_noise} shot noise & max  & $11.3$ & $3.76$ &  $3 {\times} 10^{-4}$\\
		& min  & $0.0259$ & $1.68$ &  $1 {\times} 10^{-4}$\\
		\cref{fig:constraints_noise} circuit noise & max  & $7.80$ & $ 3.09$ &  $2 {\times} 10^{-4}$\\
		& min & $0.35$ & $2.24$ &  $ 0.0$\\
		\cref{fig:operator_pool_comp} orth. pool & max  & $15.6$ & $2.99$ &  $ 7 {\times} 10^{-4}$\\
		& min & $0.15$ & $2.51$ &  $0.0$ 
		\\[0mm] \hline \hline
	\end{tabular} 
	\caption{\label{fits}
			In our numerical simulations we investigated the effect of increasing the number $\nc$ of
			covariances. We fit a function to the fidelity achieved by \cv of the form $1-\mathcal{F}=a \left(\nc/\nu \right)^{-b} +c$.
			Fits to the worst (max) and best (min) of three runs of \cv are reported.
	}
\end{table}

Data for \cref{fig:constraints_noise_comp} and \cref{fig:constraints_noise} was obtained by simulating a
14-qubit parameter rediscovery problem using 2 layers of the ansatz in \cref{fig:ansatz}. 
The initial states were initialised close to the solution by randomly perturbing the solution parameters
resulting in an initial average fidelity of $\mathcal{F}=46\pm7\%$.
Fits in these figures are of the form $1-\mathcal{F}_{min}=a \left(\nc/\nu \right)^{-b} +c$ 
and we report fitted parameters in \cref{fits}.
The orange line in \cref{fig:operator_pool_comp} is identical to the orange line in \cref{fig:constraints_noise}.

	The noisy simulations were performed assuming the following simplified noise model. 
	Recall that global depolarising noise is a relatively good approximation
	in complex quantum circuits and becomes near-exact for random circuits, refer 
	for rigorous bounds to ref.~\cite{dalzell2021random}. This error channel acts on any density matrix
	via the Kraus map $D(\rho) = F \rho + (1-F) \rho_{max}$,
	where $\rho_{max}$ is the maximally mixed state, i.e., white noise, and $F$ is the fidelity.
	The expected value of any traceless Hermitian operator $O$, such as
	Pauli strings as relevant in the present work, merely gets attenuated
	as $\langle O \rangle = \tr[O \rho] = F \langle O \rangle_{id}$ where $\langle O \rangle_{id}$ is the ideal, noiseless expected value.

In practice this error model does not capture more subtle physical processes that corrupt the expected value measurement. Nevertheless, it was shown in ref~\cite{koczor2021dominant} that nearly all typical error models used in practice admit the decomposition $F \rho + (1-F) \rho_{err}$
where $\rho_{err} \approx \rho_{max}$ is an error density matrix that we do not expect to be exactly the maximally mixed 
state, albeit in practice it is relatively close to the maximally mixed state via its vanishing commutator norm from ref~\cite{koczor2021dominant}. In order to go beyond global depolarisation, but without resorting to computationally infeasible explicit noise simulations, we compute the noisy expected value as $\langle O_k \rangle = \tr[O_k \rho] = F \langle O_k \rangle_{id} + (1-F) \tr[O_k \rho_{err}]$ by approximating the term $\tr[O_k \rho_{err}] \sim \mathcal{N}(0,\sigma^2)$ using random Gaussian numbers. Here we set $\sigma = 0.01$ which is determined by the distance of $\rho_{err}$ from the maximally mixed state which we simulate with random Gaussian numbers that are unique to each observable indexed by $k$. Furthermore, we choose the fidelity $F=0.9 \approx  (1-\epsilon_1)^{\nu_1}(1-\epsilon_2)^{\nu_2}$, such that it approximates the performance of a typical, state-of-the-art experimental device with two-qubit error rates $\epsilon_2 = 0.001$ and single-qubit error rates 4-times smaller $\epsilon_1 = 0.25 \epsilon_2$ given in our circuit we have $\nu_1 = 196$ and $\nu_2 = 52$ single- and two-qubit gates, respectively.

\subsection{Recompilation problem in Fig.~\ref{fig:gradient_comparison}} \label{sec:recomp_details}
The numerics for recompilation in \cref{fig:gradient_comparison}(c) were done on a $10$-qubit parameter rediscovery problem for two layers of HEA ($\nu=88$). Gradient descent was preformed with a learning rate $\eta=0.1$ for both VQE and V-VQE.

\subsection{Spin-chain simulations} \label{sec:spinchain_details}
The Hamiltonian in \cref{eq:spin-ring} was used with parameters $J=0.1$ and $c_i$ chosen randomly between $-1$ and $1$.
For the spin-chain simulations in \cref{fig:spinchainLowestEigs}, Imaginary Time Evolution was used from a random initialisation until an energy of $E=-5.9$ was reached, with parameters $\underline{\theta}_{imag}$. These parameters were then disturbed by $|\Delta \theta_k| \le 0.05$ to produce 7 low energy states. \cv was then run from these initial states for 40 iterations. Imaginary Time Evolution was also continued from $\underline{\theta}_{imag}$ until convergence and reached an energy difference to the ground state of $\Delta E = 0.012$ compared to the $4 \times 10^{-4}$ of CoVaR.


\begin{thebibliography}{91}%
	\makeatletter
	\providecommand \@ifxundefined [1]{%
		\@ifx{#1\undefined}
	}%
	\providecommand \@ifnum [1]{%
		\ifnum #1\expandafter \@firstoftwo
		\else \expandafter \@secondoftwo
		\fi
	}%
	\providecommand \@ifx [1]{%
		\ifx #1\expandafter \@firstoftwo
		\else \expandafter \@secondoftwo
		\fi
	}%
	\providecommand \natexlab [1]{#1}%
	\providecommand \enquote  [1]{``#1''}%
	\providecommand \bibnamefont  [1]{#1}%
	\providecommand \bibfnamefont [1]{#1}%
	\providecommand \citenamefont [1]{#1}%
	\providecommand \href@noop [0]{\@secondoftwo}%
	\providecommand \href [0]{\begingroup \@sanitize@url \@href}%
	\providecommand \@href[1]{\@@startlink{#1}\@@href}%
	\providecommand \@@href[1]{\endgroup#1\@@endlink}%
	\providecommand \@sanitize@url [0]{\catcode `\\12\catcode `\$12\catcode
		`\&12\catcode `\#12\catcode `\^12\catcode `\_12\catcode `\%12\relax}%
	\providecommand \@@startlink[1]{}%
	\providecommand \@@endlink[0]{}%
	\providecommand \url  [0]{\begingroup\@sanitize@url \@url }%
	\providecommand \@url [1]{\endgroup\@href {#1}{\urlprefix }}%
	\providecommand \urlprefix  [0]{URL }%
	\providecommand \Eprint [0]{\href }%
	\providecommand \doibase [0]{https://doi.org/}%
	\providecommand \selectlanguage [0]{\@gobble}%
	\providecommand \bibinfo  [0]{\@secondoftwo}%
	\providecommand \bibfield  [0]{\@secondoftwo}%
	\providecommand \translation [1]{[#1]}%
	\providecommand \BibitemOpen [0]{}%
	\providecommand \bibitemStop [0]{}%
	\providecommand \bibitemNoStop [0]{.\EOS\space}%
	\providecommand \EOS [0]{\spacefactor3000\relax}%
	\providecommand \BibitemShut  [1]{\csname bibitem#1\endcsname}%
	\let\auto@bib@innerbib\@empty
	\bibitem [{\citenamefont {Arute}\ \emph {et~al.}(2019)\citenamefont {Arute},
		\citenamefont {Arya}, \citenamefont {Babbush}, \citenamefont {Bacon},
		\citenamefont {Bardin}, \citenamefont {Barends}, \citenamefont {Biswas},
		\citenamefont {Boixo}, \citenamefont {Brandao}, \citenamefont {Buell},
		\citenamefont {Burkett}, \citenamefont {Chen}, \citenamefont {Chen},
		\citenamefont {Chiaro}, \citenamefont {Collins}, \citenamefont {Courtney},
		\citenamefont {Dunsworth}, \citenamefont {Farhi}, \citenamefont {Foxen},
		\citenamefont {Fowler}, \citenamefont {Gidney}, \citenamefont {Giustina},
		\citenamefont {Graff}, \citenamefont {Guerin}, \citenamefont {Habegger},
		\citenamefont {Harrigan}, \citenamefont {Hartmann}, \citenamefont {Ho},
		\citenamefont {Hoffmann}, \citenamefont {Huang}, \citenamefont {Humble},
		\citenamefont {Isakov}, \citenamefont {Jeffrey}, \citenamefont {Jiang},
		\citenamefont {Kafri}, \citenamefont {Kechedzhi}, \citenamefont {Kelly},
		\citenamefont {Klimov}, \citenamefont {Knysh}, \citenamefont {Korotkov},
		\citenamefont {Kostritsa}, \citenamefont {Landhuis}, \citenamefont
		{Lindmark}, \citenamefont {Lucero}, \citenamefont {Lyakh}, \citenamefont
		{Mandr{\`a}}, \citenamefont {McClean}, \citenamefont {McEwen}, \citenamefont
		{Megrant}, \citenamefont {Mi}, \citenamefont {Michielsen}, \citenamefont
		{Mohseni}, \citenamefont {Mutus}, \citenamefont {Naaman}, \citenamefont
		{Neeley}, \citenamefont {Neill}, \citenamefont {Niu}, \citenamefont {Ostby},
		\citenamefont {Petukhov}, \citenamefont {Platt}, \citenamefont {Quintana},
		\citenamefont {Rieffel}, \citenamefont {Roushan}, \citenamefont {Rubin},
		\citenamefont {Sank}, \citenamefont {Satzinger}, \citenamefont {Smelyanskiy},
		\citenamefont {Sung}, \citenamefont {Trevithick}, \citenamefont
		{Vainsencher}, \citenamefont {Villalonga}, \citenamefont {White},
		\citenamefont {Yao}, \citenamefont {Yeh}, \citenamefont {Zalcman},
		\citenamefont {Neven},\ and\ \citenamefont
		{Martinis}}]{aruteQuantumSupremacyUsing2019}%
	\BibitemOpen
	\bibfield  {author} {\bibinfo {author} {\bibfnamefont {F.}~\bibnamefont
			{Arute}}, \bibinfo {author} {\bibfnamefont {K.}~\bibnamefont {Arya}},
		\bibinfo {author} {\bibfnamefont {R.}~\bibnamefont {Babbush}}, \bibinfo
		{author} {\bibfnamefont {D.}~\bibnamefont {Bacon}}, \bibinfo {author}
		{\bibfnamefont {J.~C.}\ \bibnamefont {Bardin}}, \bibinfo {author}
		{\bibfnamefont {R.}~\bibnamefont {Barends}}, \bibinfo {author} {\bibfnamefont
			{R.}~\bibnamefont {Biswas}}, \bibinfo {author} {\bibfnamefont
			{S.}~\bibnamefont {Boixo}}, \bibinfo {author} {\bibfnamefont {F.~G. S.~L.}\
			\bibnamefont {Brandao}}, \bibinfo {author} {\bibfnamefont {D.~A.}\
			\bibnamefont {Buell}}, \bibinfo {author} {\bibfnamefont {B.}~\bibnamefont
			{Burkett}}, \bibinfo {author} {\bibfnamefont {Y.}~\bibnamefont {Chen}},
		\bibinfo {author} {\bibfnamefont {Z.}~\bibnamefont {Chen}}, \bibinfo {author}
		{\bibfnamefont {B.}~\bibnamefont {Chiaro}}, \bibinfo {author} {\bibfnamefont
			{R.}~\bibnamefont {Collins}}, \bibinfo {author} {\bibfnamefont
			{W.}~\bibnamefont {Courtney}}, \bibinfo {author} {\bibfnamefont
			{A.}~\bibnamefont {Dunsworth}}, \bibinfo {author} {\bibfnamefont
			{E.}~\bibnamefont {Farhi}}, \bibinfo {author} {\bibfnamefont
			{B.}~\bibnamefont {Foxen}}, \bibinfo {author} {\bibfnamefont
			{A.}~\bibnamefont {Fowler}}, \bibinfo {author} {\bibfnamefont
			{C.}~\bibnamefont {Gidney}}, \bibinfo {author} {\bibfnamefont
			{M.}~\bibnamefont {Giustina}}, \bibinfo {author} {\bibfnamefont
			{R.}~\bibnamefont {Graff}}, \bibinfo {author} {\bibfnamefont
			{K.}~\bibnamefont {Guerin}}, \bibinfo {author} {\bibfnamefont
			{S.}~\bibnamefont {Habegger}}, \bibinfo {author} {\bibfnamefont {M.~P.}\
			\bibnamefont {Harrigan}}, \bibinfo {author} {\bibfnamefont {M.~J.}\
			\bibnamefont {Hartmann}}, \bibinfo {author} {\bibfnamefont {A.}~\bibnamefont
			{Ho}}, \bibinfo {author} {\bibfnamefont {M.}~\bibnamefont {Hoffmann}},
		\bibinfo {author} {\bibfnamefont {T.}~\bibnamefont {Huang}}, \bibinfo
		{author} {\bibfnamefont {T.~S.}\ \bibnamefont {Humble}}, \bibinfo {author}
		{\bibfnamefont {S.~V.}\ \bibnamefont {Isakov}}, \bibinfo {author}
		{\bibfnamefont {E.}~\bibnamefont {Jeffrey}}, \bibinfo {author} {\bibfnamefont
			{Z.}~\bibnamefont {Jiang}}, \bibinfo {author} {\bibfnamefont
			{D.}~\bibnamefont {Kafri}}, \bibinfo {author} {\bibfnamefont
			{K.}~\bibnamefont {Kechedzhi}}, \bibinfo {author} {\bibfnamefont
			{J.}~\bibnamefont {Kelly}}, \bibinfo {author} {\bibfnamefont {P.~V.}\
			\bibnamefont {Klimov}}, \bibinfo {author} {\bibfnamefont {S.}~\bibnamefont
			{Knysh}}, \bibinfo {author} {\bibfnamefont {A.}~\bibnamefont {Korotkov}},
		\bibinfo {author} {\bibfnamefont {F.}~\bibnamefont {Kostritsa}}, \bibinfo
		{author} {\bibfnamefont {D.}~\bibnamefont {Landhuis}}, \bibinfo {author}
		{\bibfnamefont {M.}~\bibnamefont {Lindmark}}, \bibinfo {author}
		{\bibfnamefont {E.}~\bibnamefont {Lucero}}, \bibinfo {author} {\bibfnamefont
			{D.}~\bibnamefont {Lyakh}}, \bibinfo {author} {\bibfnamefont
			{S.}~\bibnamefont {Mandr{\`a}}}, \bibinfo {author} {\bibfnamefont {J.~R.}\
			\bibnamefont {McClean}}, \bibinfo {author} {\bibfnamefont {M.}~\bibnamefont
			{McEwen}}, \bibinfo {author} {\bibfnamefont {A.}~\bibnamefont {Megrant}},
		\bibinfo {author} {\bibfnamefont {X.}~\bibnamefont {Mi}}, \bibinfo {author}
		{\bibfnamefont {K.}~\bibnamefont {Michielsen}}, \bibinfo {author}
		{\bibfnamefont {M.}~\bibnamefont {Mohseni}}, \bibinfo {author} {\bibfnamefont
			{J.}~\bibnamefont {Mutus}}, \bibinfo {author} {\bibfnamefont
			{O.}~\bibnamefont {Naaman}}, \bibinfo {author} {\bibfnamefont
			{M.}~\bibnamefont {Neeley}}, \bibinfo {author} {\bibfnamefont
			{C.}~\bibnamefont {Neill}}, \bibinfo {author} {\bibfnamefont {M.~Y.}\
			\bibnamefont {Niu}}, \bibinfo {author} {\bibfnamefont {E.}~\bibnamefont
			{Ostby}}, \bibinfo {author} {\bibfnamefont {A.}~\bibnamefont {Petukhov}},
		\bibinfo {author} {\bibfnamefont {J.~C.}\ \bibnamefont {Platt}}, \bibinfo
		{author} {\bibfnamefont {C.}~\bibnamefont {Quintana}}, \bibinfo {author}
		{\bibfnamefont {E.~G.}\ \bibnamefont {Rieffel}}, \bibinfo {author}
		{\bibfnamefont {P.}~\bibnamefont {Roushan}}, \bibinfo {author} {\bibfnamefont
			{N.~C.}\ \bibnamefont {Rubin}}, \bibinfo {author} {\bibfnamefont
			{D.}~\bibnamefont {Sank}}, \bibinfo {author} {\bibfnamefont {K.~J.}\
			\bibnamefont {Satzinger}}, \bibinfo {author} {\bibfnamefont {V.}~\bibnamefont
			{Smelyanskiy}}, \bibinfo {author} {\bibfnamefont {K.~J.}\ \bibnamefont
			{Sung}}, \bibinfo {author} {\bibfnamefont {M.~D.}\ \bibnamefont
			{Trevithick}}, \bibinfo {author} {\bibfnamefont {A.}~\bibnamefont
			{Vainsencher}}, \bibinfo {author} {\bibfnamefont {B.}~\bibnamefont
			{Villalonga}}, \bibinfo {author} {\bibfnamefont {T.}~\bibnamefont {White}},
		\bibinfo {author} {\bibfnamefont {Z.~J.}\ \bibnamefont {Yao}}, \bibinfo
		{author} {\bibfnamefont {P.}~\bibnamefont {Yeh}}, \bibinfo {author}
		{\bibfnamefont {A.}~\bibnamefont {Zalcman}}, \bibinfo {author} {\bibfnamefont
			{H.}~\bibnamefont {Neven}},\ and\ \bibinfo {author} {\bibfnamefont {J.~M.}\
			\bibnamefont {Martinis}},\ }\bibfield  {title} {\bibinfo {title} {Quantum
			supremacy using a programmable superconducting processor},\ }\href
	{https://doi.org/10.1038/s41586-019-1666-5} {\bibfield  {journal} {\bibinfo
			{journal} {Nature}\ }\textbf {\bibinfo {volume} {574}},\ \bibinfo {pages}
		{505} (\bibinfo {year} {2019})}\BibitemShut {NoStop}%
	\bibitem [{\citenamefont {Zhong}\ \emph {et~al.}(2021)\citenamefont {Zhong},
		\citenamefont {Deng}, \citenamefont {Qin}, \citenamefont {Wang},
		\citenamefont {Chen}, \citenamefont {Peng}, \citenamefont {Luo},
		\citenamefont {Wu}, \citenamefont {Gong}, \citenamefont {Su}, \citenamefont
		{Hu}, \citenamefont {Hu}, \citenamefont {Yang}, \citenamefont {Zhang},
		\citenamefont {Li}, \citenamefont {Li}, \citenamefont {Jiang}, \citenamefont
		{Gan}, \citenamefont {Yang}, \citenamefont {You}, \citenamefont {Wang},
		\citenamefont {Li}, \citenamefont {Liu}, \citenamefont {Renema},
		\citenamefont {Lu},\ and\ \citenamefont
		{Pan}}]{zhongPhaseProgrammableGaussianBoson2021}%
	\BibitemOpen
	\bibfield  {author} {\bibinfo {author} {\bibfnamefont {H.-S.}\ \bibnamefont
			{Zhong}}, \bibinfo {author} {\bibfnamefont {Y.-H.}\ \bibnamefont {Deng}},
		\bibinfo {author} {\bibfnamefont {J.}~\bibnamefont {Qin}}, \bibinfo {author}
		{\bibfnamefont {H.}~\bibnamefont {Wang}}, \bibinfo {author} {\bibfnamefont
			{M.-C.}\ \bibnamefont {Chen}}, \bibinfo {author} {\bibfnamefont {L.-C.}\
			\bibnamefont {Peng}}, \bibinfo {author} {\bibfnamefont {Y.-H.}\ \bibnamefont
			{Luo}}, \bibinfo {author} {\bibfnamefont {D.}~\bibnamefont {Wu}}, \bibinfo
		{author} {\bibfnamefont {S.-Q.}\ \bibnamefont {Gong}}, \bibinfo {author}
		{\bibfnamefont {H.}~\bibnamefont {Su}}, \bibinfo {author} {\bibfnamefont
			{Y.}~\bibnamefont {Hu}}, \bibinfo {author} {\bibfnamefont {P.}~\bibnamefont
			{Hu}}, \bibinfo {author} {\bibfnamefont {X.-Y.}\ \bibnamefont {Yang}},
		\bibinfo {author} {\bibfnamefont {W.-J.}\ \bibnamefont {Zhang}}, \bibinfo
		{author} {\bibfnamefont {H.}~\bibnamefont {Li}}, \bibinfo {author}
		{\bibfnamefont {Y.}~\bibnamefont {Li}}, \bibinfo {author} {\bibfnamefont
			{X.}~\bibnamefont {Jiang}}, \bibinfo {author} {\bibfnamefont
			{L.}~\bibnamefont {Gan}}, \bibinfo {author} {\bibfnamefont {G.}~\bibnamefont
			{Yang}}, \bibinfo {author} {\bibfnamefont {L.}~\bibnamefont {You}}, \bibinfo
		{author} {\bibfnamefont {Z.}~\bibnamefont {Wang}}, \bibinfo {author}
		{\bibfnamefont {L.}~\bibnamefont {Li}}, \bibinfo {author} {\bibfnamefont
			{N.-L.}\ \bibnamefont {Liu}}, \bibinfo {author} {\bibfnamefont {J.~J.}\
			\bibnamefont {Renema}}, \bibinfo {author} {\bibfnamefont {C.-Y.}\
			\bibnamefont {Lu}},\ and\ \bibinfo {author} {\bibfnamefont {J.-W.}\
			\bibnamefont {Pan}},\ }\bibfield  {title} {\bibinfo {title}
		{Phase-{{Programmable Gaussian Boson Sampling Using Stimulated Squeezed
					Light}}},\ }\href {https://doi.org/10.1103/PhysRevLett.127.180502} {\bibfield
		{journal} {\bibinfo  {journal} {Physical Review Letters}\ }\textbf {\bibinfo
			{volume} {127}},\ \bibinfo {pages} {180502} (\bibinfo {year}
		{2021})}\BibitemShut {NoStop}%
	\bibitem [{\citenamefont {Wu}\ \emph {et~al.}(2021)\citenamefont {Wu},
		\citenamefont {Bao}, \citenamefont {Cao}, \citenamefont {Chen}, \citenamefont
		{Chen}, \citenamefont {Chen}, \citenamefont {Chung}, \citenamefont {Deng},
		\citenamefont {Du}, \citenamefont {Fan}, \citenamefont {Gong}, \citenamefont
		{Guo}, \citenamefont {Guo}, \citenamefont {Guo}, \citenamefont {Han},
		\citenamefont {Hong}, \citenamefont {Huang}, \citenamefont {Huo},
		\citenamefont {Li}, \citenamefont {Li}, \citenamefont {Li}, \citenamefont
		{Li}, \citenamefont {Liang}, \citenamefont {Lin}, \citenamefont {Lin},
		\citenamefont {Qian}, \citenamefont {Qiao}, \citenamefont {Rong},
		\citenamefont {Su}, \citenamefont {Sun}, \citenamefont {Wang}, \citenamefont
		{Wang}, \citenamefont {Wu}, \citenamefont {Xu}, \citenamefont {Yan},
		\citenamefont {Yang}, \citenamefont {Yang}, \citenamefont {Ye}, \citenamefont
		{Yin}, \citenamefont {Ying}, \citenamefont {Yu}, \citenamefont {Zha},
		\citenamefont {Zhang}, \citenamefont {Zhang}, \citenamefont {Zhang},
		\citenamefont {Zhang}, \citenamefont {Zhao}, \citenamefont {Zhao},
		\citenamefont {Zhou}, \citenamefont {Zhu}, \citenamefont {Lu}, \citenamefont
		{Peng}, \citenamefont {Zhu},\ and\ \citenamefont
		{Pan}}]{wuStrongQuantumComputational2021}%
	\BibitemOpen
	\bibfield  {author} {\bibinfo {author} {\bibfnamefont {Y.}~\bibnamefont
			{Wu}}, \bibinfo {author} {\bibfnamefont {W.-S.}\ \bibnamefont {Bao}},
		\bibinfo {author} {\bibfnamefont {S.}~\bibnamefont {Cao}}, \bibinfo {author}
		{\bibfnamefont {F.}~\bibnamefont {Chen}}, \bibinfo {author} {\bibfnamefont
			{M.-C.}\ \bibnamefont {Chen}}, \bibinfo {author} {\bibfnamefont
			{X.}~\bibnamefont {Chen}}, \bibinfo {author} {\bibfnamefont {T.-H.}\
			\bibnamefont {Chung}}, \bibinfo {author} {\bibfnamefont {H.}~\bibnamefont
			{Deng}}, \bibinfo {author} {\bibfnamefont {Y.}~\bibnamefont {Du}}, \bibinfo
		{author} {\bibfnamefont {D.}~\bibnamefont {Fan}}, \bibinfo {author}
		{\bibfnamefont {M.}~\bibnamefont {Gong}}, \bibinfo {author} {\bibfnamefont
			{C.}~\bibnamefont {Guo}}, \bibinfo {author} {\bibfnamefont {C.}~\bibnamefont
			{Guo}}, \bibinfo {author} {\bibfnamefont {S.}~\bibnamefont {Guo}}, \bibinfo
		{author} {\bibfnamefont {L.}~\bibnamefont {Han}}, \bibinfo {author}
		{\bibfnamefont {L.}~\bibnamefont {Hong}}, \bibinfo {author} {\bibfnamefont
			{H.-L.}\ \bibnamefont {Huang}}, \bibinfo {author} {\bibfnamefont {Y.-H.}\
			\bibnamefont {Huo}}, \bibinfo {author} {\bibfnamefont {L.}~\bibnamefont
			{Li}}, \bibinfo {author} {\bibfnamefont {N.}~\bibnamefont {Li}}, \bibinfo
		{author} {\bibfnamefont {S.}~\bibnamefont {Li}}, \bibinfo {author}
		{\bibfnamefont {Y.}~\bibnamefont {Li}}, \bibinfo {author} {\bibfnamefont
			{F.}~\bibnamefont {Liang}}, \bibinfo {author} {\bibfnamefont
			{C.}~\bibnamefont {Lin}}, \bibinfo {author} {\bibfnamefont {J.}~\bibnamefont
			{Lin}}, \bibinfo {author} {\bibfnamefont {H.}~\bibnamefont {Qian}}, \bibinfo
		{author} {\bibfnamefont {D.}~\bibnamefont {Qiao}}, \bibinfo {author}
		{\bibfnamefont {H.}~\bibnamefont {Rong}}, \bibinfo {author} {\bibfnamefont
			{H.}~\bibnamefont {Su}}, \bibinfo {author} {\bibfnamefont {L.}~\bibnamefont
			{Sun}}, \bibinfo {author} {\bibfnamefont {L.}~\bibnamefont {Wang}}, \bibinfo
		{author} {\bibfnamefont {S.}~\bibnamefont {Wang}}, \bibinfo {author}
		{\bibfnamefont {D.}~\bibnamefont {Wu}}, \bibinfo {author} {\bibfnamefont
			{Y.}~\bibnamefont {Xu}}, \bibinfo {author} {\bibfnamefont {K.}~\bibnamefont
			{Yan}}, \bibinfo {author} {\bibfnamefont {W.}~\bibnamefont {Yang}}, \bibinfo
		{author} {\bibfnamefont {Y.}~\bibnamefont {Yang}}, \bibinfo {author}
		{\bibfnamefont {Y.}~\bibnamefont {Ye}}, \bibinfo {author} {\bibfnamefont
			{J.}~\bibnamefont {Yin}}, \bibinfo {author} {\bibfnamefont {C.}~\bibnamefont
			{Ying}}, \bibinfo {author} {\bibfnamefont {J.}~\bibnamefont {Yu}}, \bibinfo
		{author} {\bibfnamefont {C.}~\bibnamefont {Zha}}, \bibinfo {author}
		{\bibfnamefont {C.}~\bibnamefont {Zhang}}, \bibinfo {author} {\bibfnamefont
			{H.}~\bibnamefont {Zhang}}, \bibinfo {author} {\bibfnamefont
			{K.}~\bibnamefont {Zhang}}, \bibinfo {author} {\bibfnamefont
			{Y.}~\bibnamefont {Zhang}}, \bibinfo {author} {\bibfnamefont
			{H.}~\bibnamefont {Zhao}}, \bibinfo {author} {\bibfnamefont {Y.}~\bibnamefont
			{Zhao}}, \bibinfo {author} {\bibfnamefont {L.}~\bibnamefont {Zhou}}, \bibinfo
		{author} {\bibfnamefont {Q.}~\bibnamefont {Zhu}}, \bibinfo {author}
		{\bibfnamefont {C.-Y.}\ \bibnamefont {Lu}}, \bibinfo {author} {\bibfnamefont
			{C.-Z.}\ \bibnamefont {Peng}}, \bibinfo {author} {\bibfnamefont
			{X.}~\bibnamefont {Zhu}},\ and\ \bibinfo {author} {\bibfnamefont {J.-W.}\
			\bibnamefont {Pan}},\ }\bibfield  {title} {\bibinfo {title} {Strong {{Quantum
					Computational Advantage Using}} a {{Superconducting Quantum Processor}}},\
	}\href {https://doi.org/10.1103/PhysRevLett.127.180501} {\bibfield  {journal}
		{\bibinfo  {journal} {Physical Review Letters}\ }\textbf {\bibinfo {volume}
			{127}},\ \bibinfo {pages} {180501} (\bibinfo {year} {2021})}\BibitemShut
	{NoStop}%
	\bibitem [{\citenamefont {Ebadi}\ \emph {et~al.}(2021)\citenamefont {Ebadi},
		\citenamefont {Wang}, \citenamefont {Levine}, \citenamefont {Keesling},
		\citenamefont {Semeghini}, \citenamefont {Omran}, \citenamefont {Bluvstein},
		\citenamefont {Samajdar}, \citenamefont {Pichler}, \citenamefont {Ho},
		\citenamefont {Choi}, \citenamefont {Sachdev}, \citenamefont {Greiner},
		\citenamefont {Vuleti{\'c}},\ and\ \citenamefont
		{Lukin}}]{ebadiQuantumPhasesMatter2021}%
	\BibitemOpen
	\bibfield  {author} {\bibinfo {author} {\bibfnamefont {S.}~\bibnamefont
			{Ebadi}}, \bibinfo {author} {\bibfnamefont {T.~T.}\ \bibnamefont {Wang}},
		\bibinfo {author} {\bibfnamefont {H.}~\bibnamefont {Levine}}, \bibinfo
		{author} {\bibfnamefont {A.}~\bibnamefont {Keesling}}, \bibinfo {author}
		{\bibfnamefont {G.}~\bibnamefont {Semeghini}}, \bibinfo {author}
		{\bibfnamefont {A.}~\bibnamefont {Omran}}, \bibinfo {author} {\bibfnamefont
			{D.}~\bibnamefont {Bluvstein}}, \bibinfo {author} {\bibfnamefont
			{R.}~\bibnamefont {Samajdar}}, \bibinfo {author} {\bibfnamefont
			{H.}~\bibnamefont {Pichler}}, \bibinfo {author} {\bibfnamefont {W.~W.}\
			\bibnamefont {Ho}}, \bibinfo {author} {\bibfnamefont {S.}~\bibnamefont
			{Choi}}, \bibinfo {author} {\bibfnamefont {S.}~\bibnamefont {Sachdev}},
		\bibinfo {author} {\bibfnamefont {M.}~\bibnamefont {Greiner}}, \bibinfo
		{author} {\bibfnamefont {V.}~\bibnamefont {Vuleti{\'c}}},\ and\ \bibinfo
		{author} {\bibfnamefont {M.~D.}\ \bibnamefont {Lukin}},\ }\bibfield  {title}
	{\bibinfo {title} {Quantum phases of matter on a 256-atom programmable
			quantum simulator},\ }\href {https://doi.org/10.1038/s41586-021-03582-4}
	{\bibfield  {journal} {\bibinfo  {journal} {Nature}\ }\textbf {\bibinfo
			{volume} {595}},\ \bibinfo {pages} {227} (\bibinfo {year}
		{2021})}\BibitemShut {NoStop}%
	\bibitem [{\citenamefont {Gong}\ \emph {et~al.}(2021)\citenamefont {Gong},
		\citenamefont {Wang}, \citenamefont {Zha}, \citenamefont {Chen},
		\citenamefont {Huang}, \citenamefont {Wu}, \citenamefont {Zhu}, \citenamefont
		{Zhao}, \citenamefont {Li}, \citenamefont {Guo}, \citenamefont {Qian},
		\citenamefont {Ye}, \citenamefont {Chen}, \citenamefont {Ying}, \citenamefont
		{Yu}, \citenamefont {Fan}, \citenamefont {Wu}, \citenamefont {Su},
		\citenamefont {Deng}, \citenamefont {Rong}, \citenamefont {Zhang},
		\citenamefont {Cao}, \citenamefont {Lin}, \citenamefont {Xu}, \citenamefont
		{Sun}, \citenamefont {Guo}, \citenamefont {Li}, \citenamefont {Liang},
		\citenamefont {Bastidas}, \citenamefont {Nemoto}, \citenamefont {Munro},
		\citenamefont {Huo}, \citenamefont {Lu}, \citenamefont {Peng}, \citenamefont
		{Zhu},\ and\ \citenamefont {Pan}}]{gongQuantumWalksProgrammable2021a}%
	\BibitemOpen
	\bibfield  {author} {\bibinfo {author} {\bibfnamefont {M.}~\bibnamefont
			{Gong}}, \bibinfo {author} {\bibfnamefont {S.}~\bibnamefont {Wang}}, \bibinfo
		{author} {\bibfnamefont {C.}~\bibnamefont {Zha}}, \bibinfo {author}
		{\bibfnamefont {M.-C.}\ \bibnamefont {Chen}}, \bibinfo {author}
		{\bibfnamefont {H.-L.}\ \bibnamefont {Huang}}, \bibinfo {author}
		{\bibfnamefont {Y.}~\bibnamefont {Wu}}, \bibinfo {author} {\bibfnamefont
			{Q.}~\bibnamefont {Zhu}}, \bibinfo {author} {\bibfnamefont {Y.}~\bibnamefont
			{Zhao}}, \bibinfo {author} {\bibfnamefont {S.}~\bibnamefont {Li}}, \bibinfo
		{author} {\bibfnamefont {S.}~\bibnamefont {Guo}}, \bibinfo {author}
		{\bibfnamefont {H.}~\bibnamefont {Qian}}, \bibinfo {author} {\bibfnamefont
			{Y.}~\bibnamefont {Ye}}, \bibinfo {author} {\bibfnamefont {F.}~\bibnamefont
			{Chen}}, \bibinfo {author} {\bibfnamefont {C.}~\bibnamefont {Ying}}, \bibinfo
		{author} {\bibfnamefont {J.}~\bibnamefont {Yu}}, \bibinfo {author}
		{\bibfnamefont {D.}~\bibnamefont {Fan}}, \bibinfo {author} {\bibfnamefont
			{D.}~\bibnamefont {Wu}}, \bibinfo {author} {\bibfnamefont {H.}~\bibnamefont
			{Su}}, \bibinfo {author} {\bibfnamefont {H.}~\bibnamefont {Deng}}, \bibinfo
		{author} {\bibfnamefont {H.}~\bibnamefont {Rong}}, \bibinfo {author}
		{\bibfnamefont {K.}~\bibnamefont {Zhang}}, \bibinfo {author} {\bibfnamefont
			{S.}~\bibnamefont {Cao}}, \bibinfo {author} {\bibfnamefont {J.}~\bibnamefont
			{Lin}}, \bibinfo {author} {\bibfnamefont {Y.}~\bibnamefont {Xu}}, \bibinfo
		{author} {\bibfnamefont {L.}~\bibnamefont {Sun}}, \bibinfo {author}
		{\bibfnamefont {C.}~\bibnamefont {Guo}}, \bibinfo {author} {\bibfnamefont
			{N.}~\bibnamefont {Li}}, \bibinfo {author} {\bibfnamefont {F.}~\bibnamefont
			{Liang}}, \bibinfo {author} {\bibfnamefont {V.~M.}\ \bibnamefont {Bastidas}},
		\bibinfo {author} {\bibfnamefont {K.}~\bibnamefont {Nemoto}}, \bibinfo
		{author} {\bibfnamefont {W.~J.}\ \bibnamefont {Munro}}, \bibinfo {author}
		{\bibfnamefont {Y.-H.}\ \bibnamefont {Huo}}, \bibinfo {author} {\bibfnamefont
			{C.-Y.}\ \bibnamefont {Lu}}, \bibinfo {author} {\bibfnamefont {C.-Z.}\
			\bibnamefont {Peng}}, \bibinfo {author} {\bibfnamefont {X.}~\bibnamefont
			{Zhu}},\ and\ \bibinfo {author} {\bibfnamefont {J.-W.}\ \bibnamefont {Pan}},\
	}\bibfield  {title} {\bibinfo {title} {Quantum walks on a programmable
			two-dimensional 62-qubit superconducting processor},\ }\href
	{https://doi.org/10.1126/science.abg7812} {\bibfield  {journal} {\bibinfo
			{journal} {Science}\ }\textbf {\bibinfo {volume} {372}},\ \bibinfo {pages}
		{948} (\bibinfo {year} {2021})}\BibitemShut {NoStop}%
	\bibitem [{\citenamefont {Preskill}(2018)}]{preskillQuantumComputingNISQ2018a}%
	\BibitemOpen
	\bibfield  {author} {\bibinfo {author} {\bibfnamefont {J.}~\bibnamefont
			{Preskill}},\ }\bibfield  {title} {\bibinfo {title} {Quantum {{Computing}} in
			the {{NISQ}} era and beyond},\ }\href
	{https://doi.org/10.22331/q-2018-08-06-79} {\bibfield  {journal} {\bibinfo
			{journal} {Quantum}\ }\textbf {\bibinfo {volume} {2}},\ \bibinfo {pages} {79}
		(\bibinfo {year} {2018})}\BibitemShut {NoStop}%
	\bibitem [{\citenamefont {Farhi}\ \emph {et~al.}(2014)\citenamefont {Farhi},
		\citenamefont {Goldstone},\ and\ \citenamefont {Gutmann}}]{farhi2014quantum}%
	\BibitemOpen
	\bibfield  {author} {\bibinfo {author} {\bibfnamefont {E.}~\bibnamefont
			{Farhi}}, \bibinfo {author} {\bibfnamefont {J.}~\bibnamefont {Goldstone}},\
		and\ \bibinfo {author} {\bibfnamefont {S.}~\bibnamefont {Gutmann}},\
	}\bibfield  {title} {\bibinfo {title} {A quantum approximate optimization
			algorithm},\ }\href@noop {} {\bibfield  {journal} {\bibinfo  {journal} {arXiv
				preprint arXiv:1411.4028}\ } (\bibinfo {year} {2014})}\BibitemShut {NoStop}%
	\bibitem [{\citenamefont {Peruzzo}\ \emph {et~al.}(2014)\citenamefont
		{Peruzzo}, \citenamefont {McClean}, \citenamefont {Shadbolt}, \citenamefont
		{Yung}, \citenamefont {Zhou}, \citenamefont {Love}, \citenamefont
		{Aspuru-Guzik},\ and\ \citenamefont {O’brien}}]{peruzzo2014variational}%
	\BibitemOpen
	\bibfield  {author} {\bibinfo {author} {\bibfnamefont {A.}~\bibnamefont
			{Peruzzo}}, \bibinfo {author} {\bibfnamefont {J.}~\bibnamefont {McClean}},
		\bibinfo {author} {\bibfnamefont {P.}~\bibnamefont {Shadbolt}}, \bibinfo
		{author} {\bibfnamefont {M.-H.}\ \bibnamefont {Yung}}, \bibinfo {author}
		{\bibfnamefont {X.-Q.}\ \bibnamefont {Zhou}}, \bibinfo {author}
		{\bibfnamefont {P.~J.}\ \bibnamefont {Love}}, \bibinfo {author}
		{\bibfnamefont {A.}~\bibnamefont {Aspuru-Guzik}},\ and\ \bibinfo {author}
		{\bibfnamefont {J.~L.}\ \bibnamefont {O’brien}},\ }\bibfield  {title}
	{\bibinfo {title} {A variational eigenvalue solver on a photonic quantum
			processor},\ }\href@noop {} {\bibfield  {journal} {\bibinfo  {journal}
			{Nature Communications}\ }\textbf {\bibinfo {volume} {5}},\ \bibinfo {pages}
		{4213} (\bibinfo {year} {2014})}\BibitemShut {NoStop}%
	\bibitem [{\citenamefont {Endo}\ \emph {et~al.}(2021)\citenamefont {Endo},
		\citenamefont {Cai}, \citenamefont {Benjamin},\ and\ \citenamefont
		{Yuan}}]{endoHybridQuantumClassicalAlgorithms2021}%
	\BibitemOpen
	\bibfield  {author} {\bibinfo {author} {\bibfnamefont {S.}~\bibnamefont
			{Endo}}, \bibinfo {author} {\bibfnamefont {Z.}~\bibnamefont {Cai}}, \bibinfo
		{author} {\bibfnamefont {S.~C.}\ \bibnamefont {Benjamin}},\ and\ \bibinfo
		{author} {\bibfnamefont {X.}~\bibnamefont {Yuan}},\ }\bibfield  {title}
	{\bibinfo {title} {Hybrid {{Quantum-Classical Algorithms}} and {{Quantum
					Error Mitigation}}},\ }\href {https://doi.org/10.7566/JPSJ.90.032001}
	{\bibfield  {journal} {\bibinfo  {journal} {Journal of the Physical Society
				of Japan}\ }\textbf {\bibinfo {volume} {90}},\ \bibinfo {pages} {032001}
		(\bibinfo {year} {2021})}\BibitemShut {NoStop}%
	\bibitem [{\citenamefont {Cerezo}\ \emph
		{et~al.}(2021{\natexlab{a}})\citenamefont {Cerezo}, \citenamefont
		{Arrasmith}, \citenamefont {Babbush}, \citenamefont {Benjamin}, \citenamefont
		{Endo}, \citenamefont {Fujii}, \citenamefont {McClean}, \citenamefont
		{Mitarai}, \citenamefont {Yuan}, \citenamefont {Cincio},\ and\ \citenamefont
		{Coles}}]{cerezoVariationalQuantumAlgorithms2021a}%
	\BibitemOpen
	\bibfield  {author} {\bibinfo {author} {\bibfnamefont {M.}~\bibnamefont
			{Cerezo}}, \bibinfo {author} {\bibfnamefont {A.}~\bibnamefont {Arrasmith}},
		\bibinfo {author} {\bibfnamefont {R.}~\bibnamefont {Babbush}}, \bibinfo
		{author} {\bibfnamefont {S.~C.}\ \bibnamefont {Benjamin}}, \bibinfo {author}
		{\bibfnamefont {S.}~\bibnamefont {Endo}}, \bibinfo {author} {\bibfnamefont
			{K.}~\bibnamefont {Fujii}}, \bibinfo {author} {\bibfnamefont {J.~R.}\
			\bibnamefont {McClean}}, \bibinfo {author} {\bibfnamefont {K.}~\bibnamefont
			{Mitarai}}, \bibinfo {author} {\bibfnamefont {X.}~\bibnamefont {Yuan}},
		\bibinfo {author} {\bibfnamefont {L.}~\bibnamefont {Cincio}},\ and\ \bibinfo
		{author} {\bibfnamefont {P.~J.}\ \bibnamefont {Coles}},\ }\bibfield  {title}
	{\bibinfo {title} {Variational quantum algorithms},\ }\href
	{https://doi.org/10.1038/s42254-021-00348-9} {\bibfield  {journal} {\bibinfo
			{journal} {Nature Reviews Physics}\ }\textbf {\bibinfo {volume} {3}},\
		\bibinfo {pages} {625} (\bibinfo {year} {2021}{\natexlab{a}})}\BibitemShut
	{NoStop}%
	\bibitem [{\citenamefont {Bharti}\ \emph {et~al.}(2022)\citenamefont {Bharti},
		\citenamefont {Cervera-Lierta}, \citenamefont {Kyaw}, \citenamefont {Haug},
		\citenamefont {Alperin-Lea}, \citenamefont {Anand}, \citenamefont {Degroote},
		\citenamefont {Heimonen}, \citenamefont {Kottmann}, \citenamefont {Menke},
		\citenamefont {Mok}, \citenamefont {Sim}, \citenamefont {Kwek},\ and\
		\citenamefont {Aspuru-Guzik}}]{bharti2021noisy}%
	\BibitemOpen
	\bibfield  {author} {\bibinfo {author} {\bibfnamefont {K.}~\bibnamefont
			{Bharti}}, \bibinfo {author} {\bibfnamefont {A.}~\bibnamefont
			{Cervera-Lierta}}, \bibinfo {author} {\bibfnamefont {T.~H.}\ \bibnamefont
			{Kyaw}}, \bibinfo {author} {\bibfnamefont {T.}~\bibnamefont {Haug}}, \bibinfo
		{author} {\bibfnamefont {S.}~\bibnamefont {Alperin-Lea}}, \bibinfo {author}
		{\bibfnamefont {A.}~\bibnamefont {Anand}}, \bibinfo {author} {\bibfnamefont
			{M.}~\bibnamefont {Degroote}}, \bibinfo {author} {\bibfnamefont
			{H.}~\bibnamefont {Heimonen}}, \bibinfo {author} {\bibfnamefont {J.~S.}\
			\bibnamefont {Kottmann}}, \bibinfo {author} {\bibfnamefont {T.}~\bibnamefont
			{Menke}}, \bibinfo {author} {\bibfnamefont {W.-K.}\ \bibnamefont {Mok}},
		\bibinfo {author} {\bibfnamefont {S.}~\bibnamefont {Sim}}, \bibinfo {author}
		{\bibfnamefont {L.-C.}\ \bibnamefont {Kwek}},\ and\ \bibinfo {author}
		{\bibfnamefont {A.}~\bibnamefont {Aspuru-Guzik}},\ }\bibfield  {title}
	{\bibinfo {title} {{Noisy intermediate-scale quantum algorithms}},\ }\href
	{https://doi.org/10.1103/RevModPhys.94.015004} {\bibfield  {journal}
		{\bibinfo  {journal} {Rev. Mod. Phys.}\ }\textbf {\bibinfo {volume} {94}},\
		\bibinfo {pages} {015004} (\bibinfo {year} {2022})}\BibitemShut {NoStop}%
	\bibitem [{\citenamefont {Bittel}\ and\ \citenamefont
		{Kliesch}(2021{\natexlab{a}})}]{PhysRevLett.127.120502}%
	\BibitemOpen
	\bibfield  {author} {\bibinfo {author} {\bibfnamefont {L.}~\bibnamefont
			{Bittel}}\ and\ \bibinfo {author} {\bibfnamefont {M.}~\bibnamefont
			{Kliesch}},\ }\bibfield  {title} {\bibinfo {title} {Training variational
			quantum algorithms is np-hard},\ }\href
	{https://doi.org/10.1103/PhysRevLett.127.120502} {\bibfield  {journal}
		{\bibinfo  {journal} {Phys. Rev. Lett.}\ }\textbf {\bibinfo {volume} {127}},\
		\bibinfo {pages} {120502} (\bibinfo {year} {2021}{\natexlab{a}})}\BibitemShut
	{NoStop}%
	\bibitem [{\citenamefont {McClean}\ \emph {et~al.}(2018)\citenamefont
		{McClean}, \citenamefont {Boixo}, \citenamefont {Smelyanskiy}, \citenamefont
		{Babbush},\ and\ \citenamefont {Neven}}]{mcclean2018barren}%
	\BibitemOpen
	\bibfield  {author} {\bibinfo {author} {\bibfnamefont {J.~R.}\ \bibnamefont
			{McClean}}, \bibinfo {author} {\bibfnamefont {S.}~\bibnamefont {Boixo}},
		\bibinfo {author} {\bibfnamefont {V.~N.}\ \bibnamefont {Smelyanskiy}},
		\bibinfo {author} {\bibfnamefont {R.}~\bibnamefont {Babbush}},\ and\ \bibinfo
		{author} {\bibfnamefont {H.}~\bibnamefont {Neven}},\ }\bibfield  {title}
	{\bibinfo {title} {{Barren plateaus in quantum neural network training
				landscapes}},\ }\href@noop {} {\bibfield  {journal} {\bibinfo  {journal}
			{Nature Communications}\ }\textbf {\bibinfo {volume} {9}},\ \bibinfo {pages}
		{4812} (\bibinfo {year} {2018})}\BibitemShut {NoStop}%
	\bibitem [{\citenamefont {Larocca}\ \emph {et~al.}(2021)\citenamefont
		{Larocca}, \citenamefont {Ju}, \citenamefont {Garc{\'\i}a-Mart{\'\i}n},
		\citenamefont {Coles},\ and\ \citenamefont
		{Cerezo}}]{laroccaTheoryOverparametrizationQuantum2021}%
	\BibitemOpen
	\bibfield  {author} {\bibinfo {author} {\bibfnamefont {M.}~\bibnamefont
			{Larocca}}, \bibinfo {author} {\bibfnamefont {N.}~\bibnamefont {Ju}},
		\bibinfo {author} {\bibfnamefont {D.}~\bibnamefont
			{Garc{\'\i}a-Mart{\'\i}n}}, \bibinfo {author} {\bibfnamefont {P.~J.}\
			\bibnamefont {Coles}},\ and\ \bibinfo {author} {\bibfnamefont
			{M.}~\bibnamefont {Cerezo}},\ }\bibfield  {title} {\bibinfo {title} {Theory
			of overparametrization in quantum neural networks},\ }\href@noop {}
	{\bibfield  {journal} {\bibinfo  {journal} {arXiv preprint arXiv:2109.11676}\
		} (\bibinfo {year} {2021})}\BibitemShut {NoStop}%
	\bibitem [{\citenamefont {van Straaten}\ and\ \citenamefont
		{Koczor}(2021)}]{van2020measurement}%
	\BibitemOpen
	\bibfield  {author} {\bibinfo {author} {\bibfnamefont {B.}~\bibnamefont {van
				Straaten}}\ and\ \bibinfo {author} {\bibfnamefont {B.}~\bibnamefont
			{Koczor}},\ }\bibfield  {title} {\bibinfo {title} {{Measurement Cost of
				Metric-Aware Variational Quantum Algorithms}},\ }\href
	{https://doi.org/10.1103/PRXQuantum.2.030324} {\bibfield  {journal} {\bibinfo
			{journal} {PRX Quantum}\ }\textbf {\bibinfo {volume} {2}},\ \bibinfo {pages}
		{030324} (\bibinfo {year} {2021})}\BibitemShut {NoStop}%
	\bibitem [{\citenamefont {Huang}\ \emph {et~al.}(2020)\citenamefont {Huang},
		\citenamefont {Kueng},\ and\ \citenamefont {Preskill}}]{classical_shadows}%
	\BibitemOpen
	\bibfield  {author} {\bibinfo {author} {\bibfnamefont {H.-Y.}\ \bibnamefont
			{Huang}}, \bibinfo {author} {\bibfnamefont {R.}~\bibnamefont {Kueng}},\ and\
		\bibinfo {author} {\bibfnamefont {J.}~\bibnamefont {Preskill}},\ }\bibfield
	{title} {\bibinfo {title} {Predicting many properties of a quantum system
			from very few measurements},\ }\href
	{https://doi.org/10.1038/s41567-020-0932-7} {\bibfield  {journal} {\bibinfo
			{journal} {Nature Physics}\ }\textbf {\bibinfo {volume} {16}},\ \bibinfo
		{pages} {1050} (\bibinfo {year} {2020})}\BibitemShut {NoStop}%
	\bibitem [{\citenamefont {Anschuetz}\ and\ \citenamefont
		{Kiani}(2022)}]{anschuetz2022beyond}%
	\BibitemOpen
	\bibfield  {author} {\bibinfo {author} {\bibfnamefont {E.~R.}\ \bibnamefont
			{Anschuetz}}\ and\ \bibinfo {author} {\bibfnamefont {B.~T.}\ \bibnamefont
			{Kiani}},\ }\bibfield  {title} {\bibinfo {title} {{Beyond Barren Plateaus:
				Quantum Variational Algorithms Are Swamped With Traps}},\ }\href@noop {}
	{\bibfield  {journal} {\bibinfo  {journal} {arXiv preprint arXiv:2205.05786}\
		} (\bibinfo {year} {2022})}\BibitemShut {NoStop}%
	\bibitem [{\citenamefont {Cerezo}\ \emph
		{et~al.}(2021{\natexlab{b}})\citenamefont {Cerezo}, \citenamefont {Sone},
		\citenamefont {Volkoff}, \citenamefont {Cincio},\ and\ \citenamefont
		{Coles}}]{cerezoCostFunctionDependent2021a}%
	\BibitemOpen
	\bibfield  {author} {\bibinfo {author} {\bibfnamefont {M.}~\bibnamefont
			{Cerezo}}, \bibinfo {author} {\bibfnamefont {A.}~\bibnamefont {Sone}},
		\bibinfo {author} {\bibfnamefont {T.}~\bibnamefont {Volkoff}}, \bibinfo
		{author} {\bibfnamefont {L.}~\bibnamefont {Cincio}},\ and\ \bibinfo {author}
		{\bibfnamefont {P.~J.}\ \bibnamefont {Coles}},\ }\bibfield  {title} {\bibinfo
		{title} {Cost function dependent barren plateaus in shallow parametrized
			quantum circuits},\ }\href {https://doi.org/10.1038/s41467-021-21728-w}
	{\bibfield  {journal} {\bibinfo  {journal} {Nature Communications}\ }\textbf
		{\bibinfo {volume} {12}},\ \bibinfo {pages} {1791} (\bibinfo {year}
		{2021}{\natexlab{b}})}\BibitemShut {NoStop}%
	\bibitem [{\citenamefont {Ferraro}\ \emph {et~al.}(2005)\citenamefont
		{Ferraro}, \citenamefont {Olivares},\ and\ \citenamefont
		{Paris}}]{ferraro2005gaussian}%
	\BibitemOpen
	\bibfield  {author} {\bibinfo {author} {\bibfnamefont {A.}~\bibnamefont
			{Ferraro}}, \bibinfo {author} {\bibfnamefont {S.}~\bibnamefont {Olivares}},\
		and\ \bibinfo {author} {\bibfnamefont {M.~G.}\ \bibnamefont {Paris}},\
	}\bibfield  {title} {\bibinfo {title} {Gaussian states in continuous variable
			quantum information},\ }\href@noop {} {\bibfield  {journal} {\bibinfo
			{journal} {arXiv preprint quant-ph/0503237}\ } (\bibinfo {year}
		{2005})}\BibitemShut {NoStop}%
	\bibitem [{\citenamefont {Carmi}\ and\ \citenamefont
		{Cohen}(2018)}]{carmi2018significance}%
	\BibitemOpen
	\bibfield  {author} {\bibinfo {author} {\bibfnamefont {A.}~\bibnamefont
			{Carmi}}\ and\ \bibinfo {author} {\bibfnamefont {E.}~\bibnamefont {Cohen}},\
	}\bibfield  {title} {\bibinfo {title} {On the significance of the quantum
			mechanical covariance matrix},\ }\href@noop {} {\bibfield  {journal}
		{\bibinfo  {journal} {Entropy}\ }\textbf {\bibinfo {volume} {20}},\ \bibinfo
		{pages} {500} (\bibinfo {year} {2018})}\BibitemShut {NoStop}%
	\bibitem [{\citenamefont {Tripathi}\ \emph {et~al.}(2020)\citenamefont
		{Tripathi}, \citenamefont {Radhakrishnan},\ and\ \citenamefont
		{Byrnes}}]{tripathi2020covariance}%
	\BibitemOpen
	\bibfield  {author} {\bibinfo {author} {\bibfnamefont {V.}~\bibnamefont
			{Tripathi}}, \bibinfo {author} {\bibfnamefont {C.}~\bibnamefont
			{Radhakrishnan}},\ and\ \bibinfo {author} {\bibfnamefont {T.}~\bibnamefont
			{Byrnes}},\ }\bibfield  {title} {\bibinfo {title} {Covariance matrix
			entanglement criterion for an arbitrary set of operators},\ }\href@noop {}
	{\bibfield  {journal} {\bibinfo  {journal} {New Journal of Physics}\ }\textbf
		{\bibinfo {volume} {22}},\ \bibinfo {pages} {073055} (\bibinfo {year}
		{2020})}\BibitemShut {NoStop}%
	\bibitem [{\citenamefont {Hagan}\ and\ \citenamefont
		{Menhaj}(1994)}]{neuralnet}%
	\BibitemOpen
	\bibfield  {author} {\bibinfo {author} {\bibfnamefont {M.}~\bibnamefont
			{Hagan}}\ and\ \bibinfo {author} {\bibfnamefont {M.}~\bibnamefont {Menhaj}},\
	}\bibfield  {title} {\bibinfo {title} {{Training feedforward networks with
				the Marquardt algorithm}},\ }\href {https://doi.org/10.1109/72.329697}
	{\bibfield  {journal} {\bibinfo  {journal} {{IEEE Transactions on Neural
					Networks}}\ }\textbf {\bibinfo {volume} {5}},\ \bibinfo {pages} {989}
		(\bibinfo {year} {1994})}\BibitemShut {NoStop}%
	\bibitem [{\citenamefont {Demuth}\ \emph {et~al.}(2014)\citenamefont {Demuth},
		\citenamefont {Beale}, \citenamefont {De~Jess},\ and\ \citenamefont
		{Hagan}}]{demuth2014neural}%
	\BibitemOpen
	\bibfield  {author} {\bibinfo {author} {\bibfnamefont {H.~B.}\ \bibnamefont
			{Demuth}}, \bibinfo {author} {\bibfnamefont {M.~H.}\ \bibnamefont {Beale}},
		\bibinfo {author} {\bibfnamefont {O.}~\bibnamefont {De~Jess}},\ and\ \bibinfo
		{author} {\bibfnamefont {M.~T.}\ \bibnamefont {Hagan}},\ }\href@noop {}
	{\emph {\bibinfo {title} {Neural network design}}}\ (\bibinfo  {publisher}
	{Oklahoma State University, Stillwater},\ \bibinfo {year} {2014})\BibitemShut {NoStop}%
	\bibitem [{\citenamefont {Beale}\ \emph {et~al.}(2010)\citenamefont {Beale},
		\citenamefont {Hagan},\ and\ \citenamefont {Demuth}}]{beale2010neural}%
	\BibitemOpen
	\bibfield  {author} {\bibinfo {author} {\bibfnamefont {M.~H.}\ \bibnamefont
			{Beale}}, \bibinfo {author} {\bibfnamefont {M.~T.}\ \bibnamefont {Hagan}},\
		and\ \bibinfo {author} {\bibfnamefont {H.~B.}\ \bibnamefont {Demuth}},\
	}\bibfield  {title} {\bibinfo {title} {Neural network toolbox},\ }\href@noop
	{} {\bibfield  {journal} {\bibinfo  {journal} {User’s Guide, MathWorks}\
		}\textbf {\bibinfo {volume} {2}},\ \bibinfo {pages} {77} (\bibinfo {year}
		{2010})}\BibitemShut {NoStop}%
	\bibitem [{\citenamefont {Yu}\ and\ \citenamefont
		{Wilamowski}(2018)}]{yu2018levenberg}%
	\BibitemOpen
	\bibfield  {author} {\bibinfo {author} {\bibfnamefont {H.}~\bibnamefont
			{Yu}}\ and\ \bibinfo {author} {\bibfnamefont {B.~M.}\ \bibnamefont
			{Wilamowski}},\ }\bibfield  {title} {\bibinfo {title} {{Levenberg--marquardt
				training}},\ }in\ \href@noop {} {\emph {\bibinfo {booktitle} {Intelligent
				systems}}}\ (\bibinfo  {publisher} {CRC Press},\ \bibinfo {year} {2018})\
	pp.\ \bibinfo {pages} {12--1}\BibitemShut {NoStop}%
	\bibitem [{\citenamefont {Kandala}\ \emph {et~al.}(2017)\citenamefont
		{Kandala}, \citenamefont {Mezzacapo}, \citenamefont {Temme}, \citenamefont
		{Takita}, \citenamefont {Brink}, \citenamefont {Chow},\ and\ \citenamefont
		{Gambetta}}]{kandala2017hardware}%
	\BibitemOpen
	\bibfield  {author} {\bibinfo {author} {\bibfnamefont {A.}~\bibnamefont
			{Kandala}}, \bibinfo {author} {\bibfnamefont {A.}~\bibnamefont {Mezzacapo}},
		\bibinfo {author} {\bibfnamefont {K.}~\bibnamefont {Temme}}, \bibinfo
		{author} {\bibfnamefont {M.}~\bibnamefont {Takita}}, \bibinfo {author}
		{\bibfnamefont {M.}~\bibnamefont {Brink}}, \bibinfo {author} {\bibfnamefont
			{J.~M.}\ \bibnamefont {Chow}},\ and\ \bibinfo {author} {\bibfnamefont
			{J.~M.}\ \bibnamefont {Gambetta}},\ }\bibfield  {title} {\bibinfo {title}
		{{Hardware-efficient variational quantum eigensolver for small molecules and
				quantum magnets}},\ }\href@noop {} {\bibfield  {journal} {\bibinfo  {journal}
			{Nature}\ }\textbf {\bibinfo {volume} {549}},\ \bibinfo {pages} {242}
		(\bibinfo {year} {2017})}\BibitemShut {NoStop}%
	\bibitem [{\citenamefont {Taube}\ and\ \citenamefont
		{Bartlett}(2006)}]{taubeNewPerspectivesUnitary2006}%
	\BibitemOpen
	\bibfield  {author} {\bibinfo {author} {\bibfnamefont {A.~G.}\ \bibnamefont
			{Taube}}\ and\ \bibinfo {author} {\bibfnamefont {R.~J.}\ \bibnamefont
			{Bartlett}},\ }\bibfield  {title} {\bibinfo {title} {New perspectives on
			unitary coupled-cluster theory},\ }\href {https://doi.org/10.1002/qua.21198}
	{\bibfield  {journal} {\bibinfo  {journal} {International Journal of Quantum
				Chemistry}\ }\textbf {\bibinfo {volume} {106}},\ \bibinfo {pages} {3393}
		(\bibinfo {year} {2006})}\BibitemShut {NoStop}%
	\bibitem [{\citenamefont {Tilly}\ \emph {et~al.}(2022)\citenamefont {Tilly},
		\citenamefont {Chen}, \citenamefont {Cao}, \citenamefont {Picozzi},
		\citenamefont {Setia}, \citenamefont {Li}, \citenamefont {Grant},
		\citenamefont {Wossnig}, \citenamefont {Rungger}, \citenamefont {Booth} \emph
		{et~al.}}]{tillyVariationalQuantumEigensolver2021a}%
	\BibitemOpen
	\bibfield  {author} {\bibinfo {author} {\bibfnamefont {J.}~\bibnamefont
			{Tilly}}, \bibinfo {author} {\bibfnamefont {H.}~\bibnamefont {Chen}},
		\bibinfo {author} {\bibfnamefont {S.}~\bibnamefont {Cao}}, \bibinfo {author}
		{\bibfnamefont {D.}~\bibnamefont {Picozzi}}, \bibinfo {author} {\bibfnamefont
			{K.}~\bibnamefont {Setia}}, \bibinfo {author} {\bibfnamefont
			{Y.}~\bibnamefont {Li}}, \bibinfo {author} {\bibfnamefont {E.}~\bibnamefont
			{Grant}}, \bibinfo {author} {\bibfnamefont {L.}~\bibnamefont {Wossnig}},
		\bibinfo {author} {\bibfnamefont {I.}~\bibnamefont {Rungger}}, \bibinfo
		{author} {\bibfnamefont {G.~H.}\ \bibnamefont {Booth}}, \emph {et~al.},\
	}\bibfield  {title} {\bibinfo {title} {The variational quantum eigensolver: a
			review of methods and best practices},\ }\href@noop {} {\bibfield  {journal}
		{\bibinfo  {journal} {Physics Reports}\ }\textbf {\bibinfo {volume} {986}},\
		\bibinfo {pages} {1} (\bibinfo {year} {2022})}\BibitemShut {NoStop}%
	\bibitem [{\citenamefont {Cuzzocrea}\ \emph {et~al.}(2020)\citenamefont
		{Cuzzocrea}, \citenamefont {Scemama}, \citenamefont {Briels}, \citenamefont
		{Moroni},\ and\ \citenamefont {Filippi}}]{variance_minimisation}%
	\BibitemOpen
	\bibfield  {author} {\bibinfo {author} {\bibfnamefont {A.}~\bibnamefont
			{Cuzzocrea}}, \bibinfo {author} {\bibfnamefont {A.}~\bibnamefont {Scemama}},
		\bibinfo {author} {\bibfnamefont {W.~J.}\ \bibnamefont {Briels}}, \bibinfo
		{author} {\bibfnamefont {S.}~\bibnamefont {Moroni}},\ and\ \bibinfo {author}
		{\bibfnamefont {C.}~\bibnamefont {Filippi}},\ }\bibfield  {title} {\bibinfo
		{title} {Variational principles in quantum monte carlo: The troubled story of
			variance minimization},\ }\href {https://doi.org/10.1021/acs.jctc.0c00147}
	{\bibfield  {journal} {\bibinfo  {journal} {Journal of Chemical Theory and
				Computation}\ }\textbf {\bibinfo {volume} {16}},\ \bibinfo {pages} {4203}
		(\bibinfo {year} {2020})},\ \bibinfo {note} {pMID: 32419451}\BibitemShut
	{NoStop}%
	\bibitem [{\citenamefont {Koczor}\ and\ \citenamefont
		{Benjamin}(2019)}]{koczor2019quantum}%
	\BibitemOpen
	\bibfield  {author} {\bibinfo {author} {\bibfnamefont {B.}~\bibnamefont
			{Koczor}}\ and\ \bibinfo {author} {\bibfnamefont {S.~C.}\ \bibnamefont
			{Benjamin}},\ }\bibfield  {title} {\bibinfo {title} {{Quantum natural
				gradient generalised to non-unitary circuits}},\ }\href@noop {} {\bibfield
		{journal} {\bibinfo  {journal} {arXiv preprint arXiv:1912.08660}\ } (\bibinfo
		{year} {2019})}\BibitemShut {NoStop}%
	\bibitem [{\citenamefont {Koczor}\ and\ \citenamefont
		{Benjamin}(2022)}]{koczor2020quantumAnalytic}%
	\BibitemOpen
	\bibfield  {author} {\bibinfo {author} {\bibfnamefont {B.}~\bibnamefont
			{Koczor}}\ and\ \bibinfo {author} {\bibfnamefont {S.~C.}\ \bibnamefont
			{Benjamin}},\ }\bibfield  {title} {\bibinfo {title} {Quantum analytic
			descent},\ }\href@noop {} {\bibfield  {journal} {\bibinfo  {journal}
			{Physical Review Research}\ }\textbf {\bibinfo {volume} {4}},\ \bibinfo
		{pages} {023017} (\bibinfo {year} {2022})}\BibitemShut {NoStop}%
	\bibitem [{\citenamefont {Reiher}\ \emph {et~al.}(2017)\citenamefont {Reiher},
		\citenamefont {Wiebe}, \citenamefont {Svore}, \citenamefont {Wecker},\ and\
		\citenamefont {Troyer}}]{reiherElucidatingReactionMechanisms2017}%
	\BibitemOpen
	\bibfield  {author} {\bibinfo {author} {\bibfnamefont {M.}~\bibnamefont
			{Reiher}}, \bibinfo {author} {\bibfnamefont {N.}~\bibnamefont {Wiebe}},
		\bibinfo {author} {\bibfnamefont {K.~M.}\ \bibnamefont {Svore}}, \bibinfo
		{author} {\bibfnamefont {D.}~\bibnamefont {Wecker}},\ and\ \bibinfo {author}
		{\bibfnamefont {M.}~\bibnamefont {Troyer}},\ }\bibfield  {title} {\bibinfo
		{title} {Elucidating reaction mechanisms on quantum computers},\ }\href
	{https://doi.org/10.1073/pnas.1619152114} {\bibfield  {journal} {\bibinfo
			{journal} {Proceedings of the National Academy of Sciences}\ }\textbf
		{\bibinfo {volume} {114}},\ \bibinfo {pages} {7555} (\bibinfo {year}
		{2017})}\BibitemShut {NoStop}%
	\bibitem [{\citenamefont {Khatri}\ \emph {et~al.}(2019)\citenamefont {Khatri},
		\citenamefont {LaRose}, \citenamefont {Poremba}, \citenamefont {Cincio},
		\citenamefont {Sornborger},\ and\ \citenamefont
		{Coles}}]{khatriQuantumassistedQuantumCompiling2019}%
	\BibitemOpen
	\bibfield  {author} {\bibinfo {author} {\bibfnamefont {S.}~\bibnamefont
			{Khatri}}, \bibinfo {author} {\bibfnamefont {R.}~\bibnamefont {LaRose}},
		\bibinfo {author} {\bibfnamefont {A.}~\bibnamefont {Poremba}}, \bibinfo
		{author} {\bibfnamefont {L.}~\bibnamefont {Cincio}}, \bibinfo {author}
		{\bibfnamefont {A.~T.}\ \bibnamefont {Sornborger}},\ and\ \bibinfo {author}
		{\bibfnamefont {P.~J.}\ \bibnamefont {Coles}},\ }\bibfield  {title} {\bibinfo
		{title} {Quantum-assisted quantum compiling},\ }\href
	{https://doi.org/10.22331/q-2019-05-13-140} {\bibfield  {journal} {\bibinfo
			{journal} {Quantum}\ }\textbf {\bibinfo {volume} {3}},\ \bibinfo {pages}
		{140} (\bibinfo {year} {2019})}\BibitemShut {NoStop}%
	\bibitem [{\citenamefont {Johnson}\ \emph {et~al.}(2017)\citenamefont
		{Johnson}, \citenamefont {Romero}, \citenamefont {Olson}, \citenamefont
		{Cao},\ and\ \citenamefont {{Aspuru-Guzik}}}]{QVECTOR}%
	\BibitemOpen
	\bibfield  {author} {\bibinfo {author} {\bibfnamefont {P.~D.}\ \bibnamefont
			{Johnson}}, \bibinfo {author} {\bibfnamefont {J.}~\bibnamefont {Romero}},
		\bibinfo {author} {\bibfnamefont {J.}~\bibnamefont {Olson}}, \bibinfo
		{author} {\bibfnamefont {Y.}~\bibnamefont {Cao}},\ and\ \bibinfo {author}
		{\bibfnamefont {A.}~\bibnamefont {{Aspuru-Guzik}}},\ }\bibfield  {title}
	{\bibinfo {title} {{{QVECTOR}}: An algorithm for device-tailored quantum
			error correction},\ }\href@noop {} {\bibfield  {journal} {\bibinfo  {journal}
			{arXiv:1711.02249 [quant-ph]}\ } (\bibinfo {year} {2017})},\ \Eprint
	{https://arxiv.org/abs/1711.02249} {arXiv:1711.02249 [quant-ph]} \BibitemShut
	{NoStop}%
	\bibitem [{\citenamefont {Nielsen}\ and\ \citenamefont
		{Chuang}(2011)}]{nielsenChuang}%
	\BibitemOpen
	\bibfield  {author} {\bibinfo {author} {\bibfnamefont {M.~A.}\ \bibnamefont
			{Nielsen}}\ and\ \bibinfo {author} {\bibfnamefont {I.~L.}\ \bibnamefont
			{Chuang}},\ }\href@noop {} {\emph {\bibinfo {title} {{Quantum Computation and
					Quantum Information}}}},\ \bibinfo {edition} {10th}\ ed.\ (\bibinfo
	{publisher} {Cambridge University Press},\ \bibinfo {address} {New York, NY,
		USA},\ \bibinfo {year} {2011})\BibitemShut {NoStop}%
	\bibitem [{\citenamefont {Jones}\ and\ \citenamefont
		{Benjamin}(2022)}]{jonesRobustQuantumCompilation2022}%
	\BibitemOpen
	\bibfield  {author} {\bibinfo {author} {\bibfnamefont {T.}~\bibnamefont
			{Jones}}\ and\ \bibinfo {author} {\bibfnamefont {S.~C.}\ \bibnamefont
			{Benjamin}},\ }\bibfield  {title} {\bibinfo {title} {Robust quantum
			compilation and circuit optimisation via energy minimisation},\ }\href
	{https://doi.org/10.22331/q-2022-01-24-628} {\bibfield  {journal} {\bibinfo
			{journal} {Quantum}\ }\textbf {\bibinfo {volume} {6}},\ \bibinfo {pages}
		{628} (\bibinfo {year} {2022})}\BibitemShut {NoStop}%
	\bibitem [{\citenamefont {Dennis~Jr}\ and\ \citenamefont
		{Schnabel}(1996)}]{dennis1996numerical}%
	\BibitemOpen
	\bibfield  {author} {\bibinfo {author} {\bibfnamefont {J.~E.}\ \bibnamefont
			{Dennis~Jr}}\ and\ \bibinfo {author} {\bibfnamefont {R.~B.}\ \bibnamefont
			{Schnabel}},\ }\href@noop {} {\emph {\bibinfo {title} {Numerical methods for
				unconstrained optimization and nonlinear equations}}}\ (\bibinfo  {publisher}
	{SIAM},\ \bibinfo {year} {1996})\BibitemShut {NoStop}%
	\bibitem [{\citenamefont {Press}\ \emph {et~al.}(2007)\citenamefont {Press},
		\citenamefont {Teukolsky}, \citenamefont {Vetterling},\ and\ \citenamefont
		{Flannery}}]{press2007numerical}%
	\BibitemOpen
	\bibfield  {author} {\bibinfo {author} {\bibfnamefont {W.~H.}\ \bibnamefont
			{Press}}, \bibinfo {author} {\bibfnamefont {S.~A.}\ \bibnamefont
			{Teukolsky}}, \bibinfo {author} {\bibfnamefont {W.~T.}\ \bibnamefont
			{Vetterling}},\ and\ \bibinfo {author} {\bibfnamefont {B.~P.}\ \bibnamefont
			{Flannery}},\ }\href@noop {} {\emph {\bibinfo {title} {Numerical recipes 3rd
				edition: The art of scientific computing}}}\ (\bibinfo  {publisher}
	{Cambridge university press},\ \bibinfo {year} {2007})\BibitemShut {NoStop}%
	\bibitem [{\citenamefont {Okawa}\ \emph {et~al.}(2018)\citenamefont {Okawa},
		\citenamefont {Fujisawa}, \citenamefont {Yamamoto}, \citenamefont {Hirai},
		\citenamefont {Yasutake}, \citenamefont {Nagakura},\ and\ \citenamefont
		{Yamada}}]{okawa2018w4}%
	\BibitemOpen
	\bibfield  {author} {\bibinfo {author} {\bibfnamefont {H.}~\bibnamefont
			{Okawa}}, \bibinfo {author} {\bibfnamefont {K.}~\bibnamefont {Fujisawa}},
		\bibinfo {author} {\bibfnamefont {Y.}~\bibnamefont {Yamamoto}}, \bibinfo
		{author} {\bibfnamefont {R.}~\bibnamefont {Hirai}}, \bibinfo {author}
		{\bibfnamefont {N.}~\bibnamefont {Yasutake}}, \bibinfo {author}
		{\bibfnamefont {H.}~\bibnamefont {Nagakura}},\ and\ \bibinfo {author}
		{\bibfnamefont {S.}~\bibnamefont {Yamada}},\ }\bibfield  {title} {\bibinfo
		{title} {The w4 method: a new multi-dimensional root-finding scheme for
			nonlinear systems of equations},\ }\href@noop {} {\bibfield  {journal}
		{\bibinfo  {journal} {arXiv preprint arXiv:1809.04495}\ } (\bibinfo {year}
		{2018})}\BibitemShut {NoStop}%
	\bibitem [{\citenamefont {Pasquini}\ and\ \citenamefont
		{Trigiante}(1985)}]{pasquini1985globally}%
	\BibitemOpen
	\bibfield  {author} {\bibinfo {author} {\bibfnamefont {L.}~\bibnamefont
			{Pasquini}}\ and\ \bibinfo {author} {\bibfnamefont {D.}~\bibnamefont
			{Trigiante}},\ }\bibfield  {title} {\bibinfo {title} {A globally convergent
			method for simultaneously finding polynomial roots},\ }\href@noop {}
	{\bibfield  {journal} {\bibinfo  {journal} {mathematics of computation}\
		}\textbf {\bibinfo {volume} {44}},\ \bibinfo {pages} {135} (\bibinfo {year}
		{1985})}\BibitemShut {NoStop}%
	\bibitem [{\citenamefont {Bergou}\ \emph {et~al.}(2022)\citenamefont {Bergou},
		\citenamefont {Diouane}, \citenamefont {Kungurtsev},\ and\ \citenamefont
		{Royer}}]{bergouStochasticLevenbergMarquardtMethod2021}%
	\BibitemOpen
	\bibfield  {author} {\bibinfo {author} {\bibfnamefont {E.~H.}\ \bibnamefont
			{Bergou}}, \bibinfo {author} {\bibfnamefont {Y.}~\bibnamefont {Diouane}},
		\bibinfo {author} {\bibfnamefont {V.}~\bibnamefont {Kungurtsev}},\ and\
		\bibinfo {author} {\bibfnamefont {C.~W.}\ \bibnamefont {Royer}},\ }\bibfield
	{title} {\bibinfo {title} {A stochastic levenberg--marquardt method using
			random models with complexity results},\ }\href@noop {} {\bibfield  {journal}
		{\bibinfo  {journal} {SIAM/ASA Journal on Uncertainty Quantification}\
		}\textbf {\bibinfo {volume} {10}},\ \bibinfo {pages} {507} (\bibinfo {year}
		{2022})}\BibitemShut {NoStop}%
	\bibitem [{\citenamefont {Liew}\ \emph {et~al.}(2016)\citenamefont {Liew},
		\citenamefont {{Khalil-Hani}},\ and\ \citenamefont
		{Bakhteri}}]{liewOptimizedSecondOrder2016}%
	\BibitemOpen
	\bibfield  {author} {\bibinfo {author} {\bibfnamefont {S.~S.}\ \bibnamefont
			{Liew}}, \bibinfo {author} {\bibfnamefont {M.}~\bibnamefont
			{{Khalil-Hani}}},\ and\ \bibinfo {author} {\bibfnamefont {R.}~\bibnamefont
			{Bakhteri}},\ }\bibfield  {title} {\bibinfo {title} {An optimized second
			order stochastic learning algorithm for neural network training},\ }\href
	{https://doi.org/10.1016/j.neucom.2015.12.076} {\bibfield  {journal}
		{\bibinfo  {journal} {Neurocomputing}\ }\textbf {\bibinfo {volume} {186}},\
		\bibinfo {pages} {74} (\bibinfo {year} {2016})}\BibitemShut {NoStop}%
	\bibitem [{\citenamefont {Ruder}(2016)}]{ruder2016overview}%
	\BibitemOpen
	\bibfield  {author} {\bibinfo {author} {\bibfnamefont {S.}~\bibnamefont
			{Ruder}},\ }\bibfield  {title} {\bibinfo {title} {{An overview of gradient
				descent optimization algorithms}},\ }\href@noop {} {\bibfield  {journal}
		{\bibinfo  {journal} {arXiv preprint arXiv:1609.04747}\ } (\bibinfo {year}
		{2016})}\BibitemShut {NoStop}%
	\bibitem [{\citenamefont {Sweke}\ \emph {et~al.}(2020)\citenamefont {Sweke},
		\citenamefont {Wilde}, \citenamefont {Meyer}, \citenamefont {Schuld},
		\citenamefont {F{\"a}hrmann}, \citenamefont {Meynard-Piganeau},\ and\
		\citenamefont {Eisert}}]{sweke2019stochastic}%
	\BibitemOpen
	\bibfield  {author} {\bibinfo {author} {\bibfnamefont {R.}~\bibnamefont
			{Sweke}}, \bibinfo {author} {\bibfnamefont {F.}~\bibnamefont {Wilde}},
		\bibinfo {author} {\bibfnamefont {J.}~\bibnamefont {Meyer}}, \bibinfo
		{author} {\bibfnamefont {M.}~\bibnamefont {Schuld}}, \bibinfo {author}
		{\bibfnamefont {P.~K.}\ \bibnamefont {F{\"a}hrmann}}, \bibinfo {author}
		{\bibfnamefont {B.}~\bibnamefont {Meynard-Piganeau}},\ and\ \bibinfo {author}
		{\bibfnamefont {J.}~\bibnamefont {Eisert}},\ }\bibfield  {title} {\bibinfo
		{title} {Stochastic gradient descent for hybrid quantum-classical
			optimization},\ }\href@noop {} {\bibfield  {journal} {\bibinfo  {journal}
			{Quantum}\ }\textbf {\bibinfo {volume} {4}},\ \bibinfo {pages} {314}
		(\bibinfo {year} {2020})}\BibitemShut {NoStop}%
	\bibitem [{Note1()}]{Note1}%
	\BibitemOpen
	\bibinfo {note} {Stochastic gradient descent for VQE has been termed for
		instances when shot noise on estimated gradients is significant \cite
		{sweke2019stochastic}. In contrast, the present approach is stochastic due to
		the random selection of constraints}\BibitemShut {NoStop}%
	\bibitem [{\citenamefont {Cai}\ \emph {et~al.}(2022)\citenamefont {Cai},
		\citenamefont {Babbush}, \citenamefont {Benjamin}, \citenamefont {Endo},
		\citenamefont {Huggins}, \citenamefont {Li}, \citenamefont {McClean},\ and\
		\citenamefont {O'Brien}}]{qemreview}%
	\BibitemOpen
	\bibfield  {author} {\bibinfo {author} {\bibfnamefont {Z.}~\bibnamefont
			{Cai}}, \bibinfo {author} {\bibfnamefont {R.}~\bibnamefont {Babbush}},
		\bibinfo {author} {\bibfnamefont {S.~C.}\ \bibnamefont {Benjamin}}, \bibinfo
		{author} {\bibfnamefont {S.}~\bibnamefont {Endo}}, \bibinfo {author}
		{\bibfnamefont {W.~J.}\ \bibnamefont {Huggins}}, \bibinfo {author}
		{\bibfnamefont {Y.}~\bibnamefont {Li}}, \bibinfo {author} {\bibfnamefont
			{J.~R.}\ \bibnamefont {McClean}},\ and\ \bibinfo {author} {\bibfnamefont
			{T.~E.}\ \bibnamefont {O'Brien}},\ }\bibfield  {title} {\bibinfo {title}
		{{Quantum Error Mitigation}},\ }\href@noop {} {\bibfield  {journal} {\bibinfo
			{journal} {arXiv preprint arXiv:2210.00921}\ } (\bibinfo {year}
		{2022})}\BibitemShut {NoStop}%
	\bibitem [{\citenamefont {Koczor}(2021{\natexlab{a}})}]{koczor2020exponential}%
	\BibitemOpen
	\bibfield  {author} {\bibinfo {author} {\bibfnamefont {B.}~\bibnamefont
			{Koczor}},\ }\bibfield  {title} {\bibinfo {title} {Exponential error
			suppression for near-term quantum devices},\ }\href
	{https://doi.org/10.1103/PhysRevX.11.031057} {\bibfield  {journal} {\bibinfo
			{journal} {Phys. Rev. X}\ }\textbf {\bibinfo {volume} {11}},\ \bibinfo
		{pages} {031057} (\bibinfo {year} {2021}{\natexlab{a}})}\BibitemShut
	{NoStop}%
	\bibitem [{\citenamefont {Koczor}(2021{\natexlab{b}})}]{koczor2021dominant}%
	\BibitemOpen
	\bibfield  {author} {\bibinfo {author} {\bibfnamefont {B.}~\bibnamefont
			{Koczor}},\ }\bibfield  {title} {\bibinfo {title} {The dominant eigenvector
			of a noisy quantum state},\ }\href@noop {} {\bibfield  {journal} {\bibinfo
			{journal} {New Journal of Physics}\ }\textbf {\bibinfo {volume} {23}},\
		\bibinfo {pages} {123047} (\bibinfo {year} {2021}{\natexlab{b}})}\BibitemShut
	{NoStop}%
	\bibitem [{\citenamefont {Huggins}\ \emph {et~al.}(2021)\citenamefont
		{Huggins}, \citenamefont {McArdle}, \citenamefont {O’Brien}, \citenamefont
		{Lee}, \citenamefont {Rubin}, \citenamefont {Boixo}, \citenamefont {Whaley},
		\citenamefont {Babbush},\ and\ \citenamefont {McClean}}]{huggins2020virtual}%
	\BibitemOpen
	\bibfield  {author} {\bibinfo {author} {\bibfnamefont {W.~J.}\ \bibnamefont
			{Huggins}}, \bibinfo {author} {\bibfnamefont {S.}~\bibnamefont {McArdle}},
		\bibinfo {author} {\bibfnamefont {T.~E.}\ \bibnamefont {O’Brien}}, \bibinfo
		{author} {\bibfnamefont {J.}~\bibnamefont {Lee}}, \bibinfo {author}
		{\bibfnamefont {N.~C.}\ \bibnamefont {Rubin}}, \bibinfo {author}
		{\bibfnamefont {S.}~\bibnamefont {Boixo}}, \bibinfo {author} {\bibfnamefont
			{K.~B.}\ \bibnamefont {Whaley}}, \bibinfo {author} {\bibfnamefont
			{R.}~\bibnamefont {Babbush}},\ and\ \bibinfo {author} {\bibfnamefont {J.~R.}\
			\bibnamefont {McClean}},\ }\bibfield  {title} {\bibinfo {title} {Virtual
			distillation for quantum error mitigation},\ }\href@noop {} {\bibfield
		{journal} {\bibinfo  {journal} {Physical Review X}\ }\textbf {\bibinfo
			{volume} {11}},\ \bibinfo {pages} {041036} (\bibinfo {year}
		{2021})}\BibitemShut {NoStop}%
	\bibitem [{\citenamefont {Sharma}\ \emph {et~al.}(2020)\citenamefont {Sharma},
		\citenamefont {Khatri}, \citenamefont {Cerezo},\ and\ \citenamefont
		{Coles}}]{sharma2020noise}%
	\BibitemOpen
	\bibfield  {author} {\bibinfo {author} {\bibfnamefont {K.}~\bibnamefont
			{Sharma}}, \bibinfo {author} {\bibfnamefont {S.}~\bibnamefont {Khatri}},
		\bibinfo {author} {\bibfnamefont {M.}~\bibnamefont {Cerezo}},\ and\ \bibinfo
		{author} {\bibfnamefont {P.~J.}\ \bibnamefont {Coles}},\ }\bibfield  {title}
	{\bibinfo {title} {{Noise resilience of variational quantum compiling}},\
	}\href@noop {} {\bibfield  {journal} {\bibinfo  {journal} {New Journal of
				Physics}\ }\textbf {\bibinfo {volume} {22}},\ \bibinfo {pages} {043006}
		(\bibinfo {year} {2020})}\BibitemShut {NoStop}%
	\bibitem [{\citenamefont {Derby}\ \emph {et~al.}(2021)\citenamefont {Derby},
		\citenamefont {Klassen}, \citenamefont {Bausch},\ and\ \citenamefont
		{Cubitt}}]{derbyCompactFermionQubit2021}%
	\BibitemOpen
	\bibfield  {author} {\bibinfo {author} {\bibfnamefont {C.}~\bibnamefont
			{Derby}}, \bibinfo {author} {\bibfnamefont {J.}~\bibnamefont {Klassen}},
		\bibinfo {author} {\bibfnamefont {J.}~\bibnamefont {Bausch}},\ and\ \bibinfo
		{author} {\bibfnamefont {T.}~\bibnamefont {Cubitt}},\ }\bibfield  {title}
	{\bibinfo {title} {Compact fermion to qubit mappings},\ }\href
	{https://doi.org/10.1103/PhysRevB.104.035118} {\bibfield  {journal} {\bibinfo
			{journal} {Physical Review B}\ }\textbf {\bibinfo {volume} {104}},\ \bibinfo
		{pages} {035118} (\bibinfo {year} {2021})}\BibitemShut {NoStop}%
	\bibitem [{\citenamefont {Cerezo}\ \emph {et~al.}(2020)\citenamefont {Cerezo},
		\citenamefont {Arrasmith}, \citenamefont {Babbush}, \citenamefont {Benjamin},
		\citenamefont {Endo}, \citenamefont {Fujii}, \citenamefont {McClean},
		\citenamefont {Mitarai}, \citenamefont {Yuan}, \citenamefont {Cincio} \emph
		{et~al.}}]{cerezo2020variationalreview}%
	\BibitemOpen
	\bibfield  {author} {\bibinfo {author} {\bibfnamefont {M.}~\bibnamefont
			{Cerezo}}, \bibinfo {author} {\bibfnamefont {A.}~\bibnamefont {Arrasmith}},
		\bibinfo {author} {\bibfnamefont {R.}~\bibnamefont {Babbush}}, \bibinfo
		{author} {\bibfnamefont {S.~C.}\ \bibnamefont {Benjamin}}, \bibinfo {author}
		{\bibfnamefont {S.}~\bibnamefont {Endo}}, \bibinfo {author} {\bibfnamefont
			{K.}~\bibnamefont {Fujii}}, \bibinfo {author} {\bibfnamefont {J.~R.}\
			\bibnamefont {McClean}}, \bibinfo {author} {\bibfnamefont {K.}~\bibnamefont
			{Mitarai}}, \bibinfo {author} {\bibfnamefont {X.}~\bibnamefont {Yuan}},
		\bibinfo {author} {\bibfnamefont {L.}~\bibnamefont {Cincio}}, \emph
		{et~al.},\ }\bibfield  {title} {\bibinfo {title} {{Variational quantum
				algorithms}},\ }\href@noop {} {\bibfield  {journal} {\bibinfo  {journal}
			{Nat. Rev. Phys.}\ }3, 625 (\bibinfo {year} {2021})}
   \BibitemShut
	{NoStop}%
	\bibitem [{\citenamefont {Pagano}\ \emph {et~al.}(2020)\citenamefont {Pagano},
		\citenamefont {Bapat}, \citenamefont {Becker}, \citenamefont {Collins},
		\citenamefont {De}, \citenamefont {Hess}, \citenamefont {Kaplan},
		\citenamefont {Kyprianidis}, \citenamefont {Tan}, \citenamefont {Baldwin}
		\emph {et~al.}}]{pagano2020quantum}%
	\BibitemOpen
	\bibfield  {author} {\bibinfo {author} {\bibfnamefont {G.}~\bibnamefont
			{Pagano}}, \bibinfo {author} {\bibfnamefont {A.}~\bibnamefont {Bapat}},
		\bibinfo {author} {\bibfnamefont {P.}~\bibnamefont {Becker}}, \bibinfo
		{author} {\bibfnamefont {K.~S.}\ \bibnamefont {Collins}}, \bibinfo {author}
		{\bibfnamefont {A.}~\bibnamefont {De}}, \bibinfo {author} {\bibfnamefont
			{P.~W.}\ \bibnamefont {Hess}}, \bibinfo {author} {\bibfnamefont {H.~B.}\
			\bibnamefont {Kaplan}}, \bibinfo {author} {\bibfnamefont {A.}~\bibnamefont
			{Kyprianidis}}, \bibinfo {author} {\bibfnamefont {W.~L.}\ \bibnamefont
			{Tan}}, \bibinfo {author} {\bibfnamefont {C.}~\bibnamefont {Baldwin}}, \emph
		{et~al.},\ }\bibfield  {title} {\bibinfo {title} {{Quantum approximate
				optimization of the long-range Ising model with a trapped-ion quantum
				simulator}},\ }\href@noop {} {\bibfield  {journal} {\bibinfo  {journal}
			{PNAS}\ }\textbf {\bibinfo {volume} {117}},\ \bibinfo {pages} {25396}
		(\bibinfo {year} {2020})}\BibitemShut {NoStop}%
	\bibitem [{\citenamefont {Harrigan}\ \emph {et~al.}(2021)\citenamefont
		{Harrigan}, \citenamefont {Sung}, \citenamefont {Neeley}, \citenamefont
		{Satzinger}, \citenamefont {Arute}, \citenamefont {Arya}, \citenamefont
		{Atalaya}, \citenamefont {Bardin}, \citenamefont {Barends}, \citenamefont
		{Boixo} \emph {et~al.}}]{arute2020quantum}%
	\BibitemOpen
	\bibfield  {author} {\bibinfo {author} {\bibfnamefont {M.~P.}\ \bibnamefont
			{Harrigan}}, \bibinfo {author} {\bibfnamefont {K.~J.}\ \bibnamefont {Sung}},
		\bibinfo {author} {\bibfnamefont {M.}~\bibnamefont {Neeley}}, \bibinfo
		{author} {\bibfnamefont {K.~J.}\ \bibnamefont {Satzinger}}, \bibinfo {author}
		{\bibfnamefont {F.}~\bibnamefont {Arute}}, \bibinfo {author} {\bibfnamefont
			{K.}~\bibnamefont {Arya}}, \bibinfo {author} {\bibfnamefont {J.}~\bibnamefont
			{Atalaya}}, \bibinfo {author} {\bibfnamefont {J.~C.}\ \bibnamefont {Bardin}},
		\bibinfo {author} {\bibfnamefont {R.}~\bibnamefont {Barends}}, \bibinfo
		{author} {\bibfnamefont {S.}~\bibnamefont {Boixo}}, \emph {et~al.},\
	}\bibfield  {title} {\bibinfo {title} {Quantum approximate optimization of
			non-planar graph problems on a planar superconducting processor},\
	}\href@noop {} {\bibfield  {journal} {\bibinfo  {journal} {Nature Physics}\
		}\textbf {\bibinfo {volume} {17}},\ \bibinfo {pages} {332} (\bibinfo {year}
		{2021})}\BibitemShut {NoStop}%
	\bibitem [{\citenamefont {Schuld}\ \emph {et~al.}(2019)\citenamefont {Schuld},
		\citenamefont {Bergholm}, \citenamefont {Gogolin}, \citenamefont {Izaac},\
		and\ \citenamefont {Killoran}}]{paramshift}%
	\BibitemOpen
	\bibfield  {author} {\bibinfo {author} {\bibfnamefont {M.}~\bibnamefont
			{Schuld}}, \bibinfo {author} {\bibfnamefont {V.}~\bibnamefont {Bergholm}},
		\bibinfo {author} {\bibfnamefont {C.}~\bibnamefont {Gogolin}}, \bibinfo
		{author} {\bibfnamefont {J.}~\bibnamefont {Izaac}},\ and\ \bibinfo {author}
		{\bibfnamefont {N.}~\bibnamefont {Killoran}},\ }\bibfield  {title} {\bibinfo
		{title} {{Evaluating analytic gradients on quantum hardware}},\ }\href
	{https://doi.org/10.1103/PhysRevA.99.032331} {\bibfield  {journal} {\bibinfo
			{journal} {Phys. Rev. A}\ }\textbf {\bibinfo {volume} {99}},\ \bibinfo
		{pages} {032331} (\bibinfo {year} {2019})}\BibitemShut {NoStop}%
	\bibitem [{\citenamefont {Huang}\ \emph {et~al.}(2021)\citenamefont {Huang},
		\citenamefont {Kueng},\ and\ \citenamefont
		{Preskill}}]{shadows_derandomization}%
	\BibitemOpen
	\bibfield  {author} {\bibinfo {author} {\bibfnamefont {H.-Y.}\ \bibnamefont
			{Huang}}, \bibinfo {author} {\bibfnamefont {R.}~\bibnamefont {Kueng}},\ and\
		\bibinfo {author} {\bibfnamefont {J.}~\bibnamefont {Preskill}},\ }\bibfield
	{title} {\bibinfo {title} {Efficient {{Estimation}} of {{Pauli Observables}}
			by {{Derandomization}}},\ }\href
	{https://doi.org/10.1103/PhysRevLett.127.030503} {\bibfield  {journal}
		{\bibinfo  {journal} {Physical Review Letters}\ }\textbf {\bibinfo {volume}
			{127}},\ \bibinfo {pages} {030503} (\bibinfo {year} {2021})}\BibitemShut
	{NoStop}%
	\bibitem [{\citenamefont {Moore}\ \emph {et~al.}(2009)\citenamefont {Moore},
		\citenamefont {Kearfott},\ and\ \citenamefont
		{Cloud}}]{mooreIntroductionIntervalAnalysis2009}%
	\BibitemOpen
	\bibfield  {author} {\bibinfo {author} {\bibfnamefont {R.~E.}\ \bibnamefont
			{Moore}}, \bibinfo {author} {\bibfnamefont {R.~B.}\ \bibnamefont
			{Kearfott}},\ and\ \bibinfo {author} {\bibfnamefont {M.~J.}\ \bibnamefont
			{Cloud}},\ }\href {https://doi.org/10.1137/1.9780898717716} {\emph {\bibinfo
			{title} {Introduction to {{Interval Analysis}}}}}\ (\bibinfo  {publisher}
	{{Society for Industrial and Applied Mathematics}},\ \bibinfo {year}
	{2009})\BibitemShut {NoStop}%
	\bibitem [{\citenamefont {Schuch}\ and\ \citenamefont
		{Siewert}(2003)}]{schuchProgrammableNetworksQuantum2003}%
	\BibitemOpen
	\bibfield  {author} {\bibinfo {author} {\bibfnamefont {N.}~\bibnamefont
			{Schuch}}\ and\ \bibinfo {author} {\bibfnamefont {J.}~\bibnamefont
			{Siewert}},\ }\bibfield  {title} {\bibinfo {title} {Programmable {{Networks}}
			for {{Quantum Algorithms}}},\ }\href
	{https://doi.org/10.1103/PhysRevLett.91.027902} {\bibfield  {journal}
		{\bibinfo  {journal} {Physical Review Letters}\ }\textbf {\bibinfo {volume}
			{91}},\ \bibinfo {pages} {027902} (\bibinfo {year} {2003})}\BibitemShut
	{NoStop}%
	\bibitem [{\citenamefont {Moro}\ \emph {et~al.}(2021)\citenamefont {Moro},
		\citenamefont {Paris}, \citenamefont {Restelli},\ and\ \citenamefont
		{Prati}}]{moroQuantumCompilingDeep2021}%
	\BibitemOpen
	\bibfield  {author} {\bibinfo {author} {\bibfnamefont {L.}~\bibnamefont
			{Moro}}, \bibinfo {author} {\bibfnamefont {M.~G.~A.}\ \bibnamefont {Paris}},
		\bibinfo {author} {\bibfnamefont {M.}~\bibnamefont {Restelli}},\ and\
		\bibinfo {author} {\bibfnamefont {E.}~\bibnamefont {Prati}},\ }\bibfield
	{title} {\bibinfo {title} {Quantum compiling by deep reinforcement
			learning},\ }\href {https://doi.org/10.1038/s42005-021-00684-3} {\bibfield
		{journal} {\bibinfo  {journal} {Communications Physics}\ }\textbf {\bibinfo
			{volume} {4}},\ \bibinfo {pages} {1} (\bibinfo {year} {2021})}\BibitemShut
	{NoStop}%
	\bibitem [{\citenamefont {Lucas}(2014)}]{lucasIsingFormulationsMany2014}%
	\BibitemOpen
	\bibfield  {author} {\bibinfo {author} {\bibfnamefont {A.}~\bibnamefont
			{Lucas}},\ }\bibfield  {title} {\bibinfo {title} {Ising formulations of many
			{{NP}} problems},\ }\href@noop {} {\bibfield  {journal} {\bibinfo  {journal}
			{Frontiers in Physics}\ }\textbf {\bibinfo {volume} {2}} (\bibinfo {year}
		{2014})}\BibitemShut {NoStop}%
	\bibitem [{\citenamefont {Kokail}\ \emph {et~al.}(2019)\citenamefont {Kokail},
		\citenamefont {Maier}, \citenamefont {{van Bijnen}}, \citenamefont {Brydges},
		\citenamefont {Joshi}, \citenamefont {Jurcevic}, \citenamefont {Muschik},
		\citenamefont {Silvi}, \citenamefont {Blatt}, \citenamefont {Roos},\ and\
		\citenamefont {Zoller}}]{latticeschwinger}%
	\BibitemOpen
	\bibfield  {author} {\bibinfo {author} {\bibfnamefont {C.}~\bibnamefont
			{Kokail}}, \bibinfo {author} {\bibfnamefont {C.}~\bibnamefont {Maier}},
		\bibinfo {author} {\bibfnamefont {R.}~\bibnamefont {{van Bijnen}}}, \bibinfo
		{author} {\bibfnamefont {T.}~\bibnamefont {Brydges}}, \bibinfo {author}
		{\bibfnamefont {M.~K.}\ \bibnamefont {Joshi}}, \bibinfo {author}
		{\bibfnamefont {P.}~\bibnamefont {Jurcevic}}, \bibinfo {author}
		{\bibfnamefont {C.~A.}\ \bibnamefont {Muschik}}, \bibinfo {author}
		{\bibfnamefont {P.}~\bibnamefont {Silvi}}, \bibinfo {author} {\bibfnamefont
			{R.}~\bibnamefont {Blatt}}, \bibinfo {author} {\bibfnamefont {C.~F.}\
			\bibnamefont {Roos}},\ and\ \bibinfo {author} {\bibfnamefont
			{P.}~\bibnamefont {Zoller}},\ }\bibfield  {title} {\bibinfo {title}
		{Self-verifying variational quantum simulation of lattice models},\ }\href
	{https://doi.org/10.1038/s41586-019-1177-4} {\bibfield  {journal} {\bibinfo
			{journal} {Nature}\ }\textbf {\bibinfo {volume} {569}},\ \bibinfo {pages}
		{355} (\bibinfo {year} {2019})}\BibitemShut {NoStop}%
	\bibitem [{\citenamefont {Luitz}\ \emph {et~al.}(2015)\citenamefont {Luitz},
		\citenamefont {Laflorencie},\ and\ \citenamefont
		{Alet}}]{luitzManybodyLocalizationEdge2015}%
	\BibitemOpen
	\bibfield  {author} {\bibinfo {author} {\bibfnamefont {D.~J.}\ \bibnamefont
			{Luitz}}, \bibinfo {author} {\bibfnamefont {N.}~\bibnamefont {Laflorencie}},\
		and\ \bibinfo {author} {\bibfnamefont {F.}~\bibnamefont {Alet}},\ }\bibfield
	{title} {\bibinfo {title} {Many-body localization edge in the random-field
			{Heisenberg} chain},\ }\href {https://doi.org/10.1103/PhysRevB.91.081103}
	{\bibfield  {journal} {\bibinfo  {journal} {Physical Review B}\ }\textbf
		{\bibinfo {volume} {91}},\ \bibinfo {pages} {081103} (\bibinfo {year}
		{2015})}\BibitemShut {NoStop}%
	\bibitem [{\citenamefont {Childs}\ \emph {et~al.}(2018)\citenamefont {Childs},
		\citenamefont {Maslov}, \citenamefont {Nam}, \citenamefont {Ross},\ and\
		\citenamefont {Su}}]{childs2018toward}%
	\BibitemOpen
	\bibfield  {author} {\bibinfo {author} {\bibfnamefont {A.~M.}\ \bibnamefont
			{Childs}}, \bibinfo {author} {\bibfnamefont {D.}~\bibnamefont {Maslov}},
		\bibinfo {author} {\bibfnamefont {Y.}~\bibnamefont {Nam}}, \bibinfo {author}
		{\bibfnamefont {N.~J.}\ \bibnamefont {Ross}},\ and\ \bibinfo {author}
		{\bibfnamefont {Y.}~\bibnamefont {Su}},\ }\bibfield  {title} {\bibinfo
		{title} {{Toward the first quantum simulation with quantum speedup}},\
	}\href@noop {} {\bibfield  {journal} {\bibinfo  {journal} {PNAS}\ }\textbf
		{\bibinfo {volume} {115}},\ \bibinfo {pages} {9456} (\bibinfo {year}
		{2018})}\BibitemShut {NoStop}%
	\bibitem [{\citenamefont {Nandkishore}\ and\ \citenamefont
		{Huse}(2015)}]{nandkishoreManyBodyLocalizationThermalization2015}%
	\BibitemOpen
	\bibfield  {author} {\bibinfo {author} {\bibfnamefont {R.}~\bibnamefont
			{Nandkishore}}\ and\ \bibinfo {author} {\bibfnamefont {D.~A.}\ \bibnamefont
			{Huse}},\ }\bibfield  {title} {\bibinfo {title} {Many-{Body} {Localization}
			and {Thermalization} in {Quantum} {Statistical} {Mechanics}},\ }\href
	{https://doi.org/10.1146/annurev-conmatphys-031214-014726} {\bibfield
		{journal} {\bibinfo  {journal} {Annual Review of Condensed Matter Physics}\
		}\textbf {\bibinfo {volume} {6}},\ \bibinfo {pages} {15} (\bibinfo {year}
		{2015})}\BibitemShut {NoStop}%
	\bibitem [{\citenamefont {McArdle}\ \emph {et~al.}(2019)\citenamefont
		{McArdle}, \citenamefont {Jones}, \citenamefont {Endo}, \citenamefont {Li},
		\citenamefont {Benjamin},\ and\ \citenamefont {Yuan}}]{samimagtime}%
	\BibitemOpen
	\bibfield  {author} {\bibinfo {author} {\bibfnamefont {S.}~\bibnamefont
			{McArdle}}, \bibinfo {author} {\bibfnamefont {T.}~\bibnamefont {Jones}},
		\bibinfo {author} {\bibfnamefont {S.}~\bibnamefont {Endo}}, \bibinfo {author}
		{\bibfnamefont {Y.}~\bibnamefont {Li}}, \bibinfo {author} {\bibfnamefont
			{S.~C.}\ \bibnamefont {Benjamin}},\ and\ \bibinfo {author} {\bibfnamefont
			{X.}~\bibnamefont {Yuan}},\ }\bibfield  {title} {\bibinfo {title}
		{Variational ansatz-based quantum simulation of imaginary time evolution},\
	}\href@noop {} {\bibfield  {journal} {\bibinfo  {journal} {npj Quantum
				Information}\ }\textbf {\bibinfo {volume} {5}},\ \bibinfo {pages} {75}
		(\bibinfo {year} {2019})}\BibitemShut {NoStop}%
	\bibitem [{\citenamefont {Stokes}\ \emph {et~al.}(2020)\citenamefont {Stokes},
		\citenamefont {Izaac}, \citenamefont {Killoran},\ and\ \citenamefont
		{Carleo}}]{quantumnatgrad}%
	\BibitemOpen
	\bibfield  {author} {\bibinfo {author} {\bibfnamefont {J.}~\bibnamefont
			{Stokes}}, \bibinfo {author} {\bibfnamefont {J.}~\bibnamefont {Izaac}},
		\bibinfo {author} {\bibfnamefont {N.}~\bibnamefont {Killoran}},\ and\
		\bibinfo {author} {\bibfnamefont {G.}~\bibnamefont {Carleo}},\ }\bibfield
	{title} {\bibinfo {title} {Quantum natural gradient},\ }\href@noop {}
	{\bibfield  {journal} {\bibinfo  {journal} {Quantum}\ }\textbf {\bibinfo
			{volume} {4}},\ \bibinfo {pages} {269} (\bibinfo {year} {2020})}\BibitemShut
	{NoStop}%
	\bibitem [{\citenamefont {Yamamoto}(2019)}]{yamamoto2019natural}%
	\BibitemOpen
	\bibfield  {author} {\bibinfo {author} {\bibfnamefont {N.}~\bibnamefont
			{Yamamoto}},\ }\bibfield  {title} {\bibinfo {title} {On the natural gradient
			for variational quantum eigensolver},\ }\href@noop {} {\bibfield  {journal}
		{\bibinfo  {journal} {arXiv preprint arXiv:1909.05074}\ } (\bibinfo {year}
		{2019})}\BibitemShut {NoStop}%
	\bibitem [{\citenamefont {Higgott}\ \emph {et~al.}(2019)\citenamefont
		{Higgott}, \citenamefont {Wang},\ and\ \citenamefont
		{Brierley}}]{higgott2018variational}%
	\BibitemOpen
	\bibfield  {author} {\bibinfo {author} {\bibfnamefont {O.}~\bibnamefont
			{Higgott}}, \bibinfo {author} {\bibfnamefont {D.}~\bibnamefont {Wang}},\ and\
		\bibinfo {author} {\bibfnamefont {S.}~\bibnamefont {Brierley}},\ }\bibfield
	{title} {\bibinfo {title} {Variational {{Quantum Computation}} of {{Excited
					States}}},\ }\href {https://doi.org/10.22331/q-2019-07-01-156} {\bibfield
		{journal} {\bibinfo  {journal} {Quantum}\ }\textbf {\bibinfo {volume} {3}},\
		\bibinfo {pages} {156} (\bibinfo {year} {2019})}\BibitemShut {NoStop}%
	\bibitem [{\citenamefont {Bittel}\ and\ \citenamefont
		{Kliesch}(2021{\natexlab{b}})}]{bittelTrainingVariationalQuantum2021}%
	\BibitemOpen
	\bibfield  {author} {\bibinfo {author} {\bibfnamefont {L.}~\bibnamefont
			{Bittel}}\ and\ \bibinfo {author} {\bibfnamefont {M.}~\bibnamefont
			{Kliesch}},\ }\bibfield  {title} {\bibinfo {title} {Training {{Variational
					Quantum Algorithms Is NP-Hard}}},\ }\href
	{https://doi.org/10.1103/PhysRevLett.127.120502} {\bibfield  {journal}
		{\bibinfo  {journal} {Physical Review Letters}\ }\textbf {\bibinfo {volume}
			{127}},\ \bibinfo {pages} {120502} (\bibinfo {year}
		{2021}{\natexlab{b}})}\BibitemShut {NoStop}%
	\bibitem [{\citenamefont {Zhang}\ \emph {et~al.}(2020)\citenamefont {Zhang},
		\citenamefont {Yuan},\ and\ \citenamefont
		{Yin}}]{zhangVariationalQuantumEigensolvers2020}%
	\BibitemOpen
	\bibfield  {author} {\bibinfo {author} {\bibfnamefont {D.-B.}\ \bibnamefont
			{Zhang}}, \bibinfo {author} {\bibfnamefont {Z.-H.}\ \bibnamefont {Yuan}},\
		and\ \bibinfo {author} {\bibfnamefont {T.}~\bibnamefont {Yin}},\ }\bibfield
	{title} {\bibinfo {title} {Variational quantum eigensolvers by variance
			minimization},\ }\href@noop {} {\bibfield  {journal} {\bibinfo  {journal}
			{preprint arXiv:2006.15781}\ } (\bibinfo {year} {2020})}\BibitemShut
	{NoStop}%
	\bibitem [{\citenamefont {McClean}\ \emph {et~al.}(2017)\citenamefont
		{McClean}, \citenamefont {{Kimchi-Schwartz}}, \citenamefont {Carter},\ and\
		\citenamefont {{de Jong}}}]{mccleanSubspaceExpansion}%
	\BibitemOpen
	\bibfield  {author} {\bibinfo {author} {\bibfnamefont {J.~R.}\ \bibnamefont
			{McClean}}, \bibinfo {author} {\bibfnamefont {M.~E.}\ \bibnamefont
			{{Kimchi-Schwartz}}}, \bibinfo {author} {\bibfnamefont {J.}~\bibnamefont
			{Carter}},\ and\ \bibinfo {author} {\bibfnamefont {W.~A.}\ \bibnamefont {{de
					Jong}}},\ }\bibfield  {title} {\bibinfo {title} {Hybrid quantum-classical
			hierarchy for mitigation of decoherence and determination of excited
			states},\ }\href {https://doi.org/10.1103/PhysRevA.95.042308} {\bibfield
		{journal} {\bibinfo  {journal} {Physical Review A}\ }\textbf {\bibinfo
			{volume} {95}},\ \bibinfo {pages} {042308} (\bibinfo {year}
		{2017})}\BibitemShut {NoStop}%
	\bibitem [{\citenamefont {Gill}\ \emph {et~al.}(1981)\citenamefont {Gill},
		\citenamefont {Murray},\ and\ \citenamefont
		{Wright}}]{gillPracticalOptimization1981}%
	\BibitemOpen
	\bibfield  {author} {\bibinfo {author} {\bibfnamefont {P.~E.}\ \bibnamefont
			{Gill}}, \bibinfo {author} {\bibfnamefont {W.}~\bibnamefont {Murray}},\ and\
		\bibinfo {author} {\bibfnamefont {M.~H.}\ \bibnamefont {Wright}},\
	}\href@noop {} {\emph {\bibinfo {title} {Practical Optimization}}}\ (\bibinfo
	{publisher} {{Academic Press}},\ \bibinfo {address} {{London; New York}},\
	\bibinfo {year} {1981})\BibitemShut {NoStop}%
	\bibitem [{\citenamefont {Mari}\ \emph {et~al.}(2021)\citenamefont {Mari},
		\citenamefont {Bromley},\ and\ \citenamefont
		{Killoran}}]{mariEstimatingGradientHigherorder2021}%
	\BibitemOpen
	\bibfield  {author} {\bibinfo {author} {\bibfnamefont {A.}~\bibnamefont
			{Mari}}, \bibinfo {author} {\bibfnamefont {T.~R.}\ \bibnamefont {Bromley}},\
		and\ \bibinfo {author} {\bibfnamefont {N.}~\bibnamefont {Killoran}},\
	}\bibfield  {title} {\bibinfo {title} {Estimating the gradient and
			higher-order derivatives on quantum hardware},\ }\href
	{https://doi.org/10.1103/PhysRevA.103.012405} {\bibfield  {journal} {\bibinfo
			{journal} {Physical Review A}\ }\textbf {\bibinfo {volume} {103}},\ \bibinfo
		{pages} {012405} (\bibinfo {year} {2021})}\BibitemShut {NoStop}%
	\bibitem [{\citenamefont {Chertkov}\ and\ \citenamefont
		{Clark}(2018)}]{PhysRevX.8.031029}%
	\BibitemOpen
	\bibfield  {author} {\bibinfo {author} {\bibfnamefont {E.}~\bibnamefont
			{Chertkov}}\ and\ \bibinfo {author} {\bibfnamefont {B.~K.}\ \bibnamefont
			{Clark}},\ }\bibfield  {title} {\bibinfo {title} {Computational inverse
			method for constructing spaces of quantum models from wave functions},\
	}\href {https://doi.org/10.1103/PhysRevX.8.031029} {\bibfield  {journal}
		{\bibinfo  {journal} {Phys. Rev. X}\ }\textbf {\bibinfo {volume} {8}},\
		\bibinfo {pages} {031029} (\bibinfo {year} {2018})}\BibitemShut {NoStop}%
	\bibitem [{\citenamefont {Qi}\ and\ \citenamefont
		{Ranard}(2019)}]{Qi2019determininglocal}%
	\BibitemOpen
	\bibfield  {author} {\bibinfo {author} {\bibfnamefont {X.-L.}\ \bibnamefont
			{Qi}}\ and\ \bibinfo {author} {\bibfnamefont {D.}~\bibnamefont {Ranard}},\
	}\bibfield  {title} {\bibinfo {title} {Determining a local {H}amiltonian from
			a single eigenstate},\ }\href {https://doi.org/10.22331/q-2019-07-08-159}
	{\bibfield  {journal} {\bibinfo  {journal} {{Quantum}}\ }\textbf {\bibinfo
			{volume} {3}},\ \bibinfo {pages} {159} (\bibinfo {year} {2019})}\BibitemShut
	{NoStop}%
	\bibitem [{\citenamefont {Bookatz}(2012)}]{bookatz2012qma}%
	\BibitemOpen
	\bibfield  {author} {\bibinfo {author} {\bibfnamefont {A.~D.}\ \bibnamefont
			{Bookatz}},\ }\bibfield  {title} {\bibinfo {title} {{QMA-complete
				problems}},\ }\href@noop {} {\bibfield  {journal} {\bibinfo  {journal} {Quantum Inf. Comput.\ }}14, 361 (\bibinfo {year} {2014})}\BibitemShut {NoStop}%
	\bibitem [{\citenamefont {Akhtar}\ \emph {et~al.}(2022)\citenamefont {Akhtar},
		\citenamefont {Hu},\ and\ \citenamefont
		{You}}]{akhtarScalableFlexibleClassical2022}%
	\BibitemOpen
	\bibfield  {author} {\bibinfo {author} {\bibfnamefont {A.~A.}\ \bibnamefont
			{Akhtar}}, \bibinfo {author} {\bibfnamefont {H.-Y.}\ \bibnamefont {Hu}},\
		and\ \bibinfo {author} {\bibfnamefont {Y.-Z.}\ \bibnamefont {You}},\
	}\href@noop {} {\bibinfo {title} {Scalable and {{Flexible Classical Shadow
					Tomography}} with {{Tensor Networks}}}} (\bibinfo {year} {2022}),\ \Eprint
	{https://arxiv.org/abs/2209.02093} {arXiv:2209.02093} \BibitemShut {NoStop}%
	\bibitem [{\citenamefont {Bertoni}\ \emph {et~al.}(2022)\citenamefont
		{Bertoni}, \citenamefont {Haferkamp}, \citenamefont {Hinsche}, \citenamefont
		{Ioannou}, \citenamefont {Eisert},\ and\ \citenamefont
		{Pashayan}}]{bertoniShallowShadowsExpectation2022}%
	\BibitemOpen
	\bibfield  {author} {\bibinfo {author} {\bibfnamefont {C.}~\bibnamefont
			{Bertoni}}, \bibinfo {author} {\bibfnamefont {J.}~\bibnamefont {Haferkamp}},
		\bibinfo {author} {\bibfnamefont {M.}~\bibnamefont {Hinsche}}, \bibinfo
		{author} {\bibfnamefont {M.}~\bibnamefont {Ioannou}}, \bibinfo {author}
		{\bibfnamefont {J.}~\bibnamefont {Eisert}},\ and\ \bibinfo {author}
		{\bibfnamefont {H.}~\bibnamefont {Pashayan}},\ }\href@noop {} {\bibinfo
		{title} {Shallow shadows: {{Expectation}} estimation using low-depth random
			{{Clifford}} circuits}} (\bibinfo {year} {2022}),\ \Eprint
	{https://arxiv.org/abs/2209.12924} {arXiv:2209.12924} \BibitemShut {NoStop}%
	\bibitem [{\citenamefont {Crawford}\ \emph {et~al.}(2021)\citenamefont
		{Crawford}, \citenamefont {van Straaten}, \citenamefont {Wang}, \citenamefont
		{Parks}, \citenamefont {Campbell},\ and\ \citenamefont
		{Brierley}}]{Crawford2019}%
	\BibitemOpen
	\bibfield  {author} {\bibinfo {author} {\bibfnamefont {O.}~\bibnamefont
			{Crawford}}, \bibinfo {author} {\bibfnamefont {B.}~\bibnamefont {van
				Straaten}}, \bibinfo {author} {\bibfnamefont {D.}~\bibnamefont {Wang}},
		\bibinfo {author} {\bibfnamefont {T.}~\bibnamefont {Parks}}, \bibinfo
		{author} {\bibfnamefont {E.}~\bibnamefont {Campbell}},\ and\ \bibinfo
		{author} {\bibfnamefont {S.}~\bibnamefont {Brierley}},\ }\bibfield  {title}
	{\bibinfo {title} {{Efficient quantum measurement of Pauli operators in the
				presence of finite sampling error}},\ }\href@noop {} {\bibfield  {journal}
		{\bibinfo  {journal} {Quantum}\ }\textbf {\bibinfo {volume} {5}},\ \bibinfo
		{pages} {385} (\bibinfo {year} {2021})}\BibitemShut {NoStop}%
	\bibitem [{\citenamefont {Yen}\ \emph {et~al.}(2020)\citenamefont {Yen},
		\citenamefont {Verteletskyi},\ and\ \citenamefont
		{Izmaylov}}]{yen2020measuring}%
	\BibitemOpen
	\bibfield  {author} {\bibinfo {author} {\bibfnamefont {T.-C.}\ \bibnamefont
			{Yen}}, \bibinfo {author} {\bibfnamefont {V.}~\bibnamefont {Verteletskyi}},\
		and\ \bibinfo {author} {\bibfnamefont {A.~F.}\ \bibnamefont {Izmaylov}},\
	}\bibfield  {title} {\bibinfo {title} {{Measuring all compatible operators in
				one series of single-qubit measurements using unitary transformations}},\
	}\href@noop {} {\bibfield  {journal} {\bibinfo  {journal} {Journal of
				chemical theory and computation}\ }\textbf {\bibinfo {volume} {16}},\
		\bibinfo {pages} {2400} (\bibinfo {year} {2020})}\BibitemShut {NoStop}%
	\bibitem [{\citenamefont {Jena}\ \emph {et~al.}(2019)\citenamefont {Jena},
		\citenamefont {Genin},\ and\ \citenamefont {Mosca}}]{jena2019pauli}%
	\BibitemOpen
	\bibfield  {author} {\bibinfo {author} {\bibfnamefont {A.}~\bibnamefont
			{Jena}}, \bibinfo {author} {\bibfnamefont {S.}~\bibnamefont {Genin}},\ and\
		\bibinfo {author} {\bibfnamefont {M.}~\bibnamefont {Mosca}},\ }\bibfield
	{title} {\bibinfo {title} {{Pauli partitioning with respect to gate sets}},\
	}\href@noop {} {\bibfield  {journal} {\bibinfo  {journal} {arXiv preprint
				arXiv:1907.07859}\ } (\bibinfo {year} {2019})}\BibitemShut {NoStop}%
	\bibitem [{\citenamefont {Gokhale}\ \emph {et~al.}(2020)\citenamefont
		{Gokhale}, \citenamefont {Angiuli}, \citenamefont {Ding}, \citenamefont
		{Gui}, \citenamefont {Tomesh}, \citenamefont {Suchara}, \citenamefont
		{Martonosi},\ and\ \citenamefont {Chong}}]{gokhale2020n}%
	\BibitemOpen
	\bibfield  {author} {\bibinfo {author} {\bibfnamefont {P.}~\bibnamefont
			{Gokhale}}, \bibinfo {author} {\bibfnamefont {O.}~\bibnamefont {Angiuli}},
		\bibinfo {author} {\bibfnamefont {Y.}~\bibnamefont {Ding}}, \bibinfo {author}
		{\bibfnamefont {K.}~\bibnamefont {Gui}}, \bibinfo {author} {\bibfnamefont
			{T.}~\bibnamefont {Tomesh}}, \bibinfo {author} {\bibfnamefont
			{M.}~\bibnamefont {Suchara}}, \bibinfo {author} {\bibfnamefont
			{M.}~\bibnamefont {Martonosi}},\ and\ \bibinfo {author} {\bibfnamefont
			{F.~T.}\ \bibnamefont {Chong}},\ }\bibfield  {title} {\bibinfo {title}
		{{O($N^3$) Measurement Cost for Variational Quantum Eigensolver on Molecular
				Hamiltonians}},\ }\href@noop {} {\bibfield  {journal} {\bibinfo  {journal}
			{IEEE Transactions on Quantum Engineering}\ }\textbf {\bibinfo {volume}
			{1}},\ \bibinfo {pages} {1} (\bibinfo {year} {2020})}\BibitemShut {NoStop}%
	\bibitem [{\citenamefont {Wan}\ \emph {et~al.}(2022)\citenamefont {Wan},
		\citenamefont {Huggins}, \citenamefont {Lee},\ and\ \citenamefont
		{Babbush}}]{wanMatchgateShadowsFermionic2022}%
	\BibitemOpen
	\bibfield  {author} {\bibinfo {author} {\bibfnamefont {K.}~\bibnamefont
			{Wan}}, \bibinfo {author} {\bibfnamefont {W.~J.}\ \bibnamefont {Huggins}},
		\bibinfo {author} {\bibfnamefont {J.}~\bibnamefont {Lee}},\ and\ \bibinfo
		{author} {\bibfnamefont {R.}~\bibnamefont {Babbush}},\ }\href@noop {}
	{\bibinfo {title} {Matchgate {{Shadows}} for {{Fermionic Quantum
					Simulation}}}} (\bibinfo {year} {2022}),\ \Eprint
	{https://arxiv.org/abs/2207.13723} {arXiv:2207.13723} \BibitemShut {NoStop}%
	\bibitem [{\citenamefont {Zhao}\ \emph {et~al.}(2021)\citenamefont {Zhao},
		\citenamefont {Rubin},\ and\ \citenamefont
		{Miyake}}]{zhaoFermionicPartialTomography2021}%
	\BibitemOpen
	\bibfield  {author} {\bibinfo {author} {\bibfnamefont {A.}~\bibnamefont
			{Zhao}}, \bibinfo {author} {\bibfnamefont {N.~C.}\ \bibnamefont {Rubin}},\
		and\ \bibinfo {author} {\bibfnamefont {A.}~\bibnamefont {Miyake}},\
	}\bibfield  {title} {\bibinfo {title} {Fermionic partial tomography via
			classical shadows},\ }\href {https://doi.org/10.1103/PhysRevLett.127.110504}
	{\bibfield  {journal} {\bibinfo  {journal} {Physical Review Letters}\
		}\textbf {\bibinfo {volume} {127}},\ \bibinfo {pages} {110504} (\bibinfo
		{year} {2021})}\BibitemShut {NoStop}%
	\bibitem [{\citenamefont {Jones}\ and\ \citenamefont
		{Benjamin}(2020)}]{QuESTlink}%
	\BibitemOpen
	\bibfield  {author} {\bibinfo {author} {\bibfnamefont {T.}~\bibnamefont
			{Jones}}\ and\ \bibinfo {author} {\bibfnamefont {S.}~\bibnamefont
			{Benjamin}},\ }\bibfield  {title} {\bibinfo {title} {Questlink—mathematica
			embiggened by a hardware-optimised quantum emulator},\ }\href@noop {}
	{\bibfield  {journal} {\bibinfo  {journal} {Quantum Sci. Techn.}\ }\textbf
		{\bibinfo {volume} {5}},\ \bibinfo {pages} {034012} (\bibinfo {year}
		{2020})}\BibitemShut {NoStop}%
	\bibitem [{\citenamefont {Li}\ and\ \citenamefont {Benjamin}(2017)}]{li2017}%
	\BibitemOpen
	\bibfield  {author} {\bibinfo {author} {\bibfnamefont {Y.}~\bibnamefont
			{Li}}\ and\ \bibinfo {author} {\bibfnamefont {S.~C.}\ \bibnamefont
			{Benjamin}},\ }\bibfield  {title} {\bibinfo {title} {{Efficient Variational
				Quantum Simulator Incorporating Active Error Minimization}},\ }\href
	{https://doi.org/10.1103/PhysRevX.7.021050} {\bibfield  {journal} {\bibinfo
			{journal} {Phys. Rev. X}\ }\textbf {\bibinfo {volume} {7}},\ \bibinfo {pages}
		{021050} (\bibinfo {year} {2017})}\BibitemShut {NoStop}%
	\bibitem [{Note2()}]{Note2}%
	\BibitemOpen
	\bibinfo {note} {Here the product of sets produces a set that contains all
		possible products of the elements}\BibitemShut {NoStop}%
	\bibitem [{\citenamefont {Kyriienko}\ and\ \citenamefont
		{Elfving}(2021)}]{PhysRevA.104.052417}%
	\BibitemOpen
	\bibfield  {author} {\bibinfo {author} {\bibfnamefont {O.}~\bibnamefont
			{Kyriienko}}\ and\ \bibinfo {author} {\bibfnamefont {V.~E.}\ \bibnamefont
			{Elfving}},\ }\bibfield  {title} {\bibinfo {title} {Generalized quantum
			circuit differentiation rules},\ }\href
	{https://doi.org/10.1103/PhysRevA.104.052417} {\bibfield  {journal} {\bibinfo
			{journal} {Phys. Rev. A}\ }\textbf {\bibinfo {volume} {104}},\ \bibinfo
		{pages} {052417} (\bibinfo {year} {2021})}\BibitemShut {NoStop}%
	\bibitem [{\citenamefont {Wierichs}\ \emph {et~al.}(2022)\citenamefont
		{Wierichs}, \citenamefont {Izaac}, \citenamefont {Wang},\ and\ \citenamefont
		{Lin}}]{Wierichs2022generalparameter}%
	\BibitemOpen
	\bibfield  {author} {\bibinfo {author} {\bibfnamefont {D.}~\bibnamefont
			{Wierichs}}, \bibinfo {author} {\bibfnamefont {J.}~\bibnamefont {Izaac}},
		\bibinfo {author} {\bibfnamefont {C.}~\bibnamefont {Wang}},\ and\ \bibinfo
		{author} {\bibfnamefont {C.~Y.-Y.}\ \bibnamefont {Lin}},\ }\bibfield  {title}
	{\bibinfo {title} {General parameter-shift rules for quantum gradients},\
	}\href {https://doi.org/10.22331/q-2022-03-30-677} {\bibfield  {journal}
		{\bibinfo  {journal} {{Quantum}}\ }\textbf {\bibinfo {volume} {6}},\ \bibinfo
		{pages} {677} (\bibinfo {year} {2022})}\BibitemShut {NoStop}%
	\bibitem [{\citenamefont {Jones}\ \emph {et~al.}(2019)\citenamefont {Jones},
		\citenamefont {Brown}, \citenamefont {Bush},\ and\ \citenamefont
		{Benjamin}}]{quest}%
	\BibitemOpen
	\bibfield  {author} {\bibinfo {author} {\bibfnamefont {T.}~\bibnamefont
			{Jones}}, \bibinfo {author} {\bibfnamefont {A.}~\bibnamefont {Brown}},
		\bibinfo {author} {\bibfnamefont {I.}~\bibnamefont {Bush}},\ and\ \bibinfo
		{author} {\bibfnamefont {S.~C.}\ \bibnamefont {Benjamin}},\ }\bibfield
	{title} {\bibinfo {title} {{QuEST and high performance simulation of quantum
				computers}},\ }\href@noop {} {\bibfield  {journal} {\bibinfo  {journal} {Sci.
				Rep.}\ }\textbf {\bibinfo {volume} {9}},\ \bibinfo {pages} {10736} (\bibinfo
		{year} {2019})}\BibitemShut {NoStop}%
	\bibitem [{\citenamefont {Dalzell}\ \emph {et~al.}(2021)\citenamefont
		{Dalzell}, \citenamefont {Hunter-Jones},\ and\ \citenamefont
		{Brand{\~a}o}}]{dalzell2021random}%
	\BibitemOpen
	\bibfield  {author} {\bibinfo {author} {\bibfnamefont {A.~M.}\ \bibnamefont
			{Dalzell}}, \bibinfo {author} {\bibfnamefont {N.}~\bibnamefont
			{Hunter-Jones}},\ and\ \bibinfo {author} {\bibfnamefont {F.~G.}\ \bibnamefont
			{Brand{\~a}o}},\ }\bibfield  {title} {\bibinfo {title} {Random quantum
			circuits transform local noise into global white noise},\ }\href@noop {}
	{\bibfield  {journal} {\bibinfo  {journal} {arXiv preprint arXiv:2111.14907}\
		} (\bibinfo {year} {2021})}\BibitemShut {NoStop}%
\end{thebibliography}

%

\end{document}